\documentclass[11pt]{article}             
\usepackage{osid-tsh}                         %

\usepackage{amsmath,amssymb,amscd,amsthm,mathrsfs}
\usepackage{tikz-cd}
\usepackage{graphicx}
\usepackage{bbm}
\usepackage{verbatim} 

\usepackage[bookmarks=false,pdfstartview={FitH}]{hyperref}

\usepackage{lipsum}                    
\usepackage[most]{tcolorbox}
\usepackage{xcolor}
\definecolor{block-y}{gray}{0.92}
\newtcolorbox{myexample}{colback=block-y,grow to right by=4mm,grow to left by=4mm,
boxrule=0pt,boxsep=0pt,breakable}

\newcommand{\RED}[1]{\textcolor{red}{#1}}
\newcommand{\BLUE}[1]{\textcolor{blue}{#1}}
\newcommand{\MAGE}[1]{\textcolor{magenta}{#1}}

\newcommand{\Adr}{\operatorname{Ad}}
\newcommand{\adr}{\operatorname{ad}}
\newcommand{\ad}{\mathrm{ad}}

\newcommand{\diag}{\operatorname{diag}}

\newcommand{\pos}[1]{\mathfrak{pos}_1(#1)}
\newcommand{\tr}{\operatorname{tr}}

\newcommand{\expt}[1]{\ensuremath{\langle #1 \rangle}{}}

\newcommand{\gks}{\textsc{gksl}}
\newcommand{\cptp}{\mathsf{CPTP}}
\newcommand{\conv}{\operatorname{conv}}

\newcommand{\su}{\mathfrak{su}}

\newcommand{\fk}{\mathfrak{k}}

\newcommand{\fw}{\mathfrak{w}}

\newcommand{\bK}{\mathbf{K}}

\newcommand{\C}{\mathbb{C}}
\newcommand{\R}{\mathbb{R}}

\newcommand{\GL}{\mathsf{GL}}
\newcommand{\SU}{\mathrm{SU}}

\newcommand{\reach}{\mathfrak{reach}}
\newcommand{\stab}{\mathfrak{stab}}
\newcommand{\derv}{\mathfrak{derv}}

\newcommand{\hGA}{\mathbf{\Gamma}}

\newcounter{app}
\renewcommand*{\theapp}{\Alph{app}}

\newcommand*\app[1]{%
\refstepcounter{app}\label{app_#1}\hypertarget{app:#1}{\theapp}}
\newcommand*\appref[1]{\hyperlink{app:#1}{\ref*{app_#1}}}

\theoremstyle{plain}
\newtheorem{lemma}{Lemma}
\newtheorem{thm}{Theorem}
\newtheorem{conjecture}{Conjecture}
\newtheorem{corollary}{Corollary}
\newtheorem{proposition}{Proposition}

\theoremstyle{definition}

\theoremstyle{remark}
\newtheorem{remark}{Remark}

\title{Exploring the Limits of Controlled Markovian Quantum\\ Dynamics with Thermal Resources} 
\author{Frederik vom Ende
	\\[1mm]{\footnotesize\it Technische Universit{\"a}t M{\"u}nchen, 
	School of Natural Sciences, 85747 Garching, Germany  \mbox{just moved to:} 
	Dahlem Center for Complex Quantum 
Systems, Freie Universit{\"a}t Berlin, Arnimallee 14, 14195 Berlin, Germany
 \& {frederik.vom.ende@fu-berlin.de}}\\[2ex]
Emanuel Malvetti
	\\[1mm]{\footnotesize\it Technische Universit{\"a}t M{\"u}nchen, 
	School of Natural Sciences, 85747 Garching, Germany \& {emanuel.malvetti@tum.de}}\\[2ex]	
        Gunther Dirr
        \\[1mm]{\footnotesize\it Universit{\"a}t W{\"u}rzburg, 
        Institut f{\"u}r Mathematik, Emil-Fischer-Stra{\ss}e 40,\\
	97074 W{\"u}rzburg, Germany \& dirr@mathematik.uni-wuerzburg.de}\\[2ex] 
Thomas Schulte-Herbr{\"u}ggen\thanks{ \; The project was funded i.a.~by the Excellence Network of Bavaria under ExQM and by 
{\em Munich Quantum Valley} of the Bavarian State Government with funds from Hightech Agenda {\em Bayern Plus}\,;
after moving, F.v.E. has been supported by the Einstein Foundation (Einstein Research Unit on quantum devices) and the 
{\sc math+} 
Cluster of Excellence.} 
	\\[1mm]{\footnotesize\it Technische Universit{\"a}t M{\"u}nchen, 
	School of Natural Sciences, 85747 Garching and\\
   Munich Centre for Quantum Science and Technology (MCQST),  Schellingstra{\ss}e~4,\\ 80799~M{\"u}nchen, Germany \& {tosh@tum.de}}
}

\AtBeginDocument{
  \label{CorrectFirstPageLabel}
  
}

\begin{document}
\maketitle
\vspace{-4mm}
\begin{abstract}
Our aim is twofold:
First, we rigorously analyse the generators of quantum-dynamical semigroups of thermodynamic processes.
We
characterise a wide class of \textsc{gksl}-generators for quantum 
maps within
thermal operations and argue that every infinitesimal generator of (a one-parameter semigroup of) Markovian thermal operations belongs to this class.
We completely classify and visualise
them
and their non-Markovian counterparts 
for the case of a single qubit.

Second, we use this description in the framework of bilinear control systems to characterise reachable sets
of coherently controllable quantum 
systems with switchable coupling to a thermal bath.
The core problem reduces to studying a hybrid control system (``toy model'') on the standard simplex 
allowing for two types of evolution:
(i) instantaneous permutations and (ii) 
a one-parameter semigroup of $d$-stochastic maps. 
We generalise upper bounds of the reachable set of this toy model 
invoking new results on thermomajorisation.
Using 
tools of control theory
we fully characterise these reachable sets as well as the set of stabilisable states 
as exemplified by exact results in qutrit systems.
\end{abstract}



\bigskip
\begin{center}
\small{\em Dedicated to the memory of Andrzej Kossakowski.} 
\end{center}
\vspace{-2mm}
\noindent
{\em%
It is a pleasure to contribute to honour Andrzej Kossakowski as key figure in laying and triggering the now classical groundwork (around \cite{Koss72,Koss72b,GKS76,GK76,GFKVS78}, see \cite{ChruPas17})
that completely characterises the infinitesimal generators of (Markovian) quantum maps in finite dimensions.
His towering work continues to be a well of inspiration. With the recent focus in quantum dynamics now taking
``the bath" as 
quantum  
thermodynamical resource, here we try to carry on somewhat in
his spirit, in particular when characterising the generators of
Markovian 
thermal quantum operations\footnote{see Thm.~\ref{thm:markov_generator} in Sec.~\ref{sec_thermo_markov}}
in a Lie-semigroup frame.}

\section{Introduction}
Linking the well established field of quantum control \cite{dAlessandroBook2022,DiHeGAMM08,Roadmap2015,Koch22} with the emerging subject of quantum thermodynamics \cite{QThermo2018,Lostaglio19r}
is quite a novel and important line of research. 
Thus in this article we focus on framing and 
studying {\em quantum-dynamical control problems 
with coherent controls and thermal resources as additional controls}\footnote{
i.e.~with non-unitary controls chosen from the set of thermal operations \cite{Janzing2000,Horodecki13,Brandao13}
}
\cite{BSH16,CDC19},
where the time evolution of controlled
systems is taken to be defined by
a controlled \gks-equation.
We build on recent progress in interfacing  \mbox{(non-)Markovian}\footnote{In accordance with \cite{Wolf08a} the (time-dependent) Markovian
quantum maps are taken as those which are infinitesimal $\cptp$-divisible.\label{footnote_timedep_markov}
} processes with quantum thermodynamics in general 
\cite{Szczygielski13,Bylicka16,Abiuso19,Dann22,Spaventa22, Chakraborty22,Colla22}, 
and with its resource-theory approach
in particular \cite{Gour19,Bhatta20,LosKor22a,Ptas22}.

To this end, (Lie-)semigroup techniques \cite{HHL89,LNM1552,HofRupp97div,Lawson99} lend themselves as 
a common frame
naturally extending to concepts of (i) classical majorisation and of (ii) thermomajorisation \cite{MarshallOlkin}
as well as (iii) Markovianity of quantum maps
\cite{DHKS08,OSID17}.
Moreover, (iv) the set of reachable states related to a given initial state
of a
(Markovian)
quantum control system takes
the form of a Lie-semigroup orbit \cite{DHKS08,OSID17}.

Studying reachable sets of such control systems is paramount, e.g., to ensure
well-posedness of many 
(optimal) control tasks.
The main question is whether a desired target state can be prepared 
given an equation of motion (plus some control variables) and an initial condition, and how to characterise feasible state transfers in general.
Interestingly, the core problem of the resource approach to quantum thermodynamics 
(as initially sparked by Brand\~ao et al.~\cite{Brandao15}, Horodecki \& Oppenheim 
 \cite{Horodecki13}, as well as Renes \cite{Renes14}, and further pursued in
 \cite{Faist17,Gour15,Lostaglio18,Sagawa19,Mazurek19,Alhambra19,
 vomEnde22thermal}) is similar---namely:
 given a fixed background temperature as well as initial and target states 
of a quantum system,
 can the former be mapped to the 
 latter in a ``thermodynamically consistent'' \cite{Kosloff13} manner?
Here the admissible quantum maps 
are 
``thermal operations''
which form the fundamental building block of the resource theory approach to quantum thermodynamics. 
Roughly speaking, they comprise 
operations 
(assumed to be) performable in arbitrary number without cost.
A precise definition is given in Sec.~\ref{sec_thermo_markov}, and for a comprehensive introduction to the general topic 
see, e.g., the review by Lostaglio \cite{Lostaglio19r}.

As teasers for the power of combining control theory with ``thermal resources'' 
in the sense of allowing for
(non-)Markovian processes consider the following two (known) examples:
\begin{itemize}
\item[(1)]
Take any closed quantum system with non-trivial Hamiltonian which can be 
fully unitarily controlled. Then there always exists a {\em non-Markovian thermal operation} 
(known as $\beta$-{\sc swap}~\cite{Lostaglio18}) such 
that adding it to the setup as an additional control allows to generate the ground 
state up to arbitrary precision \cite{Alhambra19}, \cite[Prop.~4.12]{vomEnde20Dmaj}
{(and thus \textit{every} state: use the Schur-Horn theorem to generate the eigenvalues of the target state on the diagonal, followed by a full dephasing via the $\beta$-swap\footnote{
This only holds unconditionally if $e^{-1/T}\geq\frac12$.
For temperatures lower than that the $\beta$-swap has to be implemented as a two-step dephasing thermal{is}ation (i.e.~one needs to be able to implement the dephasing independent of the diagonal action of the $\beta$-swap) in order to get from the ground state to a ball around the maximally mixed state. 
See also the Worked Example at the end of Sec.~\ref{sec_thermo_markov}}).
}
\item[(2)] 
$n$-qubit systems with full coherent control plus switchable amplitude damping 
(coupling to a bath of temperature  $T=0$)
for one of the qubits act (up to closure) transitively on the set of all density operators \cite{BSH16}.
The result generalises from qu{\em b}\/its to arbitrary qu{\em d}\/its and
$m$-level systems \cite{CDC19}.
\end{itemize}
Hence in these two instances reachability
is settled. However, for controlled {\em Markovian} dynamics with thermal resources 
(e.g., coupling to a bath of temperature $0<T<\infty$) the reachability
problem is open and subject of this work.
%
%
Here, a fundamental property of thermal processes will be central: Under thermal
operations, diagonal density matrices (in the eigenbasis of the system Hamiltonian)
remain diagonal, meaning diagonal elements (populations) evolve independently from off-diagonal ones (coherences), refer also to Cor.~\ref{cor:thermo-diagonal}.
%
This greatly simplifying property lends itself for studying
the restricted control system on diagonal states, later called ``toy model''.
General results and numerical illustrations, as well as analytic results in a low-dimensional setting will be presented in Secs.~\ref{sec_toy_model_sub1}~and \ref{sec:Qutrit-Results}, respectively.

\bigskip
\noindent
{\bf Structure and Main Results.} Sec.~\ref{sec:MarkovControl}~sets the stage for discussing the dynamics of
open Markovian systems as incarnations of {\em bilinear control systems} with two types of controls:
coherent Hamiltonian ones as well as incoherent dissipative ones such as switchable thermal
operations as brought about by coupling the quantum system to a thermal bath. 

\medskip

Sec.~\ref{sec:semigroups-maj-markov}~paves the semigroup background for introducing general
concepts of majorisation (such as $d$-majorisation, a special case of thermomajorisation) needed in the context 
of describing reachability under thermal operations. Within this framework, {\em Markovian} quantum maps  
come with the particular structure of {\em Lie} semigroups \cite{DHKS08,OSID17}, which
allows for putting majorisation and Markovianity on a common ground. 

In particular, we give the respective Lie wedges to certain sets of quantum maps such as Gibbs-preserving ones ($\mathsf{Gibbs}$) as well as thermal ($\overline{\mathsf{TO}}$) and enhanced thermal operations ($\mathsf{EnTO}$);
the corresponding generated Lie semigroups thus define the respective Markovian counterparts
$\mathsf{MGibbs}$, $\mathsf{MTO}$, and $\mathsf{MEnTO}$.
By giving set-inclusions, in 
Sec.~\ref{sec_thermo_markov}~we interrelate them all. 
As another important result (Thm.~\ref{thm:markov_generator}) one gets an explicit construction for 
(possibly all)
generators of Markovian thermal
operations via temperature-weighted projections out of a total Hamiltonian 
(preserving energy of system and bath) in the Stinespring dilation.  
With regard to reachable sets
of diagonal states under such controlled Markovian dynamics, in Sec.~\ref{sec:d_maj}~we recall the \mbox{$d$-majorisation} polytope
and its properties.

\medskip

Sec.~\ref{sec:toy_model}
illustrates the theory of bilinear open quantum
control systems, where the incoherent controls are brought about by Markovian thermal operations.
Reachable sets for diagonal states under such operations are characterised by extreme points of the $d$-majorisation
polytope. More precisely, in
$n$-level systems, the 
enclosure of the reachable set by a convex set is given by the convex hull
over the permuted extreme points of the \mbox{$d$-majorisation} polytope (Thm.~\ref{thm:general_dmaj_bound}). 
Sec.~\ref{sec:Qutrit-Results} discusses three-level systems in detail. By reformulating the control system as a differential inclusion, we find an analytic
expression for the set of stabilisable states. Moreover we compute the set of reachable states for arbitrary initial states, and we deduce the structure of control sets and their reachability order.

\section{Control Setting of Markovian Quantum Dynamics}\label{sec:MarkovControl}

To fix notations, we write $\pos{n}$ for the convex set of all $n \times n$ density matrices (i.e.~all positive semi-definite
$n \times n$-matrices of trace one),
$\mathcal L(\mathbb C^{n\times n})$ for the space of all linear operators acting on complex
$n\times n$-matrices, and $\cptp(n)$ for the convex subset of $\mathcal L(\mathbb C^{n\times n})$ which consists of all completely positive and trace-preserving maps, also known as quantum channels, quantum maps, or Kraus maps.
We prefer to use the term ``quantum maps''.

\medskip
\noindent
Consider the usual Markovian master equation
\begin{equation}\label{eq:diss_evolution}
\dot{\rho}(t)=-i[H_0,\rho(t)]-\hGA(\rho(t))\,,\quad \rho(0)=\rho_0\in\pos{n}\,
\end{equation}
with $\hGA\in\mathcal L(\mathbb C^{n\times n})$ 
of \gks-form \cite{GKS76,Lindblad76}, i.e.
\begin{equation}\label{eq:lindblad_V}
\hGA(\rho):=\sum_k\Big( \tfrac12 \big(V_k^\dagger V_k \rho+\rho V_k^\dagger V_k\big)-V_k\rho V_k^\dagger \Big)
\end{equation}
and $V_k\in\mathbb C^{n\times n}$ arbitrary 
to ensure the evolution $\rho(t)=e^{-t(i\operatorname{ad}_{H_0}+\hGA)}\rho_0$
solving \eqref{eq:diss_evolution}
remains in $\pos{n}$ for all $t\in\mathbb R_+:=[0,\infty)$.
Recall $(e^{-t(i\operatorname{ad}_{H_0}+\hGA)})_{t\in\mathbb R_+}$ is
$\cptp$, hence a 
(trace norm-)contraction semigroup
\cite{PG06}
leaving $\pos{n}$ invariant. 

In this work, an overarching goal is to character{is}e control systems $\Sigma$ extending Eq.~\eqref{eq:diss_evolution} 
by coherent controls (generated by Hermitian $H_j$ and (piece-wise constant) $u_j(t)\in\mathbb R$) 
and by making dissipation bang-bang switchable in the sense 
\begin{equation}\label{eq:control-diss_evolution}
\dot{\rho}(t)= -{i}\Big[H_0 + \sum_{j=1}^m u_j(t) H_j,\rho(t)\Big] - \gamma(t) \hGA(\rho(t))
\end{equation}
with $\gamma(t) \in \{0,1\}$.
This setting is a typical incarnation of the wide class of {\em bilinear control systems} \cite{Jurdjevic97,Elliott09} 
\begin{equation}\label{eq:bilin}
\dot{\mathbf{x}}(t) = -\Big(A + \sum_j u_j(t) B_j\Big) \mathbf{x}(t)\,,\quad \mathbf{x}(0) = \mathbf{x}_0\,,
\end{equation}%
where $A$ denotes an uncontrolled drift, while the control terms\footnote{The bilinearity of the control terms w.r.t\ $u$ and $x$ 
entails the terminology of Eq.~\eqref{eq:bilin}.} consist of (piecewise
constant) control amplitudes $u_j(t)\in\mathbb R$ and linear control operators $B_j$. Here the
state $\mathbf{x}(t)$ should be thought of as density operator. 

A paramount notion in such control systems is the \textit{reachable set} of ${\bf x}_0$ at time $\tau\geq 0$, denoted $\mathfrak{reach}(\mathbf{x}_0,\tau)$. It is defined as the collection of all ${\bf x}(\tau)$, where $t\mapsto{\bf x}(t)$ 
is any solution of~\eqref{eq:bilin}.
Likewise, the reachable set until time $\tau$ is defined as $\mathfrak{reach}_{[0,\tau]}({\bf x}_0):= \bigcup_{\tau'\in[0,\tau]}\mathfrak{reach}(\mathbf{x}_0,\tau')$,
and the overall reachable set as $\mathfrak{reach}(\mathbf{x}_0):=
\bigcup_{\tau'\geq 0}\mathfrak{reach}(\mathbf{x}_0,\tau')$. 
The latter is mostly used in this work, where analogously, we write $\mathfrak{reach}(\mathbf{\rho}_0)$ for
  the entire reachable set of Eq.~\eqref{eq:control-diss_evolution}.
  The system
Lie algebra of~\eqref{eq:bilin}, which provides the crucial tool for analyzing controllability, accessibility, and reachability questions, reads 
$\fk:=\expt{A, B_j\,|\, j=0,\ldots,m\,}_{\sf Lie}$. 
  
In particular for closed quantum systems, i.e.~systems which do not interact with their environment ($\hGA$=$0$=$B_0$)
one would choose $A$ as ${i}\,\ad_{H_0}$ and $B_j$ as ${i}\,\ad_{H_j}$, $j = 1, \dots, m$ 
to recover the
  right-hand side of \eqref{eq:control-diss_evolution}.
   Then it is known \cite{JS72,Bro72,Jurdjevic97,DiHeGAMM08} that 
$\reach(\mathbf{x}_0)$ 
is given by the orbit of the
  initial state under the action of the dynamical systems group $\bK:=\expt{\exp \fk}$, provided $\bK$ is a closed and thus compact
  subgroup\footnote{According to the above definition, in closed systems the system Lie algebra $\fk$ is a Lie subalgebra of the adjoint representation of
    the special unitary Lie algebra, i.e.~$\fk\subseteq\adr_{\su(n)}$. 
    Yet the commutator identity $[\ad_{H},\ad_{H'}] = \ad_{[H,H']}$ allows for identifying $\fk$ 
    with the Lie subalgebra
    of $\su(n)$ generated by ${i} H_0, \dots, {i}H_m$.
In generic open systems, however, the \gks-term $\hGA$ precludes any similar a-priori simplification of the
      system Lie algebra $\fk$.
    }
    of ${\rm Ad}_{\SU(n)} \cong \{\bar{U} \otimes U : U \in \SU(n)\} \subseteq \SU(n^2)$.
     More generally, for open systems undergoing Markovian dissipation which are driven by coherent controls
  $B_j= {i}\,\ad_{H_j}$ for $j=1,\ldots,m$ one sets the operator $B_{0} = \hGA$ to include the environmental interaction.
  This operator takes the form of a \gks-dissipator as given in Eq.~\eqref{eq:lindblad_V}. If $B_0$ is  {\em bang-bang switchable} it acts
  as additional control, and if it is uncontrolled it contributes to the drift term $A = {i}\,\ad_{H_0}$. Motivated by recent experimental progress
  \cite{Mart09,Mart13,Mart14,McDermott_TunDissip_2019} as described in~\cite{BSH16} here we address the first case
  and refer to it as {\bf Scenario BB} in the sequel.\medskip

In general, a concise description of reachable sets of Eq.~\eqref{eq:control-diss_evolution} is challenging in particular in higher-dimensional cases.
Although it is known that it always takes the form of a (Lie-)semigroup orbit, see,
  e.g.,~\cite{DHKS08}
this is usually not enough to obtain explicit characterisations.
Currently there are only a few scenarios for which reachability is settled:
(a) 
In the unital case $\hGA(\mathbbm1_n) = 0$, one has \cite{Ando89,Yuan10}
\begin{equation}\label{eq:reach-unital}
\mathfrak{reach}_\Sigma(\rho_0) \subseteq \{\rho \in \mathbb C^{n\times n} \,|\, \rho \prec \rho_0 \}\,.
\end{equation}
(b)
If in addition $\hGA$ is generated by a single normal $V$,
one gets (up to closure) equality in \eqref{eq:reach-unital} provided
the unitary part of Eq.~\eqref{eq:control-diss_evolution} is {\em controllable}
and the switching function $\gamma(t)$ gives extra control 
(cf.~\cite{BSH16}, \cite[Prop.~5.2.1]{vE_PhD_2020} for finite and \cite{OSID19} for infinite dimensions).
%

Under the scenario {\bf BB} above 
plus the invariance of diagonal
states imposed by thermal processes (see
\cite{Lostaglio19r} and Cor.~\ref{cor:thermo-diagonal} below)
as, e.g., implemented by the \gks-generators $V_1$ and $V_2$ of Eqs.~\eqref{eq:sigma+} and \eqref{eq:sigma-},
the closure of the unitary orbit of $\diag\big(\mathfrak{reach}_\Lambda(x_0)\big)$
sits in the closure of 
$\mathfrak{reach}_\Sigma(U\diag(x_0)U^\dagger)$.
Here $\Lambda$ denotes a simplified version of $\Sigma$---later called ``toy model''---which will be introduced in Sec.~\ref{sec_toy_model_sub1}
Settings beyond thermal relaxation 
are pursued with similar
techniques, e.g.,~by \cite{rooney2018} at the expense of arriving at conditions that are hard to verify
for higher-dimensional systems. 


\section{Semigroups, Major{is}ation, and Markovian Quantum Dynamics}\label{sec:semigroups-maj-markov}

In this 
section, we introduce 
background and terminology 
for
unifying several existing and seemingly different notions in the literature.

\paragraph*{Major{is}ation via Semigroups.} First, let us have a closer look at ``the'' concept of major{is}ation
from the point-of-view of semigroups. Let $\mathcal Z$ be a real
or complex finite-dimensional\footnote{
This is just to avoid further technicalities; in principle this approach works in arbitrary dimensions.
}
vector space and let $\mathcal C \subseteq \mathcal Z$ be a closed and convex subset. Moreover, let
$\mathcal B(\mathcal Z)$ denote the set of all linear operators acting on $\mathcal Z$ and let $S(\mathcal C)$ be the set of
operators in $\mathcal B(\mathcal Z)$ which leave $\mathcal C$ invariant. Obviously, $S(\mathcal C)$ is a subsemigroup
of $\mathcal B(\mathcal Z)$, i.e.~$S(\mathcal C)$ is closed under multiplication and contains the identity.
Finally, let $S_0 \subseteq S(\mathcal C)$ be any subsemigroup of $S(\mathcal C)$. 
Then $x \in \mathcal Z$ is said to be \mbox{\textit{$S_0$-major{is}ed}} by $y \in \mathcal Z$ (denoted by $x \prec_{S_0} y$)
if there exists a transformation $A \in S_0$ such that $Ay = x$ holds. 
This concept was first introduced by Parker and Ram \cite{Parker96} and is 
known as semigroup major{is}ation \cite[Ch.~14.C]{MarshallOlkin}.
A well-studied example of $S_0$-major{is}ation
is 
classical vector-major{is}ation on $\R^n$, where
$\mathcal C :=\{x\in \R_+^n \,:\,\sum_{j=1}^n x_j=1\}\subset \R^n$ is
the standard simplex and $S_0$ is chosen as the set
of all real doubly-stochastic matrices of size $n \times n$ \cite[Ch.~1 \& 2]{MarshallOlkin}. Further physically relevant examples are discussed in the following sections; here we summar{is}e some basic properties which result from the above definition.
%
\begin{proposition}
  \label{prop:S-majorization}
  Let $S(\mathcal C)$ be the subsemigroup of $\mathcal B(\mathcal Z)$ which leaves $\mathcal C$ invariant and
  let $S_0$ be any subsemigroup of $S(\mathcal C)$.
  Then the following hold:
  \begin{enumerate}
  \item[(i)]
    $S(\mathcal C)$ is (topologically) closed and convex, and thus simply connected.
  \item[(ii)]
   If $\mathcal C$ is compact and the convex cone $\mathcal C_0 := \mathbb R_+ \cdot \mathcal C$ generated by $\mathcal C$
    has non-empty interior,
    then $S(\mathcal C)$ is compact.
  \item[(iii)]
    Given any $y\in\mathcal Z$, if $S_0$ is convex or compact,
    then so is the set $M_{S_0}(x) := \{y \in \mathcal C \,:\, y \prec_{S_0}x\} = S_0 \cdot x$.
  \end{enumerate}
\end{proposition}
\begin{proof}
  (i) and (iii) are obvious so we only prove (ii). If $\mathcal C$ is compact, then so is $\mathcal C' := [0,1] \cdot \mathcal C$;
  in particular $\mathcal C'\subseteq B_K(0)$ for some $K>0$. Moreover, $\mathcal C'$ is invariant under $S(\mathcal C)$ by
  linearity. Now, by assumption on $\mathcal C_0$ there exist $x_0\in\mathcal C'$ and $r>0$ such that $B_r(x_0)\subseteq\mathcal C'$. Thus---given
  any $A \in S(\mathcal C)$---for all $\Delta \in B_r(0)$ we obtain the estimate
  \begin{equation*}
    \|A\Delta\| = \|A(\Delta+x_0) - Ax_0\| \leq \|A(\Delta+x_0)\| + \|Ax_0\| \leq 2 K\,,
  \end{equation*}
  where we used the invariance of $\mathcal C'$ and its boundedness.
This shows that $S(\mathcal C)$ is bounded by $2Kr$ (with respect to the operator norm).
Together with closedness from (i) we conclude that $S(\mathcal C)$ is compact.
\end{proof}
\noindent The relation $\prec_{S_0}$ induced by $S_0$ 
 is always reflexive and transitive and thus a preorder on $\mathcal C$.
In general, however, $\prec_{S_0}$ fails to be anti-symmetric and hence a partial order, e.g., if $S_0$ contains
a subgroup which acts non-trivially on $\mathcal C$. This
can always be resolved by a suitable
equivalence relation\footnote{$x \sim y$ $:\Longleftrightarrow$ $x \prec_{S_0} y$ and $y \prec_{S_0} x$} such that
the corresponding quotient space provides a natural domain on which $\prec_{S_0}$ becomes an order. In the above
example of classical vector-major{is}ation such a ``natural domain'' can be identified with a fixed ``Weyl
chamber''.
%
%

\paragraph*{Lie Theory for Semigroups, e.g., Quantum Maps.} Next, we elucidate Markovianity 
of quantum maps
from the perspective
of Lie-semigroup theory. To this end let $S$ be any closed (sub-)semigroup of $\mathcal B(\mathcal Z)$ or $\GL(\mathcal Z)$,
where $\GL(\mathcal Z)$ denotes the set of all invertible elements of $\mathcal B(\mathcal Z)$. For what follows, recall that
$\mathcal B(\mathcal Z)$ is the Lie algebra of $\GL(\mathcal Z)$. Then one defines the concept of a
\emph{Lie wedge of $S$} via
\vspace{-2mm}
\begin{equation}\label{eq:liewedge_def}
  \mathsf{L}(S) := \{A \in \mathcal B(\mathcal Z) \,:\, e^{tA} \in S \text{ for all } t\geq 0\}\,.
  \vspace{-2mm}
\end{equation}
Obviously this construction general{is}es the notion of a Lie (sub-)algebra for Lie
(sub-)groups\footnote{For simplicity, we drop the prefix ``sub'' whenever the corresponding superset is obvious.} 
of $\GL(\mathcal Z)$.  The natural 
example
of a {\em Lie wedge} in relation to quantum control is the set of all \gks{}-generators \cite{GKS76,Lindblad76} 
which constitutes the
Lie wedge to the semigroup $\cptp(n)$ of {\em all} quantum maps \cite{DHKS08} (of dimension $n$).
Here, $\mathcal Z$ can be taken to be ${i}\mathfrak u(n)$,
the set of all Hermitian $n \times n$-matrices, and $\mathcal C$ to be $\pos n$.
Then $\cptp(n)$ is a compact convex subsemigroup
of $S(\pos n) \subset \mathcal B({i}\mathfrak u(n))$.

The following proposition and corollary summar{is}e elementary
properties of the above construction. In particular, part (iv) below justifies calling $\mathsf{L}(S)$ the ``tangent cone''
of $S$ at the identity.

\begin{proposition}
  \label{prop:Lie-wedge}
  Let $S$ be a closed semigroup in $\mathcal B(\mathcal Z)$ or $\GL(\mathcal Z)$ and let $\mathsf{L}(S)$ be its Lie wedge.
  Then the following properties hold:
  \begin{enumerate}
  \item[(i)]
    $\mathsf{L}(S)$ is a closed convex cone of $\mathcal B(\mathcal Z)$.
  \item[(ii)]
    $\mathsf{L}(S)$ is invariant under conjugation by arbitrary edge elements, that is,
    $e^A \mathsf{L}(S) e^{-A} = \mathsf{L}(S)$ for all
   $A \in \mathsf{E}(\mathsf{L}(S))$, where the edge $\mathsf{E}(\mathsf{L}(S))$ is defined as largest subspace contained in $\mathsf{L}(S)$.
  \item[(iii)]
    $\mathsf{E}(\mathsf{L}(S))$ is a Lie subalgebra of $\mathcal B(\mathcal Z)$. More precisely, it is the Lie algebra
      of $\mathsf{E}(S)$, where $\mathsf{E}(S)$ denotes the largest subgroup\footnote{Note that $\mathsf{E}(S)$ is also called
          edge---yet edge of $S$ instead of $\mathsf{L}(S)$.} of $S$.
  \item[(iv)]
    The operator $A \in \mathcal B(\mathcal Z)$ belongs to $\mathsf{L}(S)$ if and only if there exists a $C^1$-curve
    $\gamma:[0,\varepsilon]\to S$ for some $\varepsilon>0$ with $\gamma(0) = \operatorname{id}$, $\dot{\gamma}(0) = A$, i.e.
    \begin{equation}
      \label{eq:tangent-cone}
    \mathsf{L}(S) = \{\dot{\gamma}(0) \,:\, \gamma \in C^1([0,\varepsilon],S)
\text{ 
and } \gamma(0) = \operatorname{id}\}\,.
    \end{equation}
  \end{enumerate}
\end{proposition}
\noindent For the statements (i), (ii), and (iii) we refer to \cite[Prop.~1.14]{LNM1552} and \cite[Thm.~4.4]{Lawson99}.
Statement (iv) can be found in \cite[Def.~4.2]{Lawson99} (without proof) and in \cite[Prop.~V.1.7]{HHL89}
(yet under a much more general setting). 
For completeness, a sketch of a 
proof is given in App.~\appref{A}.
Due to Prop.~\ref{prop:Lie-wedge} any subset of $\mathcal B(\mathcal Z)$ featuring property (i) and (ii)
(and thus also (iii)) will be called an (abstract) \emph{Lie wedge}.

Later we will encounter a scenario where the first derivative of a certain family of curves is not only in the Lie wedge but
rather in its edge. In this case one can extract further information from higher derivatives:
\begin{corollary}
  \label{cor:Lie-wedge}
  Let $\varepsilon>0$, 
 a closed, convex semigroup $S\subseteq \mathcal B(\mathcal Z)$, and a $C^2$-curve $\gamma:[0,\varepsilon]\to S$
   with $\gamma(0) = \operatorname{id}$, $\dot{\gamma}(0) \in \mathsf{E}(\mathsf{L}(S))$ be given. Then $\ddot{\gamma}(0) \in \mathsf{L}(S)$.
\end{corollary}

\begin{proof}
  Let us assume, by way of contradiction, that $A_0 := \dot{\gamma}(0) \in \mathsf{E}(\mathsf{L}(S))$ and $B_0 := \ddot{\gamma}(0) \not\in \mathsf{L}(S)$.
  Since $\mathsf{L}(S)$ is closed and convex there exists a separating linear functional $\alpha:\mathcal B(\mathcal Z) \to \R$,
  i.e.~$\alpha(B_0) < 0$ and $\alpha(A) \geq 0$ for all $A \in \mathsf{L}(S)$ \cite[Thm.~3.4]{Rudin91}. Because $A_0\in \mathsf{E}(\mathsf{L}(S))$ we know $\R A_0 \in \mathsf{L}(S)$
  which forces $\alpha(A_0) = 0$. Moreover, convexity of $S$ together with Prop.~\ref{prop:Lie-wedge} (iv) implies
    $S -{\operatorname{id}} \subseteq \mathsf{L}(S)$: 
to see this, simply consider the curve $t\mapsto tA+(1-t){\operatorname{id}}\in S$ where $A\in S$ is arbitrary.
  With this we obtain the estimate
  \begin{equation*}
    \begin{split}
      \alpha({\operatorname{id}}) & \leq  \alpha({\operatorname{id}}) + \alpha(\gamma(t) - {\operatorname{id}}) = \alpha\big(\gamma(t)\big)
      = \alpha\Big(\gamma(0) + t A_0 + \frac{B_0}{2}t^2 + { \scriptstyle \mathcal{O}}(t^2)\Big)\\
      & = \alpha({\operatorname{id}}) + \frac{\alpha(B_0)}{2}t^2 + { \scriptstyle \mathcal{O}}(t^2)
    \end{split}
  \end{equation*}
  leading to the contradiction $\alpha(B_0) + { \scriptstyle \mathcal{O}}(1) \geq 0$. Thus
  $B_0 \in \mathsf{L}(S)$ which concludes the proof.
\end{proof}

Now, trivial examples of closed (and path-connected) semigroups of $\mathcal B(\mathcal Z)$ reveal that, in
general, semigroups cannot be recovered via the exponential map from their Lie wedge---in contrast to path-connected
subgroups of $\mathcal B(\mathcal Z)$ which are fully character{is}ed by their Lie subalgebras, cf.~\cite{Yamabe1950}
or \cite[Thm.~9.6.1]{HN12}. Therefore, the above concept of a Lie wedge naturally suggests the notion of a
\emph{Lie semigroup} for those closed semigroups $S$ which can be reconstructed from their Lie wedge in the following sense:
\begin{equation}
  \label{eq:Lie-semigroup}
  S = \overline{\langle \exp\big(\mathsf{L}(S)\big) \rangle}_{\rm SG}\,.
\end{equation}
Here, the overbar denotes the topological closure\footnote{\label{f-note} In Eq.~\eqref{eq:Lie-semigroup}, one can consider either the
  closure with respect to $\mathcal B(\mathcal Z)$ or with respect to $\GL(\mathcal Z)$ or an even finer
  topology. The latter case is only of interest if $S$ is not assumed to be closed in $\mathcal B(\mathcal Z)$ or in
  $\GL(\mathcal Z)$, cf.~\cite{HHL89}.} and $\langle \exp\big(\mathsf{L}(S)\big)\rangle_{\rm SG}$ is the semigroup generated by the
set $\exp\big(\mathsf{L}(S)\big)$. Notably, the operation $\langle \;\cdot\;\rangle_{\rm SG}$ cannot be avoided as
$\exp\big(\mathsf{L}(S)\big)$ is in general not closed under multiplication.

Conversely, given an (abstract) Lie wedge $\mathfrak{w}$ in $\mathcal B(\mathcal Z)$ 
does there exist a Lie
semigroup\footnote{Strictly speaking, one has to relax the concept of a Lie semigroup slightly by admitting subgroups
    which are not necessarily closed with respect to $\GL(\mathcal Z)$, cf.~\cite{HHL89} Def.~V.1.11 and footnote~\ref{f-note} .
}
$S \subseteq\GL(\mathcal Z)$
such that $\mathfrak{w} = \mathsf{L}(S)$? Again, contrary to Lie algebras, the answer is in general no. Therefore, Lie wedges which
do allow for such a representation are called \emph{global} (in $\GL(\mathcal Z)$) in the following theorem.
We collect some key results on globality, which
will be important for constructing Markovian counterparts to  (enhanced) thermal operations and Gibbs preserving maps
in Eqs.~\eqref{eqn:MSemigroups}, see also Rem.~\ref{rem:global-wedge-Markov}.
For technical details and more sophisticated
character{is}ations see \cite[Ch.~V \& VI]{HHL89}, \cite[Ch.~1]{LNM1552},
and \cite[Prop.~6.2]{Lawson99}.

\begin{thm}\label{thm_lie_global}
  Given any Lie wedge $\mathfrak{w} \subseteq\mathcal B(\mathcal Z)$ the following statements hold.
  \begin{enumerate}
  \item[(a)]
    If there exists a closed semigroup $S$ of $\mathcal B(\mathcal Z)$ or of $\GL(\mathcal Z)$ such that $\mathsf{L}(S) = \mathfrak{w}$,
    then there exists a unique closed Lie semigroup $S_0 \subseteq S$ such that $\mathsf{L}(S_0) = \mathfrak{w}$, i.e.~$\mathfrak{w}$
    is global. In particular, 
    $ S_0 = \overline{\langle \exp\big(\mathsf{L}(S)\big) \rangle}_{\rm SG}\,. $
    Moreover, $S_0$ can be character{is}ed\\[-7mm]
    \begin{enumerate}
      \item[(i)]
        either as the largest Lie semigroup in $S$
      \item[(ii)]
        or as (the closure of) the reachable set of the identity ${\rm id}$ of the bilinear control system\footnote{
        Here \eqref{eq:lift} can be regarded as the operator lift of system \eqref{eq:bilin}.
          }
          \vspace{-2mm}
        \begin{equation}\label{eq:lift}
          \dot{\Phi}(t) = L(t) \Phi(t)\,,
          \vspace{-2mm}
        \end{equation}
        where $L(t) \in \mathsf{L}(S)$ acts as control and the set of admissible controls can be any set of
        locally integrable functions which contains at least all piecewise constant ones.
      \end{enumerate}\vspace{-1mm}
        \item[(b)]
    If there exists a Lie wedge $\mathfrak{w}_0$ which is global in $\mathcal B(\mathcal Z)$ or $\GL(\mathcal Z)$
    and $\mathfrak{w}$ satisfies $\mathfrak{w} \subseteq\mathfrak{w}_0$ as well as
    $
    \mathfrak{w} \setminus \mathsf{E}(\mathsf{L}(\mathfrak{w})) \subseteq \mathfrak{w}_0 \setminus \mathsf{E}(\mathsf{L}(\mathfrak{w}_0))\,,
    $
    then $\mathfrak{w}$ is global in $\mathcal B(\mathcal Z)$ or $\GL(\mathcal Z)$.
  \end{enumerate}
\end{thm}
\noindent
The proof is given in App.~\appref{B}.


  \begin{remark}\label{rem:semigroup-scaling}
    As $\mathsf L(S)$ is a convex cone,
any ``time''-scaling $\tau \mapsto \Phi(\mu\tau)$, $\mu > 0$ of a solution $\tau \mapsto \Phi(\tau)$
    of \eqref{eq:lift} is again a solution of \eqref{eq:lift}.  Thus one has 
    $\overline{\mathfrak{reach}({\rm id},\tau)} = S_0$ for all $\tau > 0$.
    Yet, a distinction $\overline{\mathfrak{reach}_{[0,\tau]}({\rm id})} \neq S_0$ (for $\tau$ sufficiently small) is brought about for example by restricting $L(t)$ to a bounded spanning set (e.g., $\mathsf{L}(S)\cap B_K(0)$ for $K>0$), cf.~\cite{Lawson99}. 
    This allows for introducing
    a non-arbitrary physical time parameter: imposing appropriate restrictions on the global energy-conserving Hamiltonian (for system and bath in
    Stinespring dilation, see Eq.~\eqref{eqn:CD-TO})
    concomitantly and consistently restricts the generators of (Markovian) thermal operations to a bounded spanning set 
    via the constructive projection of 
    Thm.~\ref{thm:markov_generator}.
    --- For qubit thermal operations, this is explicitly carried out in the 
    Worked Example\footnote{{Once a timescale of interest is fixed, such as the dephasing of the off-diagonals (in the example denoted by $x$),
    then the only parameter in the scaled global Hamiltonian 
    $\tfrac{1}{\sqrt{x}}H_{\mathsf{tot}}$ is the \textit{bounded} ratio of thermal{is}ation of diagonal elements to dephasing $\frac{u}{x}\in[0,\frac{2}{1+\varepsilon}]$.}} in Sec.~\ref{sec_thermo_markov}.
  \end{remark}


\paragraph*{Relation to Markovianity of Quantum Maps \cite{DHKS08,OSID17}.} 
These abstract ideas have concrete implications and a direct link to Kossakowski's work. They provide
the tools for constructing Markovian counterparts to various types of quantum maps.
Following~\cite{DHKS08}, the well-established \gks-results
can be recast to 
  a characterisation of the Lie wedge of the semigroup $\cptp$ of {\em all} quantum maps (of fixed finite dimension): 
  the set of {\em all} infinitesimal generators $-(i\adr_H+\hGA)$ of \gks-form \eqref{eq:lindblad_V} constitutes the (global)
  Lie wedge $\fw_{{\sf GKSL}}$ of $\cptp$. It generates the largest {\em Lie semigroup} $S_{{\sf GKSL}}$ contained in $\cptp$.
  Thus $S_{{\sf GKSL}}$ is given by the closure of all maps which can be written as finite product of maps
  ${\rm e}^{-t(i\operatorname{ad}_H+\hGA)}$, $t\geq 0$. Equivalently, one may see $S_{{\sf GKSL}}$ 
  as the set of quantum
  maps which can be obtained as solutions to the operator lifts of (possibly time-dependent) \gks-master equations,
  cf.~\eqref{eq:lift}. This is why---in anticipation of the proper definitions given in \eqref{eqn:MSemigroups} below---we identify $S_{{\sf GKSL}}$ 
  with the set of all {\em (time-dependent) Markovian quantum maps} ${\sf MCPTP}$.
  To sum up, one gets 
  $$
  S_{{\sf GKSL}} = \overline{\expt{\exp \fw_{{\sf GKSL}}}}_{\rm SG}\,=\,{\sf MCPTP} \quad \text{from} \quad
  \fw_{{\sf GKSL}} = \mathsf{L}(S_{{\sf GKSL}}) = \mathsf{L}(\sf CPTP)\,.
  $$
  
  The remaining quantum maps which do {\em not} belong to the Lie semigroup $S_{{\sf GKSL}}$ are called {\em non-Markovian}:
  they cannot be obtained by solutions to the operator lifts of (possibly time-dependent) \gks-master
  equations and hence they are not infinitesimal $\cptp$-divisible in the sense of \cite{Wolf08a}.

  The connection to {\em time-independent Markovian quantum maps}
is revealing: 
Although the set of all time-{\em in}\/dependent {Markovian quantum maps} (i.e.~the collection $\exp(\fw_{{\sf GKSL}})$ of all one-parameter
Lie semigroups) generates up to closure the entire {\em Lie semigroup} $S_{{\sf GKSL}}$, the concatenation of
two time-{\em in}\/dependent Markovian quantum maps does in general {\em not} give yet another time-{\em in}\/dependent Markovian quantum map\footnote{More
  precisely, close to the identity 
  the concatenation of two time-{\em in}\/dependent Markovian quantum maps is
  again a time-{\em in}\/dependent Markovian quantum map if their generators are part of a Lie subwedge $\fw$ 
  of $\fw_{{\sf GKSL}}$
  taking the special form of a {Lie semialgebra}:
  A Lie wedge $\fw$ is called {\em Lie semialgebra}, if it is locally (i.e.\ near the
  origin) closed under Baker-Campbell-Hausdorff (BCH) multiplication $X,Y\in\fw \mapsto X\star Y:=\log(e^Xe^Y)$.
  This requires an open BCH neighbourhood $B$ of the origin such that $(\fw\cap B)\star(\fw\cap B) \subseteq\fw$~\cite{HHL89}.\label{footnote_BCH}}.
  These results are summarised in Tab.~\ref{table0},
  further details can be found in \cite{DHKS08,OSID17}. 
\begin{table}[!h]
\caption{Lie-semigroup properties decide Markovianity type of quantum maps (QMs) in terms of their infinitesimal
generators as detailed in \cite{DHKS08,OSID17}.}\label{table0}\vspace{1mm}
\begin{tabular}{l|cc} 
\hline\\[-2mm]
Markovianity Type 
& structure of QMs & inf.~generators\\[2mm]
\hline\hline\\[-2mm]
time-{\em in}\/dependent\textsuperscript a & collection of 1-par.~semigroups\textsuperscript d & \gks-Lie wedge\\
time-{\em de}\/pendent\textsuperscript b & largest Lie semigroup $S_{{\sf GKSL}}$ & \gks-Lie wedge\\[3mm]
{\em non}\/-Markovian\textsuperscript c & QMs outside Lie semigroup $S_{{\sf GKSL}}$ & --- \\[2mm]
\hline\hline
\multicolumn{3}{c}{}\\[-2mm]
\multicolumn{3}{l}{\footnotesize \textsuperscript a: infinitely resp.~\textsuperscript b: infinitesimal (and not just infinitely) $\mathsf{CP}$-divisible \cite{Wolf08a}}\\[-1mm]
\multicolumn{3}{l}{\footnotesize \textsuperscript c: {\em not} infinitesimal $\mathsf{CP}$-divisible \cite{Wolf08a}}\\[-1mm]
\multicolumn{3}{l}{\footnotesize \textsuperscript d: generally not closed under concatenation, cf.~footnote \ref{footnote_BCH}
}\\
\end{tabular}
\end{table}%

\subsection{Application to Thermomajor{is}ation}\label{sec_thermo_markov}

Having settled the foundations,
let us apply above semigroup theory to 
explicit sets of quantum maps.
For understanding how to ``use
thermal resources to enhance control systems''
we make the notion of thermal resources precise and
specify which resources fit into the (Markovian) dynamical picture of the
control framework.
As mentioned in the introduction
this is the goal of the resource theory approach to quantum 
thermodynamics:
it formal{is}es 
which operations can 
be carried out at no cost (e.g., work) by defining a set of operations 
``allowed'' under some basic thermodynamic assumptions.
One set commonly used is the following:
given an $n$-level system with Hamiltonian $H_0$ and some fixed background 
temperature $T\in(0,\infty]$
the \textit{thermal operations} $\mathsf{TO}(H_0,T)$ are defined to be \cite{Lostaglio19r,vomEnde22thermal}
\begin{align}\label{def:TO}
\Big\{ \operatorname{tr}_B\Big(U\Big((\cdot)\otimes 
\frac{e^{-H_B/T}}{\operatorname{tr}(e^{-H_B/T})}\Big)U^\dagger\Big) :\ 
\substack{m\in\mathbb N,H_B\in{i}\mathfrak{u}(m),U\in\mathsf{U}(mn)\\
U(H_0\otimes\mathbbm{1}_B+\mathbbm{1}\otimes 
H_B)U^\dagger=H_0\otimes\mathbbm{1}_B+\mathbbm{1}\otimes H_B }\Big\}.
\end{align}
Here $\mathsf U(m)$ is the unitary group in $m$ dimensions 
with its Lie algebra $\mathfrak{u}(m)$ being the set of all $m\times m$ 
skew-Hermitian matrices and $\operatorname{tr}_B$ is the partial trace over the bath, i.e.~the unique linear map $\operatorname{tr}_B:\mathbb C^{n\times n}\otimes\mathbb C^{m\times m}\to\mathbb C^{n\times n}$ which satisfies $\operatorname{tr}(X\operatorname{tr}_B(\rho))=\operatorname{tr}((X\otimes\mathbbm1)\rho)$ for all $X\in\mathbb C^{n\times n}$, $\rho\in\mathbb C^{n\times n}\otimes\mathbb C^{m\times m}$.

Using the short-hand notation $\rho_B^{(T)}:={e^{-H_B/T}}/{\operatorname{tr}(e^{-H_B/T})}$
for the Gibbs state of the bath
and restricting the global unitaries $U$ to the commutant of ($H_0\otimes\mathbbm{1}_B+\mathbbm{1}\otimes H_B$)
to preserve the total energy of system (S) and bath (B), 
quantum maps brought about by thermal operations of the form
of Eq.~\eqref{def:TO} can be envisaged in the spirit of a Stinespring dilation as
\begin{equation}\label{eqn:CD-TO}
\begin{tikzcd}
{\rho_S(0)\otimes\rho_B^{(T)}\quad} \arrow[r, "\Adr_{U}", "(2)" '] 
& {\quad \rho_{SB}(U)} \arrow[d, "\operatorname{tr}_B", "(3)" '] \\[2mm]
\rho_S(0)\quad
\arrow[hook,
u, "\iota_B", "(1)" '] \arrow[r,"\mathsf{TO}^{\phantom |}\negthickspace{(H_0,T)}" ']
& \quad\rho_S(U)\,.
\end{tikzcd}
\end{equation}

\noindent
Thus $\mathsf{TO}(H_0,T)$ collects all quantum 
maps $\Phi$ which model a three-step procedure:
(1) coupling the system to an arbitrary finite-dimensional 
bath with Hamiltonian $H_B$ (and 
temperature $T$), followed (2) by a unitary transformation
on the full system which leaves the global energy invariant, and finally (3) discarding the 
bath by projecting back onto the system via $\tr_B$. 

Up to now, the description has been ``thermostatic" 
with no explicit continuous time-parameter involved. 
When introducing time evolutions by constructing quantum maps as curves $\Phi(t)\in\mathsf{TO}(H_0,T)$
(see Lem.~\ref{lem:wedge} and Thm.~\ref{thm:markov_generator}), one may adopt above
diagram to a Schr{\"o}dinger picture by taking  
one-parameter groups of time evolutions $U(t)$ [again from the commutant of 
($H_0\otimes\mathbbm{1}_B+\mathbbm{1}\otimes H_B$) to preserve energy of system and bath]
as global unitaries in \eqref{eqn:CD-TO} 
to arrive at quantum maps $\Phi(t)\in\mathsf{TO}(H_0,T) : \rho_S(0) \mapsto \rho_S(t)$
describing the time evolutions of the system by curves of thermal operations.
As mentioned in Rem.~\ref{rem:semigroup-scaling} above, the time scaling itself can then be induced
by the choice of global Hamiltonian, see again Thm.~\ref{thm:markov_generator} below.

Key topological properties of $\mathsf{TO}(H_0,T)$ are collected in the following: 

\begin{proposition}\label{prop_1}
Given $H_0\in{i}\mathfrak u(n)$, $T\in(0,\infty]$ the following statements hold:
\begin{itemize}
\item[(i)] $\mathsf{TO}(H_0,T)$ is a bounded, path-connected semigroup with 
identity.
\item[(ii)] $\overline{\mathsf{TO}(H_0,T)}$ is a convex, compact semigroup with 
identity.
\item[(iii)] $\overline{\mathsf{TO}(H_0,T)}$ is a subset of $\mathsf{Gibbs}(H_0,T)$ which is defined to be the collection of all $\cptp$ maps 
which leave $e^{-H_0/T}$ invariant.
\item[(iv)] One has $[\Phi,\operatorname{ad}_{H_0}]=0$ for all $
\Phi\in\overline{\mathsf{TO}(H_0,T)}$.
\end{itemize}
\end{proposition}
Statements (i) through (iii) can be found in Sec.~II of \cite{vomEnde22thermal} 
while statement (iv) is Thm.~1 in \cite{Lostaglio15_2}.
Some intuition as to whence condition (iv)
can be gained from the following basic observation:
given any system with Hamiltonian $H_0=\sum_{j=1}^nE_j|g_j\rangle
\langle g_j|$ in state $\rho=(\langle g_j,\rho g_k\rangle)_{j,k=1}^n$ there exists 
a thermal operation which mixes $\rho_{ij}$ and $\rho_{kl}$ if and only if
$E_i-E_j=E_k-E_l$ \cite[Rem.~3]{vomEnde22thermal}. But $E_i-E_j,E_k-
E_l\in\sigma(\ad_{H_0})$ so the action of any thermal operation is restricted by the 
degeneracies of $\operatorname{ad}_{H_0}$.
{
In particular, choosing $i=j$ or $k=l$ shows that a thermal operation can mix diagonal and
off-diagonal elements only if $H_0$ is degenerate.
This turns out to be a special case of the following (well-known and readily verified) result:
\begin{corollary}\label{cor:thermo-diagonal}
If $H_0$ has
non-degenerate spectrum, then the diagonal and the off-diagonal (w.r.t~any eigenbasis of $H_0$) action of any 
$\cptp$ map satisfying the covariance law from Prop.~\ref{prop_1}~(iv) are strictly separated. 
\end{corollary}}
This symmetry 
is the defining property
of what is called the set of \textit{enhanced thermal operations} (sometimes \textit{thermal processes})~\cite{Cwiklinski15}:
\begin{align*}
\mathsf{EnTO}(H_0,T):=\{S\in \mathsf{Gibbs}(n)\,:\,[S,\operatorname{ad}_{H_0}]=0\}\;.
\end{align*}
Note that the enhanced thermal operations---just like the closure of the thermal operations---form a 
convex, compact semigroup with identity.

Now a central task in this framework is character{is}ing if 
some initial state can be transformed into a given target state
by a thermal operation.
Thus one naturally defines \textit{thermomajor{is}ation} as the major{is}ation induced by the semigroup $\overline{\mathsf{TO}(H_0,T)}$ (the latter sometimes denoted $\mathsf{CTO}(H_0,T)$ \cite{Gour22}), and the
set of all states thermomajor{is}ed by some $\rho\in\pos n$ is defined as
\begin{equation}\label{eq:def_M_H_T}
 M_{H_0,T}(\rho):=M_{\overline{\mathsf{TO}(H_0,T)}}(\rho)=\big\{ \Phi(\rho) : \Phi\in\overline{\mathsf{TO}(H_0,T)} \big\}\,.
\end{equation}
This
is also known as {\em (future) thermal 
cone} \cite{Lostaglio15_2,Korzekwa17,Oliveira22} or, in the case of 
quasi-classical states (i.e.~states $\rho$ which satisfy $[\rho,H_0]=0$) as thermal polytope \cite{Alhambra19} or thermomajor{is}ation polytope \cite{PolytopeDegen22}.
Note that in \eqref{eq:def_M_H_T} we did not choose $\mathsf{TO}(H_0,T)$ 
but its closure (which is known to make a difference \cite{Lostaglio15_2})
because this guarantees ``reasonable'' mathematical structure:
combining Prop.~\ref{prop:S-majorization} and~\ref{prop_1} shows that 
$M_{H_0,T}(\rho)$ is a convex compact subset of $\pos n$.
Similarly one can define the semigroup major{is}ation induced by $\mathsf{EnTO}(H_0,T)$ and $\mathsf{Gibbs}(H_0,T)$, respectively.
While $\overline{\mathsf{TO}(H_0,T)}
\subseteq\mathsf{EnTO}(H_0,T)\subseteq\mathsf{Gibbs}(H_0,T)$ for all $H_0\in{i}\mathfrak u(n)$, $T\in(0,
\infty]$ (Prop.~\ref{prop_1}) 
it is known that---although the action of these sets coincides on the diagonal (cf.~Sec.~\ref{sec:d_maj})---the corresponding notions of major{is}ation are strictly different
(as long as $n>2$ \cite{Cwiklinski15,vomEnde22})
due to an explicit counterexample
by Ding et al.~\cite{Ding21}.
Therefore it makes a conceptual difference whether one defines 
thermomajor{is}ation via thermal operations, enhanced thermal operations, or
the Gibbs-preserving quantum maps \cite{Faist17}.

In a recent approach to unify {\em dynamics} of open quantum systems with 
quantum thermodynamics, Lostaglio and Korzekwa 
\cite{LosKor22a} studied the intersection of enhanced thermal 
operations with (time-dependent) Markovian quantum maps: they
character{is}ed which states can be generated by such maps
in case the initial state $\rho$ is quasi-classical.
Their approach was to consider Markovian dynamics $(e^{-tL})_{t\geq 0}$ (i.e.~$L=i\operatorname{ad}_{H_0}+\hGA$ with $\hGA$ as in \eqref{eq:lindblad_V})
and study the thermodynamic restraints the condition $(e^{-tL})_{t\geq 0}\subseteq \mathsf{EnTO}(H_0,T)$ 
imposes on the generator of 
the system-environment-interaction $\hGA$ (which until now was arbitrary except its \gks-form).
They found that such dynamics are in $\mathsf{EnTO}(H_0,T)$ at all times
if and  only if\footnote{
The straightforward identity $[B^m,A]=\sum_{k=1}^mB^{m-k}[B,A]B^{k-1}$
($m\in\mathbb N$) shows that 
$[B,A]=0$ if and only if $[B^m,A]=0$ for all $m$ if and only if $[e^{tB},A]=0$ for 
all $t\geq 0$.
}
$[L,\operatorname{ad}_{H_0}]=0$ and $L(e^{-H_0/T})=0$, cf.~also 
\cite{Kosloff13}.

Cast into the framework of Lie semigroups what they did translates into 
characterising the Lie wedge of $\mathsf{EnTO}(H_0,T)$.
Accordingly
this motivates us to introduce the following constructive {\em definitions} for the {\em Markovian} counterparts of the sets above:
\begin{equation}\label{eqn:MSemigroups}
\begin{split}
\mathsf{MTO}(H_0,T)&:=\overline{\langle \exp\big(\mathsf{L}(\overline{\mathsf{TO}(H_0,T)})\big) \rangle}_{\rm SG}\\
\mathsf{MEnTO}(H_0,T)&:=\overline{\langle \exp\big(\mathsf{L}(\mathsf{EnTO}(H_0,T))\big) \rangle}_{\rm SG}\\
\mathsf{MGibbs}(H_0,T)&:=\overline{\langle \exp\big(\mathsf{L}(\mathsf{Gibbs}(H_0,T))\big) \rangle}_{\rm SG}
\end{split}
\end{equation}
As before, the overbar denotes the closure and
  $\langle \;\cdot\;\rangle_{\rm SG}$ is the semigroup generated by the
set in question.
In other words $\mathsf{MEnTO}(H_0,T)$ is the collection of all enhanced thermal operations that are time-dependent Markovian, 
which likewise holds for
$\mathsf{MTO}(H_0,T)$ and $\mathsf{MGibbs}(H_0,T)$. 
Translated to this language Lostaglio and Korzekwa studied the semigroup major{is}ation induced by $\mathsf{MEnTO}(H_0,T)$ in the quasi-classical realm.
We provide a sketch of the sets introduced in this section in Fig.~\ref{fig:setinclusions}.

\begin{remark}
By definition $\mathsf{MTO}(H_0,T)$, $\mathsf{MEnTO}(H_0,T)$, $\mathsf{MGibbs}(H_0,T)$ contain all propagators of bilinear control systems
(Sec.~\ref{sec:MarkovControl}) under the corresponding 
thermodynamic constraints.
This connection allows us to translate results from quantum control theory
into this major{is}ation framework:
for example Prop.~5.2.1 in \cite{vE_PhD_2020} implies that if $e^{-H_0/T}$ is a multiple of the identity, then the semigroup major{is}ation induced by $\mathsf{Gibbs}(H_0,T)$ and by $\mathsf{MGibbs}(H_0,T)$ coincide.
In other words each unital state transfer can also be real{is}ed by unital (time-dependent) {\em Markovian} maps.
\end{remark}

\begin{figure}[!ht] \centering
\includegraphics[width=0.55\textwidth]{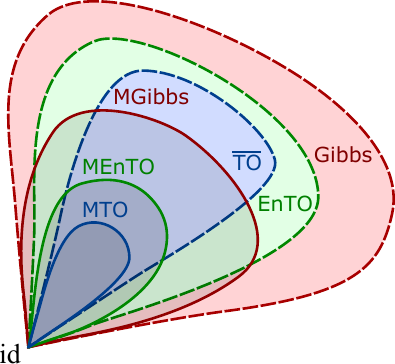}
\caption{(Colour online). 
Sketch of the semigroups $\mathsf{(En)TO}$ and $\mathsf{Gibbs}$, as well as their Markovian counterparts $\mathsf{M(En)TO}$ and $\mathsf{MGibbs}$  defined as {\em Lie semigroups} generated by their corresponding Lie wedges, see also the qubit example in Fig.~\ref{fig_TO_qubit_Markov}\,. Note that Markovianity depends on the semigroup, i.e.~there might be elements in $\mathsf{EnTO}\cap\mathsf{MGibbs}$ 
outside of $\mathsf{MEnTO}$.
{{
A similar set-inclusion for {\em elementary} thermal operations ($\mathsf{ETOs}$, which in general are a proper subset of $\overline{{\sf TO}}$s) 
can be found in~\cite[Fig.~3]{Lostaglio18}}.
}}\label{fig:setinclusions}
\end{figure}

In the remainder of this section we general{is}e the result of Lostaglio and Korzekwa 
\cite{LosKor22a} to the more physically motivated set of thermal operations.
%
%
We begin
studying its {\em Lie semigroup structure} by specifying the edge of the sets defined above:


\begin{lemma}\label{lemma_edge}
Given $H_0,H\in{i}\mathfrak u(n)$ and $T\in(0,\infty]$ the following statements are equivalent:
\begin{itemize}
\item[(i)] $[H,H_0]=0$
\item[(ii)] $(\operatorname{Ad}_{e^{-itH}})_{t\geq 0}\subseteq\mathsf{TO}(H_0,T)$
\item[(iii)] $(\operatorname{Ad}_{e^{-itH}})_{t\geq 0}\subseteq\overline{\mathsf{TO}(H_0,T)}$
\item[(iv)] $(\operatorname{Ad}_{e^{-itH}})_{t\geq 0}\subseteq\mathsf{EnTO}(H_0,T)$
\end{itemize}
Moreover, if $T\in(0,\infty)$, then all of the above statements are equivalent to:
\begin{itemize}
\item[(v)] $\operatorname{Ad}_{e^{-itH}}(e^{-H_0/T})=e^{-H_0/T}$ for all $t\geq 0$.
\end{itemize}
\end{lemma}
\begin{proof}
``(i) $\Rightarrow$ (ii)'': In the definition of the thermal operations \eqref{def:TO} 
choose $m=1$ and $U=e^{-itH}$. Doing so is allowed as $[H,H_0]=0$ by assumption.
\mbox{``(ii) $\Rightarrow$ (iii) $\Rightarrow$ (iv)'':} Obvious from $\mathsf{TO}(H_0,T)\subseteq\overline{\mathsf{TO}(H_0,T)}\subseteq \mathsf{EnTO}(H_0,T)$ \cite{Lostaglio15_2}.
``(iv) $\Rightarrow$ (i)'': Evidently, the generator of $(\operatorname{Ad}_{e^{-itH}})_{t\geq 0}$ is $-i\operatorname{ad}_H$. By our previous 
considerations $-i\operatorname{ad}_H\in \mathsf{L}(\mathsf{EnTO}(H_0,T))$ implies $[\operatorname{ad}_H,\operatorname{ad}_{H_0}]\equiv 0$.
Using $[\operatorname{ad}_H,\operatorname{ad}_{H_0}]=\operatorname{ad}_{[H,H_0]}$ this is equivalent to $[H,H_0]=\lambda\mathbbm 1$ for some $\lambda\in\mathbb C$; but this $\lambda$ has to vanish as $0=\operatorname{tr}([H,H_0])=\lambda\operatorname{tr}(\mathbbm1)=\lambda n$.

Now assume that $T<\infty$. ``(iv) $\Rightarrow$ (v)'': By definition of $\mathsf{EnTO}(H_0,T)$.
``(v) $\Rightarrow$ (i)'': Differentiating (v) at zero gives $[H,e^{-H_0/T}]=0$ so there exists $U\in\mathbb C^{n\times n}$ unitary such that $UHU^\dagger$ and $Ue^{-H_0/T}U^\dagger$ are both diagonal \cite[Thm.~2.5.5]{HJ1}.
But $Ue^{-H_0/T}U^\dagger=e^{-UH_0U^\dagger/T}$ so using functional calculus, the matrix
$
(-T)\ln\big(e^{-UH_0U^\dagger/T}\big)=UH_0U^\dagger
$
has to be diagonal due to $T\in(0,\infty)$.
Thus $U[H,H_0]U^\dagger=[UHU^\dagger,UH_0U^\dagger]=0$, hence (i) follows.
\end{proof}

As the inverse of a bijective $\cptp$ map is again $\cptp$ if and only if it is a unitary map\footnote{
Actually the map is unitary iff the inverse is positive, cf.~\cite[Cor.~3]{Wolf08a} \& \cite[Prop.~1]{vE_dirr_semigroups}.
} \cite[Thm.~III.2]{Buscemi05},
$\mathsf{E}(\mathsf{L}(\cptp(n)))=\{{i}\ad_H:H\in{i}\mathfrak u(n)\}$.
But the edge of any subsemigroup of $\cptp(n)$ has to be a subspace of $\mathsf{E}(\mathsf{L}(\cptp(n)))$, and the condition for one-parameter groups $(\operatorname{Ad}_{e^{-itH}})_{t\geq 0}$ to live in the corresponding semigroup---depending on the temperature $T$---is precisely condition (i) from Lem.~\ref{lemma_edge}.
This immediately yields the following:

\begin{corollary}\label{coro_edge}
Given $H_0\in{i}\mathfrak u(n)$, $T\in(0,\infty]$ one finds that $\mathsf{E}(\mathsf{L}(\mathsf{TO}(H_0,T)))$, $\mathsf{E}(\mathsf{L}(\overline{\mathsf{TO}(H_0,T)}))$, and $\mathsf{E}(\mathsf{L}(\mathsf{EnTO}(H_0,T)))$ equal $\{{i}\ad_H:H\in{i}\mathfrak u(n), [H,H_0]=0\}$,
and if $T\in(0,\infty)$ the same holds true for 
$\mathsf{E}(\mathsf{L}(\mathsf{Gibbs}(H_0,T)))$.
\end{corollary}
The next step is to specify the {\em Lie wedge} of the semigroups in question.
This turns out to be almost trivial for the Gibbs-preserving maps and the enhanced thermal operations:
\begin{lemma}\label{lem:wedge}
Given $H_0,H\in{i}\mathfrak u(n)$ and $T\in(0,\infty]$ one has:
\begin{align*}
\mathsf{L}(\mathsf{Gibbs}(H_0,T))\phantom{.}&= \{L\in\mathcal L(\mathbb C^{n\times n})\,:\,L\text{ is of \gks{}-form and }L(e^{-H_0/T})=0 
  \}\\
\mathsf{L}(\mathsf{EnTO}(H_0,T))&= \{L\in\mathcal L(\mathbb C^{n\times n})\,:\,  
L\in\mathsf{L}(\mathsf{Gibbs}(H_0,T))\text{ and }[L,\operatorname{ad}_{H_0}]=0
  \}
\end{align*}
\end{lemma}
\noindent One can see this by taking the condition in Eq.~\eqref{eq:liewedge_def} and differentiating at zero to get back the generator $L$.

In contrast, specifying the Lie wedge of $\overline{\mathsf{TO}(H_0,T)}$
is more involved.
This is where we make use of Prop.~\ref{prop:Lie-wedge}~(iv) as well as Cor.~\ref{cor:Lie-wedge}
which state that elements of $\mathsf{L}(S)$ are character{is}ed via
derivatives of certain curves in $S$ starting at the identity.
Thereby introducing a ``time" parameter, an obvious way to specify such curves 
for the case of $\mathsf{TO}(H_0,T)$ is by interpreting the energy-preserving unitary as 
endpoint of a curve of unitaries starting at zero in the sense of the diagram Eq.~\eqref{eqn:CD-TO}.
This is done
by choosing $H_\mathsf{tot}$ such that $U=e^{-iH_\mathsf{tot}}$ in
$$
 \operatorname{tr}_B\big(U\big((\cdot)\otimes 
\rho_B^{(T)}
\big)U^\dagger\big)= \operatorname{tr}_B\big(e^{-itH_\mathsf{tot}}\big((\cdot)\otimes 
\rho_B^{(T)}
\big)e^{itH_\mathsf{tot}}\big)\big|_{t=1}\;.
$$
By definition this curve is in $\mathsf{TO}(H_0,T)$ at all times.
Such curves and their first and second derivative have been studied recently 
by one of us: for all $m\in\mathbb N$, $H\in{i}\mathfrak u(mn)$, $\omega\in\pos m$
\begin{equation}\label{eq:stinespring_taylor}
\begin{split}
\operatorname{tr}_{\mathbb C^m}\big(e^{-itH}&((\cdot)\otimes\omega)e^{itH}\big)\equiv \\
&\equiv\operatorname{id}- it\big[\operatorname{tr}_\omega(H),\,\cdot\,\big]-\frac{t^2}{2} \sum_{j,k=1}^m \hGA_{\sqrt{2r_k}\operatorname{tr}_{|g_k\rangle\langle g_j|}(H)} +\mathcal O(t^3)
\end{split}
\end{equation}
if $\omega=\sum_{i=1}^mr_i|g_i\rangle\langle g_i|$ for $r_i\geq 0$ and an orthonormal basis $\{g_i\}_{i=1}^m$ of $\mathbb C^m$.
Here, given matrices $X\in\mathbb C^{m\times m}$, $B\in\mathbb C^{mn\times mn}$ the expression $\operatorname{tr}_X(B)$ is called ``partial trace of $B$ with respect to $X$''\,\footnote{
More precisely, $\operatorname{tr}_X(B)$ is set to be the unique $n\times n$-matrix which satisfies $\operatorname{tr}(A\operatorname{tr}_X(B))=\operatorname{tr}((A\otimes X)B)$ for all $A\in\mathbb C^{n\times n}$ \cite[Ch.~9, Lem. 1.1]{Davies76}.
Note that this recovers the ``usual'' partial trace when setting $X=\mathbbm1$.
}.
A proof can be found in \cite[Eq.~(5)]{vE22_Stinespring} or, for the reader's convenience, in App.~\appref{C}.

Now applying Prop.~\ref{prop:Lie-wedge}~(iv) to these curves yields the Hamiltonian generator ${i}
[\operatorname{tr}_{\rho_B^{(T)}}(H_\mathsf{tot}),\,\cdot\,]$ which by
Cor.~\ref{coro_edge} and Eq.~\eqref{eq:stinespring_taylor} is in $\mathsf{E}(\mathsf{L}(\overline{\mathsf{TO}(H_0,T)}))$.
However, we are interested in the dissipative action of the curve which is locked away behind the second derivative.
This is where we apply Cor.~\ref{cor:Lie-wedge}:
\begin{thm}\label{thm:markov_generator}
Let $m\in\mathbb N$, $H_B\in{i}\mathfrak u(m)$, as well as $H_{\mathsf{tot}}\in{i}\mathfrak u(mn)$ be given such that $[H_{\mathsf{tot}},H_0\otimes\mathbbm{1}+\mathbbm{1}\otimes H_B]=0$.
If
$\Phi$ is the solution to
\begin{align*}
\dot\Phi(t)=\big(-i\operatorname{ad}_H-\hGA_{B,{\mathsf{tot}}}\big)\Phi(t)\,,\quad \Phi(0)=\operatorname{id}\,
\end{align*}
with $H$ any element of ${i}\mathfrak u(n)$ such that $[H,H_0]=0$ and 
\begin{equation*}
\hGA_{B,{\mathsf{tot}}}:= \sum_{j, k=1}^m\Big( \tfrac12 \big(V_{j k}^\dagger V_{j k} (\cdot)+(\cdot) V_{j k}^\dagger V_{j k}\big)-V_{j k}(\cdot) V_{j k}^\dagger \Big)\,,
\end{equation*} 
$V_{jk}=e^{-E_k'/(2T)}\operatorname{tr}_{|g_k\rangle\langle g_j|}(H_{\mathsf{tot}})$ for all $j, k=1,\ldots,m$
where $\sum_{j=1}^m E_j'|g_j\rangle\langle g_j|$ is any spectral decomposition of the bath Hamiltonian $H_B$,
then $(\Phi(t))_{t\geq 0}$ is a 
continuous one-parameter semigroup in $\overline{\mathsf{TO}(H_0,T)}$.
\end{thm}
\begin{proof}
Consider
$
\gamma:[0,\infty) \to\overline{\mathsf{TO}(H_0,T)}$,
$t \mapsto \operatorname{tr}_B(e^{-itH_\mathsf{tot}}((\cdot)\otimes 
\rho_B^{(T)})e^{itH_\mathsf{tot}})$
and note that $\gamma$ is well-defined by assumption on $H_B,H_{\mathsf{tot}}$.
Also $\gamma(0)=\operatorname{id}$ and $\dot\gamma(0)\in \mathsf{E}(\mathsf{L}(\overline{\mathsf{TO}(H_0,T)}))$ by  Eq.~\eqref{eq:stinespring_taylor}
 \& Cor.~\ref{coro_edge}.
Thus---because $\overline{\mathsf{TO}(H_0,T)}$ is a compact, convex semigroup (Prop.~\ref{prop_1}~(ii))---we may apply Cor.~\ref{cor:Lie-wedge} to $\gamma$ which shows that $\ddot\gamma(0)\in \mathsf{L}(\overline{\mathsf{TO}(H_0,T)})$.
If $H_B=\sum_{j=1}^m E_j'|g_j\rangle\langle g_j|$, then $e^{-H_B/T}=\sum_{j=1}^m e^{-E_j'/T}|g_j\rangle\langle g_j|$ so, again using Eq.~\eqref{eq:stinespring_taylor}
\begin{align*}
\ddot\gamma(0)=-\frac{2}{\operatorname{tr}(e^{-H_B/T})} \sum_{j,k=1}^m \hGA_{\sqrt{e^{-E_k'/T}}\operatorname{tr}_{|g_k\rangle\langle g_j|}(H_{\mathsf{tot}})}
\end{align*}
is in $\mathsf{L}(\overline{\mathsf{TO}(H_0,T)})$.
But the latter is a convex cone (Prop.~\ref{prop:Lie-wedge}~(i)) so for all $H\in{i}\mathfrak u(n)$, $[H,H_0]=0$
\begin{align}\label{eq:wedge-TO}
\mathsf{L}(\overline{\mathsf{TO}(H_0,T)})\ni
-i\operatorname{ad}_H+
\frac{\operatorname{tr}(e^{-H_B/T})}{2}\ddot\gamma(0)=
-i\operatorname{ad}_H-\hGA_{B,{\mathsf{tot}}}
\end{align}
which concludes the proof.
\end{proof}
\begin{remark}\label{rem:TO-scaling}
Though $(\operatorname{tr}_{|g_k\rangle\langle g_j|}(H_{\mathsf{tot}}))^\dagger=\operatorname{tr}_{|g_j\rangle\langle g_k|}(H_{\mathsf{tot}})$, note the ``asymmetry'' in the $V_{jk}$ induced by temperature via the factor $e^{-E_k'/(2T)}$. 
Assuming $T<\infty$ prevents the dynamics from automatically being unital (i.e.~identity-preserving).
It is worth stressing that the temperature-independent projection $\operatorname{tr}_{|g_j\rangle\langle g_k|}(H_{\mathsf{tot}})$ also preserves the time scaling
chosen for the global energy conserving one-parameter unitary time evolution $e^{-itH_{\sf tot}}$ 
and takes it over into scaling the generators of the semigroup $\overline{\mathsf{TO}(H_0,T)}$ to be
{\em consistent} with the same time parameter---as anticipated in Rem.~\ref{rem:semigroup-scaling}.
\end{remark}

\noindent
At this point, Andrzej Kossakowski would have insisted that Thm.~\ref{thm:markov_generator} does not provide 
a complete characterisation yet---and rightly so: we have to leave the converse open for now,
but do so at the level of a well supported conjecture.
%
\begin{conjecture}\label{conj1}
The converse of Thm.~\ref{thm:markov_generator} holds true, up to taking the closure of the collection of all these generators.
\end{conjecture}
{\small
This conjecture is supported by the following two observations:
\begin{itemize}
\item Conjecture~\ref{conj1} holds for general quantum maps, that is, after getting rid of the commutator relation $[H_{\mathsf{tot}},H_0\otimes\mathbbm{1}+\mathbbm{1}\otimes H_B]=0$ as well as replacing $\overline{\mathsf{TO}}$ with $\cptp(n)$.
The only difference here is that only full-rank ancilla states are allowed whereas
\cite[Thm.~2]{vE22_Stinespring} uses arbitrary ancilla states.
This difference, however, vanishes in the closure.
\item 
For any
$H_0\in{i}\mathfrak u(n)$, $T\in(0,\infty]$
the collection of all generators described in Thm.~\ref{thm:markov_generator} is invariant under conjugation by arbitrary edge elements (i.e.~the necessary condition from Prop.~\ref
{prop:Lie-wedge}~(ii)).
Indeed given any $H'\in {i}\mathfrak u(n)$ with $[H',H_0]=0$---hence ${i}\ad_{H'}\in\mathsf{E}(\mathsf{L}(\overline{\mathsf{TO}(H_0,T)}))$ by Cor.~\ref{coro_edge}---one finds
\begin{align*}
&\operatorname{Ad}_{e^{-iH'}}\circ\big(-i\operatorname{ad}_H-\hGA_{B,{\mathsf{tot}}}\big)\circ\operatorname{Ad}_{e^{iH'}}\\
&\quad\ \ =-i[e^{-iH'}He^{iH'},\,\cdot\,]- \sum_{j, k=1}^m \hGA_{e^{-E_k'/(2T)}e^{-iH'}\operatorname{tr}_{|g_k\rangle\langle g_j|}(H_{\mathsf{tot}})e^{iH'}}\\
&\quad\ \  =-i[e^{-iH'}He^{iH'},\,\cdot\,]- \sum_{j, k=1}^m \hGA_{e^{-E_k'/(2T)}\operatorname{tr}_{|g_k\rangle\langle g_j|}((e^{-iH'}\otimes\mathbbm1)H_{\mathsf{tot}}(e^{-iH'}\otimes\mathbbm1)^\dagger)}\,.
\end{align*}
But this generator is of the form described in Thm.~\ref{thm:markov_generator} as can be seen by replacing $H_{\mathsf{tot}}$ by $(e^{-iH'}\otimes\mathbbm1)H_{\mathsf{tot}}(e^{-iH'}\otimes\mathbbm1)^\dagger$, replacing $H$ by $e^{-iH'}He^{iH'}$, and keeping $H_B$ as is.
One readily verifies that the new $H$ ($H_{\mathsf{tot}}$) commutes with $H_0$ ($H_0\otimes\mathbbm1+\mathbbm1\otimes H_B$), and thus is a valid generator;
this follows from $e^{-iH'}H_0e^{iH'}=e^{-i\operatorname{ad}_{H'}}(H_0)=0$ which is a direct consequence of $[H',H_0]=0$.
\end{itemize}
}%
%
\bigskip

\begin{remark}\label{rem:global-wedge-Markov}
To sum up, we emphasise that the Lie wedges to the quantum maps $\mathsf{Gibbs}$, $\mathsf{EnTO}$
(Lem.~\ref{lem:wedge}) and to $\overline{\mathsf{TO}}$ (Thm.~\ref{thm:markov_generator} and Eq.~\eqref{eq:wedge-TO})
are {\em global}, as by Thm.~\ref{thm_lie_global} they generate the corresponding Lie semigroups, which in turn
{\em define} the respective Markovian counterparts
$\mathsf{MGibbs}$,
$\mathsf{MEnTO}$, and
$\mathsf{MTO}$ via the construction of Eqs.~\eqref{eqn:MSemigroups}. The set inclusions are illustrated in Fig.~\ref{fig:setinclusions}.
\end{remark}

\bigskip
Explicitly scrutinising 
and visual{is}ing the above sets and relations in the case of qubits
will elucidate the concepts introduced in this section concretely. 

\medskip
\begin{myexample}

\subsection{\negthickspace Worked Example: Markovian Thermal Operations in Qubits}\vspace{1mm}\label{par:worked-ex1}

\noindent
Let us demonstrate how to arrive from a general non-Markovian (enhanced) thermal operation via the 
\gks-generator in its Lie wedge at the corresponding Lie semigroup giving the desired
{\em Markovian} counterpart to the given (enhanced) thermal operation as in Eq.~\eqref{eqn:MSemigroups}.
We study the simple case of a single qubit.
What is special about two-dimensional systems is twofold: first, time-independent and time-dependent Markovianity 
coincide (as explained at the end of this example), and secondly 
the thermal operations and
 the enhanced thermal operations approximately coincide \cite{Ding19}, \cite[Thm.~10]{vomEnde22thermal}.
 Thus,
every thermal operation for $H_0:=\tfrac{1}{2}\operatorname{diag}(-1,1)$ acts like
$$
\begin{pmatrix}
\rho_{11}&\rho_{12}\\\rho_{21}&\rho_{22}
\end{pmatrix}\mapsto
\begin{pmatrix}
(1-w\varepsilon)\rho_{11}+w\rho_{22}&z\rho_{12}\\z^*\rho_{21}&w\varepsilon\rho_{11}+(1-w)\rho_{22}
\end{pmatrix}
$$
for some $w\in[0,1]$, $z\in\mathbb C$ such that $|z|\leq\sqrt{(1-w)(1-w\varepsilon)}$ \cite[Sec.~IV]{vomEnde22thermal} with the short-hand $\varepsilon:=e^{-1/T}$ used henceforth. \\[0mm]

{\em Parameter{is}ing the Elements in the Lie Wedge}\\
With these stipulations, 
Lem.~\ref{lem:wedge} then allows to specify the elements in the Lie wedge of the (enhanced) thermal operations:
a linear map $L$ on $\mathbb C^{2\times 2}$ is a generator of a one-parameter semigroup in $\mathsf{EnTO}(H_0,T)$ if and only if
\begin{itemize} 
\item[(i)] $L$ preserves Hermiticity,
\item[(ii)] $L$ is trace-annihilating,
\item[(iii)] $L$ is conditionally completely positive (\cite[Cor.~1]{Lindblad76} \& \cite[Thm.~14.7]{EL77}),
	$$(\mathbbm1-|\Omega\rangle\langle\Omega|)\big((\operatorname{id}\otimes L)
	(|\Omega\rangle\langle\Omega|)\big)(\mathbbm1-|	\Omega\rangle\langle\Omega|)\geq 0$$
	where $\Omega=\frac1{\sqrt2}\begin{pmatrix}
	1&0&0&1
	\end{pmatrix}^\top$,
\item[(iv)] $L(e^{-H_0/T})=0$, and
\item[(v)] $[L,\operatorname{ad}_{H_0}]=0$.
\end{itemize}

Property (v) is equivalent to
$$
L\begin{pmatrix}
\rho_{11}&\rho_{12}\\
\rho_{21}&\rho_{22}
\end{pmatrix}=\begin{pmatrix}
x\rho_{11}+u\rho_{22}&z\rho_{12}\\ y\rho_{21}&v\rho_{11}+w\rho_{22}
\end{pmatrix}
$$
for some $u,v,w,x,y,z\in\mathbb C$ (cf.~also Cor.~\ref{cor:thermo-diagonal}).
Thereby (i), (ii), and (iv) combined are equivalent to $v=-x=ue^{-1/T}$, $w=-u$,
and $y=z^*$ where $u\in\mathbb R$.
Finally, evaluating conditional complete positivity 
yields
$u\geq 0$ and $2\operatorname{Re}z\leq -u(1+\varepsilon)$.
Altogether this shows $L\in\mathsf L(\mathsf{EnTO}(H_0,T))$
if and only if
$$
L\begin{pmatrix}
\rho_{11}&\rho_{12}\\
\rho_{21}&\rho_{22}
\end{pmatrix}=\begin{pmatrix}
u(\rho_{22}-\rho_{11}\varepsilon)&(-x+i\omega)\rho_{12}\\(-x-i\omega)\rho_{21}&u(\rho_{11}\varepsilon-\rho_{22})
\end{pmatrix}
$$
for some $u\geq 0$, $x,\omega\in\mathbb R$ such that $2x\geq u(1+\varepsilon)$.
Casting this into the Markovian $\gks$ master equation \eqref{eq:diss_evolution}
yields
%
(in superoperator representation) 
\begin{equation*}
L = -i\adr_{\operatorname{diag}({-\omega,\omega})/2} - \hGA_{\gks} 
\;\widehat=
\begin{pmatrix} -\varepsilon u &0 &0 &u\\ 0 &-x-i\omega &0 &0\\ 0 &0 &-x+i\omega &0\\ \varepsilon u &0 &0 &-u  \end{pmatrix}.
\end{equation*}
where the (by assumption well-defined) generators of $\hGA_{\gks}$ are
$$
V_{11}=\,\frac{\sqrt{2x-u(1+\varepsilon)}}{2}\begin{pmatrix}
1&0\\0&-1
\end{pmatrix},\ 
V_{12}=\,\sqrt{u}\begin{pmatrix}
0&1\\0&0
\end{pmatrix},\ 
V_{21}=\,\sqrt{u\varepsilon}
\begin{pmatrix}
0&0\\
1&0
\end{pmatrix}.
$$
Setting $V_{22}=0$, in the language of Thm.~\ref{thm:markov_generator} this corresponds to
$H_B:=\operatorname{diag}(-1,1)/2=H_0$  
and a global Hamiltonian (with $u\geq 0$ playing the role of a scaling parameter later used for fixing the time scale $t$)
$$
H_\mathsf{tot}:=\begin{pmatrix}
\frac12\sqrt{2x-u(1+\varepsilon)}&0&0&0\\
0&0&\sqrt{u}&0\\
0&\sqrt{u}&-\frac12\sqrt{2x-u(1+\varepsilon)}&0\\
0&0&0&0
\end{pmatrix}\in{i}\mathfrak u(4)\;.
$$

\noindent
{\em Character{is}ing One-Parameter Semigroups in $\mathsf{MEnTO}(H_0,T)$}\\[1mm]
For the one-parameter semigroup of time evolutions ($t\geq 0$)
one 
finds 
\begin{eqnarray}
S(t)\; =\;\; e^{tL} &\widehat=& \begin{pmatrix} \tfrac{1+\varepsilon e^{-tu(1+\varepsilon)}}{1+\varepsilon} &0 & 0 &\tfrac{1-e^{-tu(1+\varepsilon)}}{1+\varepsilon} \\ 
0 & e^{-(x+i{\omega})t} & 0 & 0\\
0 & 0 & e^{-(x-i{\omega})t} & 0\\
		\tfrac{\varepsilon(1-e^{-tu(1+\varepsilon)})}{1+\varepsilon}  &0 &0 &\tfrac{\varepsilon + e^{-tu(1+\varepsilon)}}{1+\varepsilon}   \end{pmatrix}\notag \\[2mm]
&=&\;%
\begin{pmatrix} 1-\varepsilon\mu_t &0 & 0 &\mu_t\\ 
0 & e^{-xt} e^{-i{\omega}t} & 0 & 0\\
0 & 0 & e^{-xt}e^{i{\omega}t} & 0\\
\varepsilon\mu_t &0 &0 &1-\mu_t  
\end{pmatrix}\;,\label{eq:S_t_qubit}
\end{eqnarray}
where the last identity uses another short-hand $\mu_t:=(1-e^{-tu(1+\varepsilon)})/(1+\varepsilon)$ with non-negative times
enforcing $0\leq\mu_t\leq 1/(1+\varepsilon)$.
For the action on just the vector of diagonal elements of the density operator one gets
\begin{equation*}
G(t) =
\begin{pmatrix} 1-\varepsilon\mu_t &\mu_t\\ 
\varepsilon\mu_t &1-\mu_t  
\end{pmatrix}\;,
\end{equation*}
which resembles the $\beta$-swap for a two-level system: Note that $S(t)$ and $G(t)$ are both Gibbs-stochastic for the entire fictitious 
parameter range\footnote{in which $S(t)$ resp.~$G(t)$ stabilise the Gibbs state $\frac{1}{1+\varepsilon}\operatorname{diag}(1,\varepsilon)$ 
	resp.~$\frac{1}{1+\varepsilon}(1,\varepsilon )^\top$}  
$\mu_t\in[0,1]$ exploited in the $\beta$-swap 
(then including ``complex times"\footnote{
	Formally $t=-{\ln(1-\mu_t(1+\varepsilon))}/{u(1+\varepsilon)}$; recall that  
	for complex $\ell=|\ell|e^{i\lambda}$ one has \mbox{\quad\;\, $\ln(\ell)=\ln(|\ell|)+ i\lambda$}, 
	as used in Fig.~\ref{fig_TO_qubit_Markov}\,.}), 
whereas
they are {\em Markovian} only in the positive-time segment $\mu_t\in[0,1/(1+\varepsilon)]$, whence $G(t)$ derives from a
one-parameter semigroup $S(t)\in\mathsf{MEnTO}(H_0,T)$ of the form of Eq.~\eqref{eq:S_t_qubit}
and has itself a non-negative determinant. 
\bigskip
\bigskip

\noindent
{\em How the Markovian ${\mathsf{MTO}}$s  Sit in General ${\mathsf{TO}}$s}\\[1mm]
More generally, (the superoperator of) any map $\Phi\in\overline{\mathsf{TO}(H_0,T)}$
is of the form
\begin{equation}\label{eq:TO_general_qubit}
\begin{pmatrix}
1-\varepsilon\mu&0&0&\mu\\
0&e^{-x}e^{-i\omega }&0&0\\
0&0&e^{-x}e^{i\omega }&0\\
\varepsilon\mu&0&0&1-\mu
\end{pmatrix}
\end{equation}
for $\mu\in[0,1]$, $x\geq 0$ so that $e^{-2x}\leq (1-\varepsilon\mu)(1-\mu)=1-\mu(1+\varepsilon)+\varepsilon\mu^2$, and it is again 
in $\mathsf{MTO}(H_0,T)$ if and only if $\mu\in[0, \tfrac{1}{(1+\varepsilon)}]$ (see \cite{Alhambra19,Ptas22}) plus $e^{-2x}\leq1-\mu(1+\varepsilon)$.
%
%
Since $\varepsilon=e^{-1/T}\to 0$ for temperatures $T\to 0^+$ %
it is obvious that---in the case of a single qubit---Markovianity becomes
no longer a restriction and $\mathsf{MTO}(H_0,T) \to 
\overline{\mathsf{TO}(H_0,T)}$ as $T\to 0^+$ (e.g., in the Hausdorff metric \cite[§21.VII]{Kuratowski66}).
Not only are these two scenarios equivalent in terms of \textit{state conversion} as $T\to 0^+$ 
(as was known for arbitrary finite dimensions \cite[Thm.~1]{CDC19}),
but equality also holds on the level of \textit{quantum maps}---in the case of qubits, which is new.
Whether this holds for qutrits and higher dimensions is an open question.

\bigskip
\bigskip

\noindent
{\em Visualisation}\\[1mm]
By Eqs.~\eqref{eq:S_t_qubit} \& \eqref{eq:TO_general_qubit}
one can visual{is}e the set of (Markovian) qubit thermal operations.
The main tool is the semigroup homomorphism $\Psi_T$ from the Hermiticity-preserving linear maps on $\mathbb C^{2\times 2}$ to $(\mathbb R\times\mathbb C,\circ_T)$\,\footnote{\label{fnote:assoc}%
The operation $\circ_T$ is defined
via
$(\mu_1,c_1)\circ_T(\mu_2,c_2):=(\mu_1+\mu_2-\mu_1\mu_2(1+\varepsilon),c_1c_2)$.
}
which acts like $\Psi_T:\Phi\mapsto(\langle e_1,\Phi(|e_2\rangle\langle e_2|)e_1\rangle,\langle e_1,\Phi(|e_1\rangle\langle e_2|)e_2\rangle)^\top$, cf.~\cite[Sec.~IV]{vomEnde22thermal}.
One finds 
$$
\Psi_T\big(\overline{ \exp{\big(\mathsf{L}(\overline{\mathsf{TO}(H_0,T)})}\big)}  \big)=\Big\{\left(\begin{smallmatrix}
\mu_t\\[.5mm] e^{2\pi i\phi}\, r\,
\sqrt{1-\mu_t(1+\varepsilon)}
\end{smallmatrix}\right):r, 
\phi\in[0,1];\, t\in[0,\infty]\Big\}\,.
$$
This set turns out to be a subsemigroup of $(\mathbb R\times\mathbb C,\circ_T)$ which implies that the closure of ${\exp{\big(\mathsf{L}(\overline{\mathsf{TO}(H_0,T)})\big)}}$ is a subsemigroup of $\overline{\mathsf{TO}(H_0,T)}$.
One gets
$$
\mathsf{MTO}(H_0,T)=\overline{\langle \exp\big(\mathsf{L}(\overline{\mathsf{TO}(H_0,T)})\big) \rangle}_{\rm SG}=\overline{\exp\big(\mathsf{L}(\overline{\mathsf{TO}(H_0,T)})}\big),
$$
which shows that $\mathsf{MTO}(H_0,T)$ is 
{\em weakly exponential}\,\footnote{A closed subsemigroup $\mathbf{S}\subseteq\mathcal B(\mathcal Z)$ 
with Lie wedge $\mathsf{L}(\mathbf{S})$
is {\em exponential} and {\em weakly exponential} if $\mathbf{S} = \exp(\mathsf{L}(\mathbf{S}))$
and $\mathbf{S} = \overline{\exp(\mathsf{L}(\mathbf{S}))}$, respectively. It is {\em locally exponential} if there exists
a $\operatorname{id}$-neighbourhood basis w.r.t.~$\mathbf{S}$ 
consisting of exponential subsets, both being detailed in~\cite{HofRupp97div,DHKS08}, refer also to footnote~\ref{f-note}.
Note that $\mathsf{MTO}(H_0,T)$ is actually exponential once the closure in Eq.~\eqref{eq:Lie-semigroup} is taken with respect to $\mathsf{GL}(\mathcal Z)$.
}. 
Actually, commutativity of $\circ_T$ for {\em fixed temperature}\,\footnote{
    Commutativity is lost for different temperatures, see the comment in App.~\appref{E}\,.
  }
  immediately implies that $\mathsf{MTO}(H_0,T)$ is {\em locally} as well as {\em (weakly) exponential}
  and its Lie wedge takes the special form of a Lie semialgebra, cf.~\cite[Ch.~1, Thm.~A \& Ch.~2, Ex.~1.6]{HofRupp97div}
  and \cite[Thm.~2.2]{DHKS08}.
So---here in the qubit case---the two notions of time-dependent and time-independent Markovianity coincide.
Taking the union of $\mathsf{MTO}(H_0,T)$ over {\em all temperatures} one obtains a semigroup
$\mathsf{MTO}(H_0)
:=\overline{\bigcup_{T\in(0,\infty)}\mathsf{MTO}(H_0,T)}$
which turns out to be both, weakly exponential and locally exponential as shown in App.~\appref{E}\,.\medskip

%

Finally let us actually visual{is}e how $\mathsf{MTO}(H_0,T)$ sits inside $\overline{\mathsf{TO}(H_0,T)}$ in~Fig.~\ref{fig_TO_qubit_Markov} 
(cf.~\cite[Fig.~2]{vomEnde22thermal}) for the case of a single qubit. Its upper left panel 
also illustrates how for small times $\mu_t\to 0$ the non-Markovian blue cone extends to the identity map
including an infinitesimal region {\em outside} the connected component of the (orange) Markovian Lie semigroup.

Now this becomes obvious by the shapes determined by the outer curves of the blue (orange) cone 
in the time direction $\mu_t$ which---up to rotational
symmetry---follow $\sqrt{1-\mu_t(1+\varepsilon)+\mu_t^2\varepsilon}$ 
 where $\mu_t\in[0,1]$ (resp.~$\sqrt{1-\mu_t(1+\varepsilon)}$ in the Markovian case where $\mu_t\in[0,1/(1+\varepsilon)]$).\medskip

In the lower panel of Fig.~\ref{fig_TO_qubit_Markov} the time dependence in $\mu_t=(1-e^{-tu(1+\varepsilon)})/(1+\varepsilon)$ formally 
leads to a time in multiples of the scaling factor $u > 0$ reading $t[u]=-{\ln(1-\mu_t(1+\varepsilon))}/{(1+\varepsilon)}$. 
Taking ``the'' preimage of the interval $[0,1]$ under this function yields a (partially) complex 
set: $t(\mu_t)$ is real (orange) for $0\leq\mu_t\leq 1/(1+\varepsilon)=:\mu_*$, i.e.~up to its pole at the  
Markovian limit $\mu_*$,
whence it becomes complex (blue) with constant imaginary part $\pm\pi/(1+\varepsilon)$
in the non-Markovian segment $\mu_*<\mu_t\leq 1$. 
%
\end{myexample}
\medskip
%


\begin{figure}[!ht]
\centering
\includegraphics[width=0.34\textwidth]{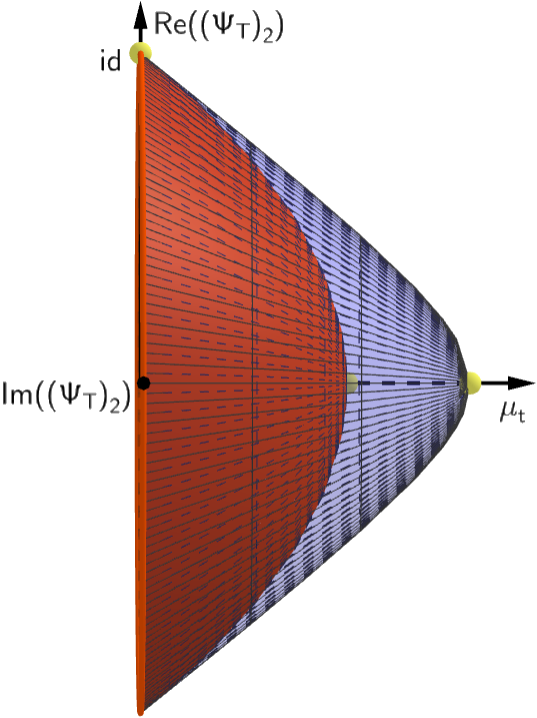}\qquad
\includegraphics[width=0.42\textwidth]{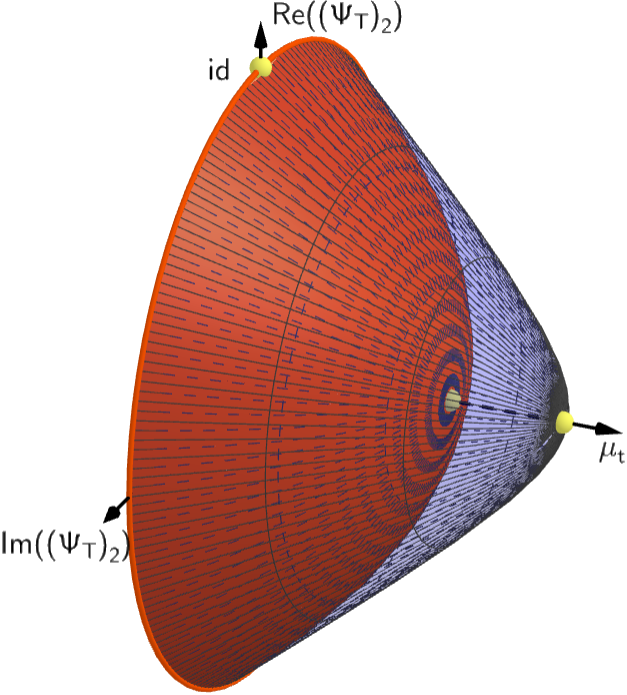}\\
\mbox{\phantom{.}\hspace{4.0cm}\hfill \includegraphics[width=0.64\textwidth]{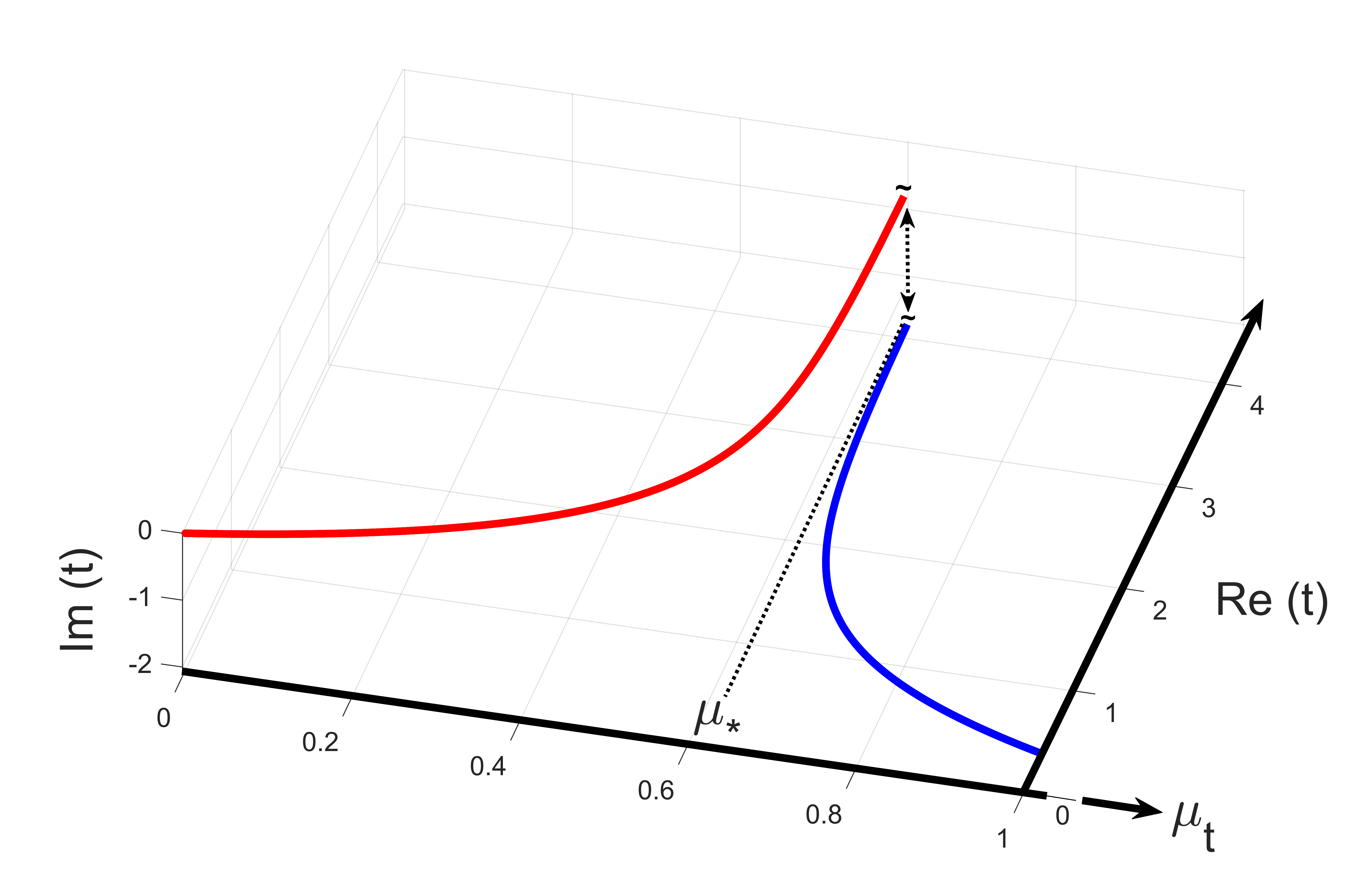}}
\caption{{\small{
(Colour online). Upper panels:
Two aspect angles for the graphs of $\Psi_T$ 
(with $\varepsilon=0.6$) when restricting the domain to 
$\overline{\mathsf{TO}(H_0,T)}$ (blue cone), and its 
Markovian counterpart $\mathsf{MTO}(H_0,T)$ (orange cone), respectively.
The yellow ``tip'' of the blue (orange) cone corresponds to the respective $\beta$-swaps (the ``thermal reset'' 
map $\rho\mapsto\rho_{\sf{Gibbs}}(H_0,T)$). ---
%
Lower panel: The time dependence in $\mu_t$ 
 formally leads to a time in multiples of the scaling factor $u > 0$ of $t[u]$, 
which is real in the Markovian segment $0\leq\mu_t\leq \mu_*$, i.e.~up to the pole at $\mu_*$
and complex in the non-Markovian segment $\mu_*<\mu_t\leq 1$ as detailed in the text.
The critical $\mu_*\in[\frac12,1]$ tends to one as $T\to 0^+$, thus illustrating
$\mathsf{MTO}(H_0,T) \to 
\overline{\mathsf{TO}(H_0,T)}$ in the zero-temperature limit.}
}}%
\label{fig_TO_qubit_Markov}
\end{figure}

\begin{myexample}
To sum up, the worked example discussed here also elucidated that the {\em single-qubit case} 
is special in as much as\\[-6.5mm]
\begin{itemize}
\item[(1)] For fixed temperature $T$ (and for the union over all $T$)  $\mathsf{MTO}(H_0,T)$ (resp.\ $\mathsf{MTO}(H_0)$)
	are generated by Lie semialgebras, respectively.
\item[(2)] So here time-independent and time-dependent Markovianity coincide.
\item[(3)] In the zero-temperature limit, Markovian thermal operations even exhaust all thermal operations: $\mathsf{MTO}(H_0,T) \to \overline{\mathsf{TO}(H_0,T)}$ as $T\to 0^+$.
\end{itemize}%
Due to statement (2) throughout this example we have been using the term ``Markovian" without
further specification.
\end{myexample}

\medskip

\noindent
{\small{%
To avoid misunderstandings, let us emphasise that 
in the {\em general case} (i.e.~beyond single qubits)
one has the following:
(1) For fixed temperature $T$ (and for the union over all $T$)  $\mathsf{MTO}(H_0,T)$ (resp.\ $\mathsf{MTO}(H_0)$)
	need {\em not} be generated by Lie semialgebras.
(2) Hence time-independent and time-dependent Markovianity need {\em not} coincide in general.
(3) In the zero-temperature limit, the relation between Markovian thermal operations $\mathsf{MTO}(H_0,T)$ and thermal operations $\overline{\mathsf{TO}(H_0,T)}$
is an open problem. }}

\color{black}
\subsection{The Role of $d$-Major{is}ation and the Associated Polytope}\label{sec:d_maj}
Let us return to the character{is}ation of state transitions via 
(enhanced) thermal operations
for the case of non-degenerate $H_0$ and quasi-classical initial states $\rho$,
i.e.~$[\rho,H_0]=0$.
In this case $\rho$ is diagonal in some basis which diagonal{is}es $H_0$ so w.l.o.g.~$\rho=\operatorname{diag}(y)$, $H_0=\operatorname{diag}(E_i)_{i=1}^n$.
It follows
\begin{equation*}
M_{H_0,T}(\operatorname{diag}(y))=\big\{ \operatorname{diag}(Ay): A\in\mathbb R^{n\times n}\text{ Gibbs-stochastic}\big\}\,,
\end{equation*}
where $\subseteq$ is due to Cor.~\ref{cor:thermo-diagonal} and $\supseteq$ is shown in \cite{Shiraishi20}
(cf.~also \cite{PolytopeDegen22}). Recall that $A\in\mathbb R^{n\times n}$ is called
\textit{Gibbs-stochastic} if $A$ is column-stochastic ($a_{ij}\geq 0$ for all $i,j$ 
and all columns of $A$ sum up to $1$) and the Gibbs-vector
$d:=(e^{-E_i/T})_{i=1}^n$ is a fixed point of $A$, that is $Ad=d$ \cite{Lostaglio19r}. In the
mathematics literature such a matrix is called \textit{$d$-stochastic} \cite[Ch.~14.B]{MarshallOlkin} which 
motivates defining
\begin{equation*}
M_d(y):=\big\{Ay:A\in\mathbb R^{n\times n}\ d\text{-stochastic}\,\big\}
\end{equation*}
as {\em all diagonals of states} one can thermodynamically generate (i.e.~by $\overline{\mathsf{TO}}$
and hence by $\mathsf{EnTO}$ or of course by $\mathsf{Gibbs}$)
starting from $\operatorname{diag}(y)$.
In other words the diagonal action of every Gibbs-preserving map (when projecting onto the diagonal) is a $d$-stochastic matrix,
and every $d$-stochastic matrix is the diagonal action of some element of $\overline{\mathsf{TO}}$.
The object $M_d(y)$ is known as \textit{$d$-major{is}ation polytope} \cite{vomEnde22}.
Note that in the high-temperature limit $d$ becomes the vector of equal weights $\frac1n(1,\ldots,1)^\top$ which recovers the concept of doubly stochastic
matrices, leading back to classical major{is}ation
\cite[Ch.~2.B]{MarshallOlkin}.

There are several ways to character{is}e the conditions for a $d$-stochastic matrix to exist so that it maps one real vector to another. 
Thus let us start with the most common one in the physics literature 
originally defined by Horodecki and Oppenheim \cite{Horodecki13}: given any vector of Gibbs weights $d\in\mathbb R^n$ with $d>0$ as well as some $y\in\mathbb R^n$, 
the \textit{thermomajor{is}ation curve of $y$} is defined to be the piecewise linear, continuous curve fully character{is}ed by the elbow points $\{ \big(\sum_{i=1}^j d_{\pi(i)},\sum_{i=1}^j y_{\pi(j)}\big) \}_{j=0}^n$. Here, $\pi\in S_n$ is any permutation such that $\frac{y_{\pi(1)}}{d_{\pi(1)}}\geq\ldots\geq\frac{y_{\pi(n)}}{d_{\pi(n)}}$.
Equivalently \cite[Rem.~7]{vomEnde22}, this map---which we denote by $\mathsf{th}_{d,y}:[0,\mathbbm e^{\top}d]\to \mathbb R$ where $\mathbbm e:=(1,\ldots,1)^\top$---satisfies
$$
\mathsf{th}_{d,y}(c)= \min_{\{i=1,\ldots,n\,:\,d_i>0\}} \Big(\Big(\sum_{j=1}^n\max\Big\{y_j-\frac{y_i}{d_i}d_j,0\Big\}\Big)+\frac{y_i}{d_i}c\Big)
$$
for all $c\in [0,\mathbbm e^{\top}d]$.
Together with \cite[Prop.~1]{vomEnde22} one thus gets:
\begin{proposition}\label{prop_char_dmaj}
Let $x,y,d\in\mathbb R^n$ with $d>0$ be given. The following statements are equivalent:
\begin{itemize}
\item[(i)] There exists a $d$-stochastic matrix $A\in\mathbb R^{n\times n}$ such that $Ay=x$. We denote this by $x\prec_d y$ and say that $x$ is \textnormal{$d$-major{is}ed} by $y$.
\item[(ii)] $\mathbbm e^{\top} x=\mathbbm e^{\top} y$ and $\mathsf{th}_{d,x}(c)\leq\mathsf{th}_{d,y}(c)$ for all $c\in[0,\mathbbm e^\top d]$.
\item[(iii)] $\mathbbm e^{\top}x=\mathbbm e^{\top}y$, and for all $j=1,\ldots,n-1 $
$$
\sum_{i=1}^jx_{\pi(i)}=\mathsf{th}_{d,x}\Big(\sum_{i=1}^j d_{\pi(i)}\Big)\leq\mathsf{th}_{d,y}\Big(\sum_{i=1}^j d_{\pi(i)}\Big)
$$
if $\pi\in S_n$ is any permutation such that $\frac{x_{\pi(1)}}{d_{\pi(1)}}\geq\ldots\geq \frac{x_{\pi(n)}}{d_{\pi(n)}}$.
\item [(iv)] $\mathbbm e^{\top} x=\mathbbm e^{\top} y$ and $\|d_ix-y_id\|_1\leq\|d_iy-y_id\|_1$ for all $i=1,\ldots,n $ where $\|\cdot\|_1$ is the usual vector $1$-norm.
\end{itemize}
Moreover, if $H_0\in{i}\mathfrak u(n)$, $T\in(0,\infty]$ are such that $H_0$ is non-degenerate and $d$ is the vector of Gibbs weights w.r.t.~$H_0$ and $T$, then the above conditions are equivalent to $\operatorname{diag}(y)$ thermomajor{is}ing $\operatorname{diag}(x)$ w.r.t.~$H_0$ and $T$.
\end{proposition}
\noindent Interestingly, above character{is}ations extend to the case where entries of the $d$-vector are allowed to be zero\footnote{
More precisely Prop.~\ref{prop_char_dmaj} continues to hold if the background temperature equals zero and if the system's ground state energy is non-degenerate \cite{PolytopeDegen22}.
}.
While (iv) $\Rightarrow$ (i) is usually proven indirectly via Farkas' Lemma \cite[Cor.~7.1.d]{Schrijver86}
there also exists a constructive algorithm which translates two comparable thermomajor{is}ation curves into a $d$-stochastic transition matrix \cite{Shiraishi20}. Notably this procedure simplifies considerably if the final state is an extreme point of the $d$-major{is}ation polytope induced by the initial state \cite{Alhambra19,PolytopeDegen22}.
Either way this leads to the following character{is}ation of the thermomajor{is}ation polytope \cite[Thm.~10]{vomEnde22}:
\begin{equation*}
M_d(y)=\big\{x\in\mathbb R^n:\mathbbm e^\top x=\mathbbm e^\top y\quad\wedge\quad\forall_{m\in\{0,1\}^n}\ m^\top x\leq  \mathsf{th}_{d,y}(m^\top d) \big\}
\end{equation*}
This finally justifies calling $M_d(y)$ a polytope because any bounded set which is the intersection of finitely many halfspaces is a convex polytope.
From a geometric point of view, the facets of the polytope $M_d(y)$ have universal orientation given by $m^\top$,
whereas their location defined by $\mathsf{th}_{d,y}(m^\top d)$ depends on $y$ and $d$. 
In particular, $M_d(y)$ being a convex polytope means it has a finite number of extreme points which in turn generate $M_d(y)$ via the convex hull, cf.~\cite[Ch.~3]{Gruenbaum03}.
Remarkably, these
can even be computed analytically \cite{Lostaglio18,Alhambra19,vomEnde22}:
given $d,y\in\mathbb R^n$ with $d\geq 0$ define the
extreme point map $E_{d,y}:S_n\to\mathbb R^n$ on the permutation group $S_n$ via
\begin{equation*}
E_{d,y}(\sigma):=\Big( \mathsf{th}_{d,y}\Big( \sum_{i=1}^{\sigma^{-1}(j)}d_{\sigma(i)} \Big)-\mathsf{th}_{d,y}\Big( \sum_{i=1}^{\sigma^{-1}(j)-1}d_{\sigma(i)} \Big) \Big)_{j=1}^n\,.
\end{equation*}

For computing the corners of $M_d(y)$, in practice one has to go 
through all permutations $\sigma\in S_n$, find the value of the 
thermomajor{is}ation curve for the inputs $d_{\sigma(1)},d_{\sigma(1)}+d_{\sigma(2)},
\ldots,\sum_{i=1}^{n-1}d_{\sigma(i)}$, and finally arrange the consecutive differences into a vector ordered according to
$\sigma$.
This procedure turns out to be quite simple as is nicely illustrated by a straightforward qutrit example, cf.~\cite[Example 1]{PolytopeDegen22}.
Interestingly, said example also shows that there exists an extreme point (corresponding to the trivial permutation) which is maximal in the polytope in the sense that it major{is}es all other (extreme) points \textit{classically}.
For physical systems, this turns out to be a general property which will be the key to upper bounding 
the 
reachable sets in~Sec.~\ref{sec:toy_model}
\begin{thm}\label{thm_max_corner}
Let $y,d\in\mathbb R^n$ with $d>0$ be given. If $y\geq 0$, then
$
x\prec E_{d,y}(\sigma)
$
for all $x\in M_d(y)$
where $\sigma\in S_n$ is any permutation which orders $d$ decreasingly, that is, $d_{\sigma(1)}\geq\ldots\geq d_{\sigma(n)}$.
Moreover $d$ and $\frac{E_{d,y}(\sigma)}{d}$ are ordered likewise\footnote{This means that there exists a permutation $\pi\in S_n$ such that $d_{\pi(1)}\geq\ldots\geq d_{\pi(n)}$ and $\frac{(E_{d,y}(\sigma))_{\pi(1)}}{d_{\pi(1)}}\geq \ldots\geq \frac{(E_{d,y}(\sigma))_{\pi(n)}}{d_{\pi(n)}}$.
},
and if $y>0$ then $E_{d,y}(\sigma)>0$.
\end{thm}
\noindent
A proof can be found in \cite[Thm.~16 \& Rem.~3]{vomEnde22}. Note that the assumption $y\geq 0$ is necessary as a simple counterexample shows \cite[Example 4]{vomEnde22}.
Either way, these extreme point techniques break down as soon as one goes beyond quasi-classical states, regardless of whether one considers the action of (enhanced) thermal operations or of general Gibbs-preserving maps.
The latter leads to the theory of $D$-matrix major{is}ation \cite{vomEnde20Dmaj}---where $D$ plays the role of the Gibbs state of the physical system---which is a general{is}ation of classical major{is}ation.
However, requiring maps to only preserve the Gibbs 
state dismisses intrinsic thermodynamic symmetries (Prop.~\ref{prop_1}~(iv))
which is why we will focus on (enhanced) thermal operations in the sequel.
\section{Thermal Operations in Quantum Control by Example} \label{sec:toy_model}
In this section we specify Markovian control problems \eqref{eq:control-diss_evolution} 
by explicitly using Markovian thermal operations as additional control resource 
in the sense $-\hGA\in \mathsf L(\overline{\mathsf{TO}(H_0,T)})$.
Given a Hamiltonian $H_0\in{i}\mathfrak u(n)$
with increasing eigenvalues $E_k$,
the corresponding equilibrium state resulting from coupling to a bath of
temperature $T$ is $\rho_{\sf Gibbs}=\frac{e^{-H_0/T}}{\operatorname{tr}(e^{-H_0/T})} \in\pos n$
where the eigenvalues of the Gibbs state are collected in
the \textit{Gibbs vector}
\begin{equation}\label{eq:gibbs_vec}
 d:=\frac{(e^{-E_k/T})_{k=1}^n}{\sum_{k=1}^n e^{-E_k/T}}\in\Delta^{n-1}\,.
\end{equation}
Here, $\Delta^{n-1}:=\{ x\in\mathbb R_+^n\,|\, {\textstyle\sum}_{i=1}^nx_i=1\}$ is the collection of all probability vectors, called standard simplex.
 As shown in \cite{CDC19}, $\rho_{\sf Gibbs}$ can then be obtained as 
the unique fixed point of \eqref{eq:diss_evolution} when choosing the \gks-terms as
\begin{align}
V_1&=\sigma_+^d:=\sum_{k=1}^{n-1}\sqrt{k(n-k)}\cos(\theta_k)\,|k\rangle\langle k+1|\label{eq:sigma+}\\
V_2&=\sigma_-^d:=\sum_{k=1}^{n-1}\sqrt{k(n-k)}\sin(\theta_k)\,|k+1\rangle\langle k|\,,\label{eq:sigma-}
\end{align}
where $|k\rangle$ is ``the'' eigenvector of $H_0$ to the eigenvalue $E_k$ and
\begin{equation}\label{eq:thermal_angle}
\theta_k:=\arccos\Big(\big({1+\frac{d_{k+1}}{d_k}}\big)^{-\frac{1}{2}} \Big)\in\Big(0,\frac{\pi}{4}\Big].
\end{equation}
Assuming non-degenerate spectrum of $H_0$, the limiting 
cases of zero and infinite temperature can be included:
\begin{itemize}
\item Taking the limit $T\to 0^+$ yields $d=(1,0,\ldots,0)^\top$, $\theta_k\to\arccos(1)=0$ for all $k$, as well as
$\sigma_+^d\to\sigma_+=\sum_{k=1}^{n-1}\sqrt{k(n-k)}|k\rangle\langle k+1|$ and $\sigma_-^d\to 0$.
\item The limit $T\to\infty$ yields $d=\frac1n(1,\ldots,1)^\top$ so $\theta_k=\frac{\pi}{4}$, i.e.~$\cos(\theta_k)=\sin(\theta_k)=\frac{1}{\sqrt{2}}$.
\end{itemize}
Confining ourselves to $\sigma_+^d$ and $\sigma_-^d$ with their non-zero entries on the first off-diagonals
is in accordance with the common dipolar selection rules allowing for ``one-quantum transitions'' 
(as governed by Wigner's $3j$-symbol)~\cite[p.~185 ff.]{Zare88}.
This is further motivated by the fact that for spin systems these generators yield dynamics within the thermal operations:
\begin{corollary}\label{coro_ladder_ops_are_TO}
Let $\Delta E>0$, $n\in\mathbb N$, and $T>0$ be given. Defining the system's Hamiltonian $H_0:=\operatorname{diag}(0,\ldots,n-1)\cdot\Delta E$ as well as $\hGA_d:=\hGA_{\sigma_+^d}+\hGA_{\sigma_-^d}$ the generator induced by Eqs.~\eqref{eq:sigma+} \& \eqref{eq:sigma-} via Eq.~\eqref{eq:lindblad_V}, one finds the inclusion
$
( e^{ -t(i\operatorname{ad}_{H}+\hGA_d) } )_{t\geq 0}\subseteq\overline{\mathsf{TO}(H_0,T)}
$
for all $H\in{i}\mathfrak u(n)$ such that $[H,H_0]=0$.
\end{corollary}
\begin{proof}
We apply
Thm.~\ref{thm:markov_generator}
with $m=2$, $H_B=\operatorname{diag}(0,1)\cdot\Delta E$, and
$$
H_{\mathsf{tot}}:=\sum_{j=1}^{n-1}\sqrt{\frac{j(n-j)}{1+e^{-\Delta E/T}}}\big(|e_j\rangle\langle e_{j+1}|\otimes|e_2\rangle\langle e_1|+|e_{j+1}\rangle\langle e_j|\otimes|e_1\rangle\langle e_2|  \big)\,.
$$
We may do so because $H_{\mathsf{tot}}$ commutes with $H_0\otimes\mathbbm1+\mathbbm1\otimes H_B$ as can be seen easily.
Now, said theorem guarantees that the dynamical semigroup generated by $-{i}\operatorname{ad}_H-\sum_{j,\ell=1}^2\hGA_{V_{j\ell}}$ with
$V_{j\ell}=\sqrt{e^{-E_j '/T}}\operatorname{tr}_{|g_j\rangle\langle g_\ell|}(H_{\mathsf{tot}})$
is in $\overline{\mathsf{TO}(H_0,T)}$. But $V_{11}=V_{22}=0$ and
\begin{align*}
V_{12}=\operatorname{tr}_{|e_1\rangle\langle e_2|}(H_{\mathsf{tot}})&= \sum_{j=1}^{n-1} \sqrt{\frac{j(n-j)}{1+e^{-\Delta E/T}}}|e_{j}\rangle\langle e_{j+1}| \\
V_{21}=\sqrt{e^{-\Delta E/T}}\operatorname{tr}_{|e_2\rangle\langle e_1|}(H_{\mathsf{tot}})&= \sum_{j=1}^{n-1} \sqrt{\frac{j(n-j)e^{-\Delta E/T}}{1+e^{-\Delta E/T}}}|e_{j+1}\rangle\langle e_{j}| \,,
\end{align*}
that is, $V_{12}=\sigma_+^d$ and $V_{21}=\sigma_-^d$. This concludes the proof.
\end{proof}

\begin{remark}\label{rem_equidist_necessary}
The assumption in Cor.~\ref{coro_ladder_ops_are_TO} that $H_0$ has equidistant eigenvalues is necessary 
since otherwise
the generator $\hGA_d$
building simply on Eqs.~\eqref{eq:sigma+} and \eqref{eq:sigma-}
is no longer in 
$\mathsf{L}(\mathsf{EnTO}(H_0,T))$ for any $T>0$ (due to $[\hGA_d,\operatorname{ad}_{H_0}]\neq 0$)\footnote{
To see this, let $n\geq 3$ and let $H_0=\operatorname{diag}(E_1,\ldots,E_n)\in {i}\mathfrak u(n)$, $E_1\leq\ldots\leq E_n$ such that $E_i-E_{i+1}\neq E_j-E_{j+1}$ for some $i,j\in\{1,\ldots,n-1\}$, $i\neq j$ (i.e.~$H_0$ does not have equidistant eigenvalues).
A straightforward computation shows $\langle e_{i+1},[\hGA_d,\operatorname{ad}_{H_0}](|e_i\rangle\langle e_j|)e_{j+1}\rangle\neq 0$.
} so it in particular cannot be in $\mathsf{L}(\overline{\mathsf{TO}(H_0,T)})$ anymore.
--- In the non-equidistant case one just has to replace the simple uniform $\sigma_+^d$ of Eq.~\eqref{eq:sigma+} 
accordingly by a family $\sigma^d_{+,1},\ldots,\sigma^d_{+,l}$ ($l\geq 2$) such that the non-zero entries of $V_{+,l}$ correspond to the neighbouring levels of $H_0$ of a certain energy distance (and similarly for $\sigma_-^d$) to ensure the resulting $-\hGA$ is again in $\mathsf L(\overline{\mathsf{TO}(H_0,T)})$.
\end{remark}

Considering the standard control system \eqref{eq:control-diss_evolution} with dissipator $\hGA_d$, Cor.~\ref{coro_ladder_ops_are_TO} shows that if all coherent controls are compatible with the thermodynamic framework from Sec.~\ref{sec_thermo_markov}, i.e.~$[H_0,H_j]=0$ for all $j=1,\ldots,m$ (cf.~Lem.~\ref{lemma_edge}), then the reachable set of this control problem is automatically upper bounded by the future thermal cone defined by Eq.~\eqref{eq:def_M_H_T}.
Now the richness thermodynamic control systems have to offer comes from the interplay between thermodynamic dissipation (i.e.~$-\hGA\in\mathsf L(\overline{\mathsf{TO}(H_0,T)})$) and general unitary controls which become an asset due to \textit{not} stabil{is}ing $H_0$.
However, this overlap of different categories comes at the expense of making it more difficult to study.

Recalling from Sec.~\ref{sec_thermo_markov}~that
diagonal elements evolve separately from off-diagonal ones
under (enhanced) thermal operations,
in the next section we study
a modified version of control system \eqref{eq:control-diss_evolution}
(with $\hGA=\hGA_d$ from Cor.~\ref{coro_ladder_ops_are_TO})
focussing on \textit{diagonal states} represented by the standard simplex $\Delta^{n-1}$.

%
%
\subsection{Toy Models by Diagonal States} \label{sec_toy_model_sub1}

The idea for reducing
reachability problems of (finite-dimensional) Markovian open
quantum systems 
to hybrid control systems on the standard
simplex of $\mathbb R^n$
will 
be to 
\textit{only} include unitary controls in the model 
which do not mix the diagonal and the off-diagonal of any state $\rho(t)$---
as the same holds for all thermal operations (Cor.~\ref{cor:thermo-diagonal}). 
This means we have to restrict the coherent controls to generators of (unitary channels induced by) permutation matrices.

Recalling the bilinear control system \eqref{eq:bilin},
now we confine the discussion to $\mathbf{x}(t)$ denoting the vector of diagonal elements of $\rho(t)$ in ``the'' eigenbasis of $H_0$.
We address a scenario with coherent controls $\{B_j\}_{j=1}^m$ and a {\em bang-bang switchable} 
dissipator $B_{0}\in\mathsf L(\overline{\mathsf{TO}(H_0,T)})$ as motivated by recent experimental progress
\cite{Mart09,Mart13,Mart14,McDermott_TunDissip_2019} 
and used in~\cite{BSH16}.
%
The controls of the toy model shall amount to permutation matrices acting instantaneously on the entries
of $x(t)$ and a continuous-time one-parameter semigroup $(e^{-tB})_{t\in\mathbb R^+}$ of stochastic maps 
with a unique fixed point $d$ in $\Delta^{n-1}$. As $(e^{-tB})_{t\in\mathbb R^+}$ results from the 
restriction of the bang-bang switchable dissipator $B_{0}$, with abuse of notation we will denote its 
infinitesimal generator by $B$. The
 ``\/{\em equilibrium state}\/'' $d$ is defined in \eqref{eq:gibbs_vec} by system parameters and
the absolute temperature $T\geq 0$ of an external bath.

These stipulations suggest the following hybrid/impulsive scenario 
to define the {\em toy model} $\Lambda_B$ on $\Delta^{n-1} \subset \mathbb R^n$ by
\begin{equation}\label{eq:control-simplex_evolution}
\begin{split}
&\dot{x}(t)  = -B x(t)\,,\quad x(t_k) = \pi_k x_k\,, \quad t \in [t_k,t_{k+1})\,,\\
& x_0  \in \Delta^{n-1}\,, \quad x_{k+1} = e^{-(t_{k+1}-t_k)B}x(t_{k})\,, \quad k\geq 0\,.
\end{split}
\end{equation}
Furthermore, $0 =: t_0 \leq t_1 \leq t_2 \leq \dots$ is an arbitrary switching sequence and $\pi_k$ are arbitrary 
permutation matrices. Both the switching points and the permutation matrices are regarded as controls 
for \eqref{eq:control-simplex_evolution}. For simplicity, we assume that the switching points do not 
accumulate on finite intervals.
For more details on hybrid/impulsive control systems see,
e.g.,~\cite{book_impulsive89,Leela1991,book_HybridSytems96}.
The reachable sets of \eqref{eq:control-simplex_evolution} 
\begin{equation*}
\reach_{\Lambda_B}(x_0) := \{x(t) \,|\,
\text{$x(\cdot)$ is a solution of \eqref{eq:control-simplex_evolution}, $t \geq 0$}\}
\end{equation*}
allow for the character{is}ation
$
\mathfrak{reach}_{\Lambda_B}(x_0) =  {\mathcal S}_{\Lambda_B} x_0\,,
$
where ${\mathcal S}_{\Lambda_B} \subseteq \GL(n,\mathbb R)$ is the ($1$-norm-)contraction semigroup generated 
by $(e^{-tB})_{t\in\mathbb R_+}$ and the set of all permutation matrices
$\pi$.
\paragraph*{Recent Results.}
For the scenario just specified, the state-of-the-art \cite{CDC19} can be sketched as follows.
Take the $n$-level toy model 
where the 
infinitesimal generator results from 
coupling to a bath of temperature $T\in[0,\infty]$, i.e.~one has $B=B(\hGA_d)$ with $\hGA_d$ from Cor.~\ref{coro_ladder_ops_are_TO}.
We denote this particular toy model by $\Lambda_d:=\Lambda_{B(\hGA_d)}$.

\begin{thm}\label{thm_1}
The closure of the reachable set of any initial vector $x_0 \in \Delta^{n-1}$ under the dynamics of
$\Lambda_{e_1}$ exhausts the full standard simplex, i.e.
$$
\overline{\mathfrak{reach}_{\Lambda_{e_1}}(x_0)}=\Delta^{n-1}\,.
$$
\end{thm}\vspace{-3mm}

Moving from a single $n$-level system (qu\/{\em d}\/it) with $x_0\in\Delta^{n-1}$ to a tensor product 
of $m$ such $n$-level systems gives $x_0 \in \Delta^{n^m-1} \subset ({\mathbb R}^n)^{\otimes m}$. 
If the bath of temperature $T=0$ is coupled to just one (say the last) of the $m$ qu\/{\em d}\/its,
$\hGA$ is generated by $V :=\mathbbm1_{n^{m-1}}\otimes\sigma_+$ 
and one obtains the following general{is}ation.

\begin{thm}\label{thm_2}
The statement of Thm.~\ref{thm_1} holds analogously for all \mbox{$m$-qudit} states
$x_0\in\Delta^{n^m-1}$.
\end{thm}

In a first step to general{is}e the findings from the extreme case $T=0$ to $T\in(0,\infty]$ we found that general statements about the reachable set of $\Lambda_d$ can only be made under the following assumption
\vspace{-2mm}
\begin{center}
\textbf{Assumption A}: $H_0$ has equidistant energy eigenvalues.
\end{center}
\vspace{-2mm}
In this case we obtained:
\begin{thm}\label{thm_3}
Assuming \textbf{A},
the reachable set of the thermal state $d$ under the dynamics of
the toy model $\Lambda_d$ satisfies 
$
\mathfrak{reach}_{\Lambda_d}(d)\subseteq \lbrace x\in\Delta^{n-1}\,|\, x\prec d\rbrace\,,
$
where \/`$\prec$\/' refers to 
classical major{is}ation.
In particular, this is the smallest convex upper bound for the reachable set one can find.
\end{thm}
\noindent There are counterexamples to Thm.~\ref{thm_3}
as soon as $H_0$ no longer has equidistant eigenvalues
\cite[Example 3]{CDC19}.
This is because assumption {\bf A} is necessary for the dynamics of $\Lambda_d$ to be thermal operations 
(see the generalising Rem.~\ref{rem_equidist_necessary}).

The recent toy-model results of \cite{CDC19} thus extend the diagonal part of the 
qu\/{\em b}\/it picture (previously analysed in \cite{BSH16}) to $n$-level systems, and 
even more generally to systems of $m$ qu\/{\em d}\/its. --- Next we explore further general{is}ations
to non-zero temperatures, e.g., by allowing for general initial states $x_0$ instead of 
the thermal state $d$ in Thm.~\ref{thm_3}.
%
\paragraph*{Generalisations.}
To general{is}e previous reachability characterisations 
we use deeper results on $d$-major{is}ation (Sec.~\ref{sec:d_maj}).
For the toy-model dynamics one gets:
\begin{itemize}\vspace{-2mm}
\item[(1)] $e^{-tB}x_0 \in M_d (x_0)$ for all $t\geq 0$;
\item[(2)] $M_d (x_0)$ is a convex subset within the simplex $\Delta^{n-1}$,
\vspace{-2mm}%
\end{itemize}
which means the \/{\em dissipative time evolution}\/ of any $x_0$ remains within the convex set of states 
$d$-majorised by $x_0$. 
Beyond pure dissipative evolution the toy model also allows for permutations
$\pi$, so one naturally obtains
$
\reach_{\Lambda_d}(x_0)=\reach_{\Lambda_d}(\pi(x_0))$
for all $\pi\in S_n$.
Clearly, the simplex region $M_d (x_0)$ intertwines overall permutations $\pi$
(in the symmetric group $S_n$)
in the sense
$\pi\,M_d (x_0)=M_{\pi(d)} (\pi(x_0))$.
For the maximally mixed state ($d\simeq\mathbbm e$) it boils down to permutation invariance
under classical majorisation
$
\pi\,M_{\mathbbm e} (x_0)=M_{\mathbbm e} (\pi(x_0))=M_{\mathbbm e} (x_0)
$.
This
immediately entails a first generalisation: \vspace{-2mm}
\begin{corollary}[general{is}ing Thm.~\ref{thm_3}]\label{coro:thm_3g}
Assuming \textbf{A} those 
initial states $ x_{0}$ classically majorised by $d$ (i.e.~$ x_{0}\in M_{\mathbbm e} (d)$)
remain within $M_{\mathbbm e} (d)$ under the dynamics of the toy model $\Lambda_d$.
In other words
$
\overline{\mathfrak{reach}_{\Lambda_d}( x_{0})}\subseteq M_{\mathbbm e} (d).
$
\end{corollary}

%
%
\noindent%
In turn, this is but a special case of the following generalisation to arbitrary initial states building
on some deeper results on $d$-major{is}ation (Sec.~\ref{sec:d_maj}):

\begin{thm}\label{thm:general_dmaj_bound}
Invoke assumption \textbf{A}. For the toy model $\Lambda_d$ with
Gibbs state $d$
corresponding to coupling to a bath of temperature $T\in(0,\infty]$, 
the reachable set of any $x_0\in\Delta^{n-1}$ is included
in the following convex hull:\vspace{-2mm}
\begin{equation}\label{eq:upper_bound_arbitrary}
\overline{\reach_{\Lambda_d}(x_0)} \subseteq \conv\big\{\pi(z)\,|\,\pi\in S_n\big\} = M_{\mathbbm e} (z)\,.
\end{equation}\vspace{-2mm}%
Here, $z$ is any element from the ``ordered past cone''\footnote{
By definition \cite[Def.~3]{Oliveira22} the (unordered) past cone of a vector $x_0$ is the set of all states \textit{starting from which} one can
generate $x_0$ via doubly-stochastic matrices.
}
\begin{equation}\label{eq:outwards_maj}
\big\{z\in\Delta^{n-1}\,:\,x_0\prec z\ \wedge\ d\text{ and }\tfrac{z}{d}\text{ are ordered likewise}\big\}\vspace{-2mm}
\end{equation}
which, most importantly, contains an element $z>0$ whenever $x_0>0$.
\end{thm}
%
\noindent The idea of the proof of Eq.~\eqref{eq:upper_bound_arbitrary}
is to show that the
vector field driving the
dynamics of $\Lambda_d$ points \/{\em inside}\/ the classical majorisation polytope $M_{\mathbbm e} (z)$ at
each of its $n\/!$ extreme points $\pi(z)$ with $\pi\in S_n$, see also Fig.~\ref{fig:cdc-mtns}.
Finally, if $x_0>0$ then the existence of a vector $z>0$ in Eq.~\eqref{eq:outwards_maj} is due to Thm.~\ref{thm_max_corner}
(as detailed in the first author's PhD thesis \cite[Thm.~5.1.15]{vE_PhD_2020}).

We emphas{is}e that---while Thm.~\ref{thm_3} becomes trivial in the limit 
$T\to\infty$---Thm.~\ref{thm:general_dmaj_bound} reproduces the known result that in the high-temperature limit the reachable set for unital dynamics is upper bounded by all states classically major{is}ed by the initial state.
%

\begin{figure}[!ht]
\vspace{-21mm}
\mbox{\hspace{-6mm}\raisebox{2mm}{\includegraphics[width=.50\linewidth]{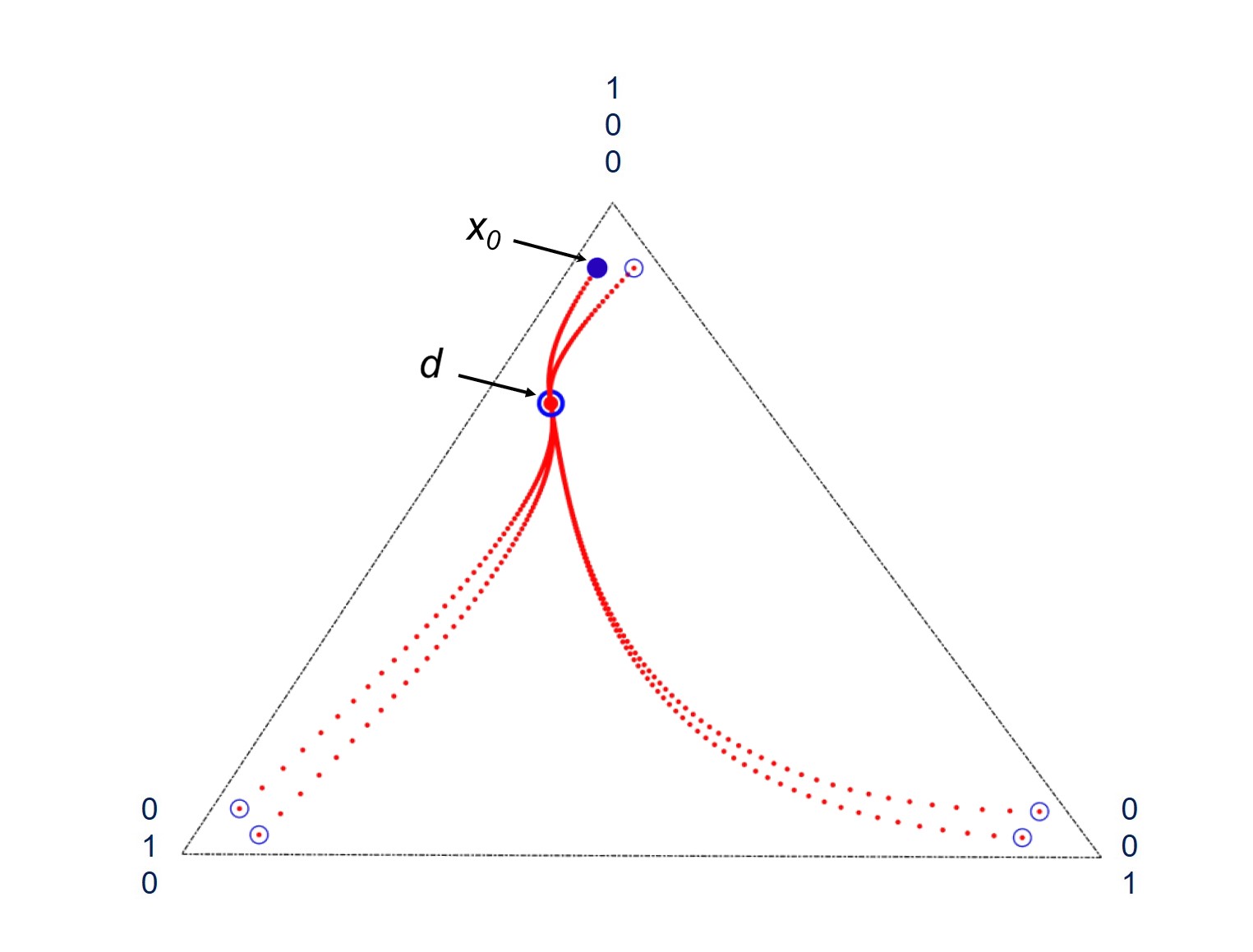}}}
\mbox{\hspace{-.5mm}\raisebox{3.4mm}{\includegraphics[width=.54\linewidth]{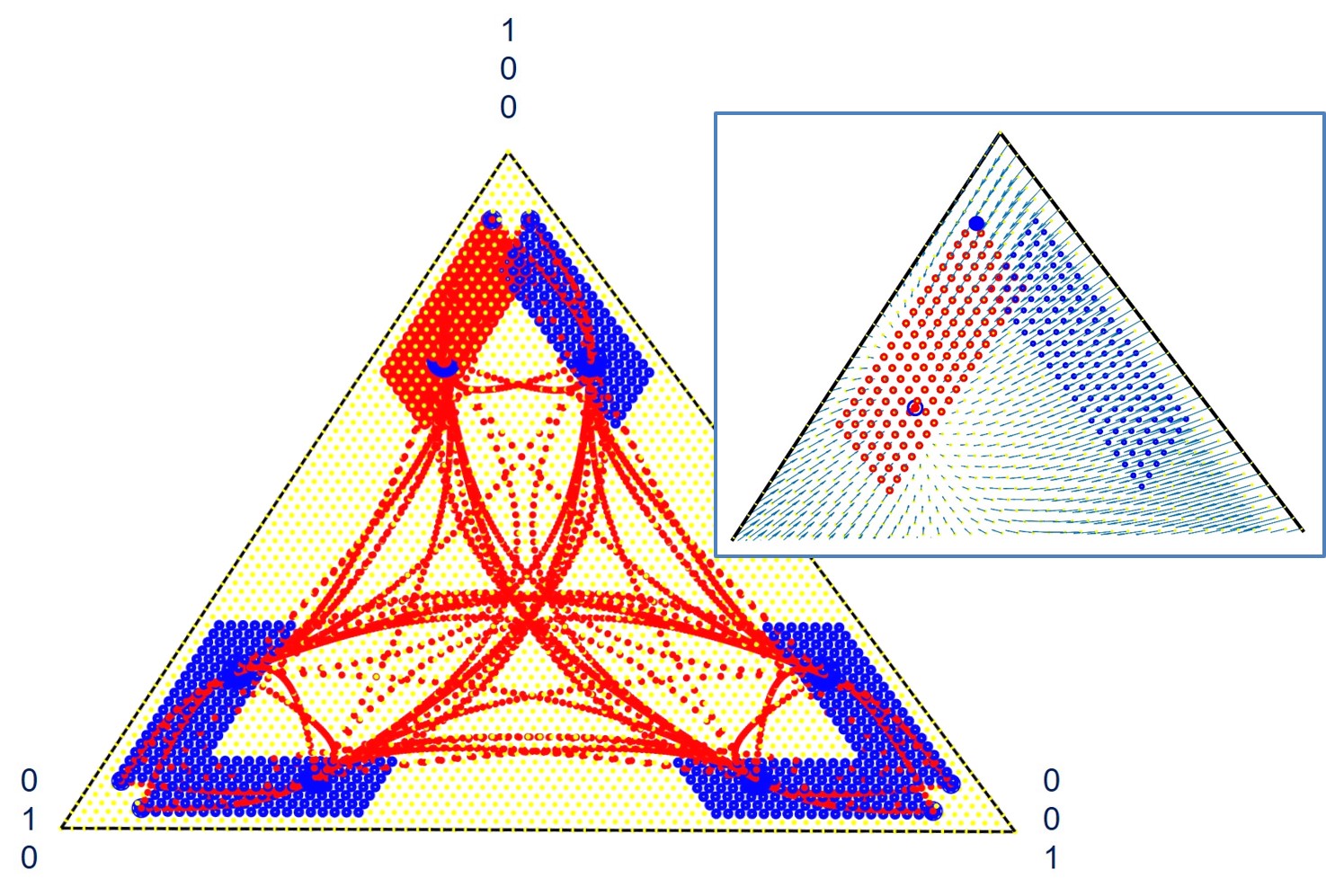}}}\\[-3mm]
\mbox{\hspace{-6mm}\raisebox{4mm}{\includegraphics[width=.55\linewidth]{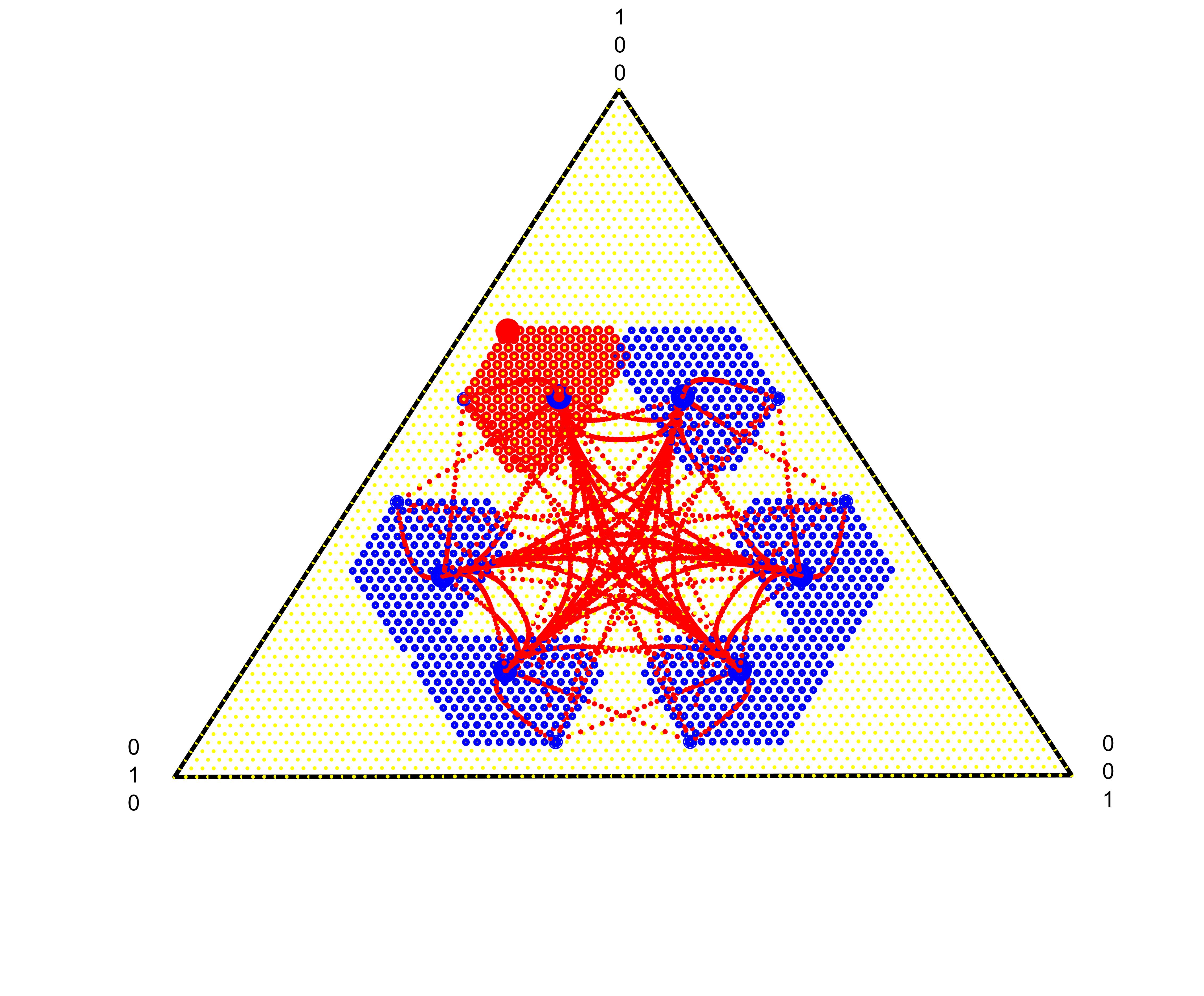}\hspace{-5mm}\raisebox{5.2mm}{\includegraphics[width=.55\linewidth]{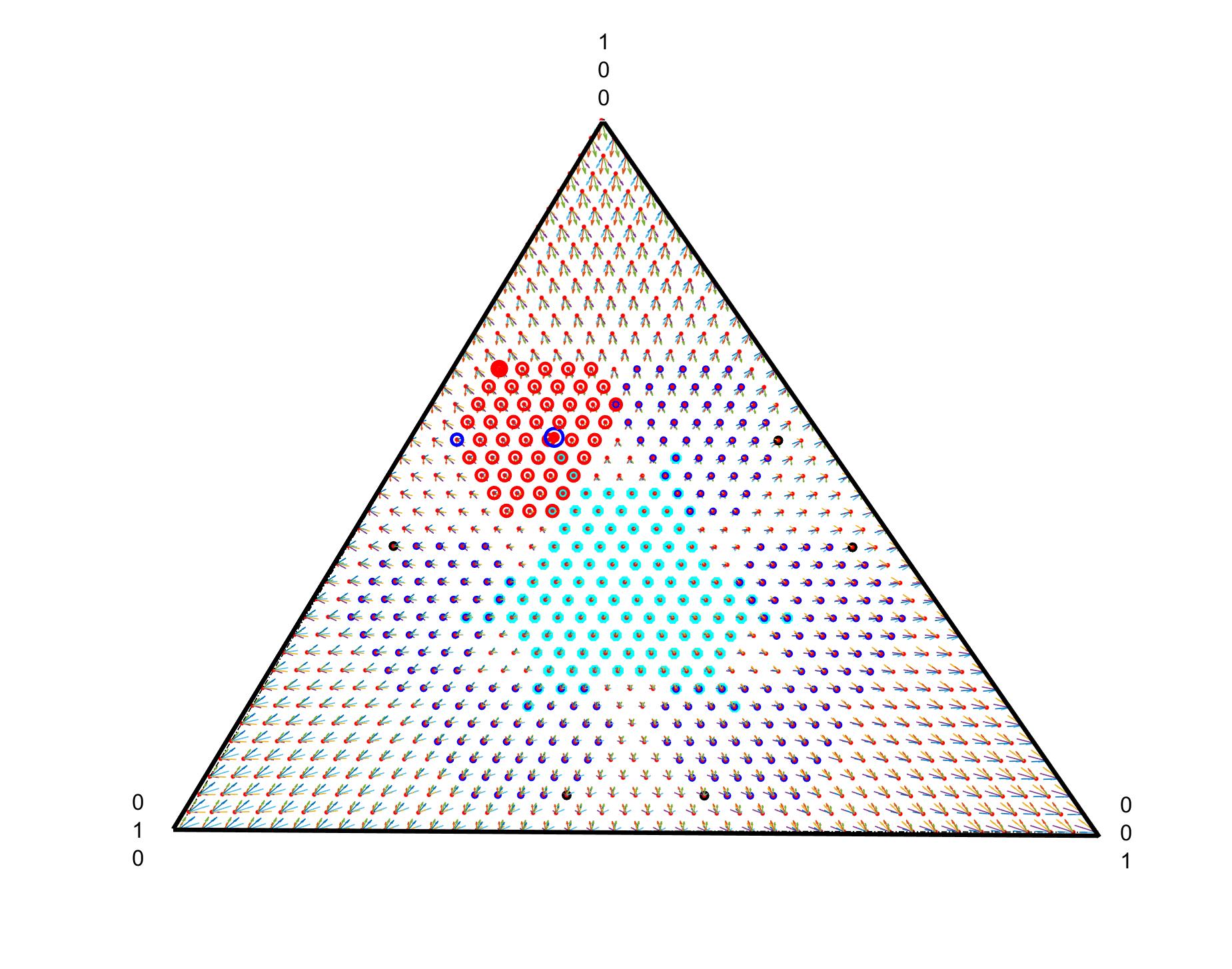}}}}\\[-15mm]
%
\caption{\label{fig:cdc-mtns}%
{\small{(Colour online). %
Upper left: Evolutions of initial $x_0=(0.9, 0.07, 0.03)^\top$ and permutations $\pi(x_0)$
under $\hGA_d$ with $V_1,V_2$, $\theta=\tfrac{\pi}{6}$ 
of Eqs.~\eqref{eq:sigma+}-\eqref{eq:thermal_angle} drive to fixed point $d$\/; 
upper right panel includes all permutations of trajectories starting with
permutations  of $d$, i.e.~$x_0=\pi(d)$; the red region shows states $d$-majorised by $x_0$, 
blue regions are their permutations; the convex hull over red and blue regions contains 
entire reachable set $\reach_{\Lambda_d}(x_0)$; 
inset gives the vector field to the dissipative part of the dynamics. 
\newline
Lower left: For $\theta=\tfrac{\pi}{5}$ in \eqref{eq:thermal_angle}, as generically, 
the extreme point 
$z=\left(0.65, 0.30, 0.05\right)^\top$ in red
differs from $x_0=\left(0.55, 0.40, 0.05\right)^\top$ 
as well as from $d=\left(0.55, 0.29, 0.16\right)^\top$. The lower right shows the vector fields
under the full dynamics $\Lambda_d$ of dissipation and permutation control with stabilisable points 
(cp.~Sec.~\ref{sec:Qutrit-Results}, Fig.~\ref{fig:reach-stab}) in turkish blue.
}}
}
\end{figure}

\begin{remark}\label{rem_tightest_bound}
One can show that 
the set \eqref{eq:outwards_maj} 
of possible extreme points from Thm.~\ref{thm:general_dmaj_bound} used for an upper bound
forms a convex polytope.
Thus by means of convex optim{is}ation one can find an ``optimal'' major{is}ation bound
in the sense that $z$ is closest to the fixed point of the dynamics, i.e.\footnote{
Of course the $1$-norm can be replaced by any other function $f:\Delta^{n-1}\to\mathbb R_+$ of interest.
}~$\|z-d\|_1=\min_{y\in\eqref{eq:outwards_maj}}\|y-d\|_1$.
While this is unambiguous if $x_0\prec d$ or if $x_0$ is in its own ordered past cone\footnote{\label{footnote_optimal_bound}
If $x_0\prec d$, then the ``optimal'' $z$ is $d$, in the sense that $\operatorname{conv}(\overline{\mathfrak{reach}_{\Lambda_d}( x_{0})})=M_{\mathbbm e}(d)$:
on the one hand $\operatorname{conv}(\overline{\mathfrak{reach}_{\Lambda_d}( x_{0})})\subseteq \operatorname{conv}(M_{\mathbbm e} (d))\subseteq M_{\mathbbm e} (d)$ by 
Thm.~\ref{thm:general_dmaj_bound} because $d$ is in the ordered past cone of $x_0$; also
$
M_{\mathbbm e} (d)=\conv(\{\pi d:\pi\in S_n\})\subseteq \operatorname{conv}(\overline{\mathfrak{reach}_{\Lambda_d}( x_{0})})
$.
Similarly one sees that if $x_0$ is in its own ordered past cone, i.e.~$d$ and $\frac{x_0}{d}$ are ordered likewise, then the ``optimal'' $z$ (as def.~above) is $x_0$ itself.
}, note that for general $x_0\in\Delta^{n-1}$ the convex hull of $\overline{\reach_{\Lambda_d}(x_0)}$ need not be a major{is}ation polytope anymore.
\end{remark}


Fig.~\ref{fig:cdc-mtns} illustrates these general findings for the special case 
of three-level systems (again assuming the drift term $H_0$ has equidistant eigenvalues).


\subsection{Explicit Results and Examples for Qutrits}\label{sec:Qutrit-Results}

So far, we have given upper bounds for reachable sets of the toy model.
In this section we will explicitly determine the shape of the reachable set and of the set of stabil{is}able states for the three-dimensional case $d\in\R^3$, $d>0$. 
For this we first introduce some general notions.

It pays off to approach the toy model~\eqref{eq:control-simplex_evolution} from a different, 
but equivalent\footnote{
The systems are equivalent in the sense that every solution of one system has a corresponding solution in the other system differing only by some (time-dependent) permutation. Note however that we allow more general controls in the differential inclusion, so that this equivalence is only approximate in general.},
perspective: instead of letting the permutations act on the states, leading to discontinuous paths, 
we let the permutations act on the drift vector field, leading to the following differential inclusion\footnote{
By abuse of notation, $\pi\in S_n$ also denotes the induced permutation matrix.},
where---in analogy to $\reach_B(x)$---we write $\derv(x)$ for the \emph{set of achievable derivatives at $x$}:
\begin{equation}\label{eq:diff-incl-toy-model}
\dot x(t)\in\conv(\derv(x(t))), \quad \derv(x):=\{-\pi B\pi^{-1}x:\pi\in S_n\},
\end{equation}
cf.~\cite{Smirnov02}\footnote{
In particular Thm.~2.3 therein shows the equivalence of control systems and the corresponding differential inclusions.
Note that taking the convex hull leads to a relaxation of the differential inclusion, which is still approximately equivalent to the original control system, see~\cite[Ch.~2.4, Thm.~2]{Aubin84}.}
for an introduction to this topic.
Many ideas work for any matrix $-B$ which generates a one-parameter semigroup of stochastic matrices and has unique fixed point $d$, but for some results we will restrict $B$ to the case where the generator is of the form given in Eqs.~\eqref{eq:sigma+} \& \eqref{eq:sigma-} and the corresponding Hamiltonian has equidistant energies. This ensures that we obtain sensible formulas, and it is physically motivated, see Rem.~\ref{rem_equidist_necessary}. As above we call this Assumption~\textbf{A}.

\paragraph*{Stabil{is}able States.} 
The set of \emph{stabil{is}able states} $\mathfrak{stab}_B$ is defined to be all $x\in\Delta^{n-1}$ such that $0\in\mathrm{conv}(\derv(x))$.
Intuitively, these are the points in $\Delta^{n-1}$ that, when taken as starting point, one can remain arbitrarily close to. 
More precisely we have the following result:

\begin{lemma}
A state $x_0\in\Delta^{n-1}$ is stabil{is}able if and only if for every $\varepsilon>0$ and $\tau>0$ there is a solution $x:[0,\tau]\to\Delta^{n-1}$
to~\eqref{eq:diff-incl-toy-model} with $x(0)=x_0$ which remains inside of the $\varepsilon$-ball $B_\varepsilon(x_0)\cap\Delta^{n-1}$.
\end{lemma} 

\begin{proof}
If $x_0$ is stabil{is}able, then the constant path $x\equiv x_0$ is a solution to~\eqref{eq:diff-incl-toy-model}.
Conversely, assume that $x_0$ is not stabil{is}able. Then, by continuity, there is some $\delta>0$ and some linear functional $\beta$ on $\mathbb R^n$ such that $\beta$ is less than $-\delta$ on $\derv(y)$ for all $y$ in some neighborhood of $x_0$. Hence there is some time $\tau>0$ where any solution must leave $B_\varepsilon(x_0)$ for some $\varepsilon$ small enough.
\end{proof}

\begin{remark}
It is possible to define a control system on the simplex $\Delta^{n-1}$ similar to the toy model (i.e.~by projecting~\eqref{eq:control-diss_evolution} onto ``the'' diagonal) but allowing for the full unitary control of the system given by Eq.~\eqref{eq:control-diss_evolution}.
In this case there is a characterisation of stabil{is}ability in basic Lie-algebraic terms:
Every point in the simplex $\Delta^{n-1}$ is stabil{is}able if and only if all \gks-terms $V_k$ can be simultaneously (upper) triangular{is}ed~\cite{MEDS23}. 
By Lie's Theorem, this is equivalent to the $V_k$ generating a solvable Lie algebra.
Also be aware of the special cases if all $V_k$ commute, or one just has a single $V_k$,
such as $\sigma_+^d$ of Eq.~\eqref{eq:sigma+} in the case $T=0$ (i.e.~$\theta_k=0$
in \eqref{eq:thermal_angle}).
As soon as $T>0$, however, the situation gets more involved, as the qutrit example below shows.
\end{remark}

\noindent
If zero is not contained in the convex hull of achievable derivatives at $x$, then there must exist some linear functional $\alpha$ on $\R^n$ which is negative on $\derv(x)$. 
Note that while $\alpha$ lives on $\R^n$, only the part parallel to the simplex $\Delta^{n-1}$ matters.
Based on this observation, the idea is to consider the ``permuted'' functionals $\alpha_\pi(x):=-\alpha(\pi B\pi^{-1}x)$ because, given any $x\in\Delta^{n-1}$, if there exists $\alpha$ such that $\alpha_\pi(x)<0$ for all $\pi\in S_n$, then $x$ cannot be stabilisable. Conversely, if $x$ is not stabilisable, then there exists some $\alpha$ for which $\alpha_\pi(x)<0$ for all $\pi\in S_n$.
Obviously, $d$ as well as all permutations of $d$ are stabil{is}able.

\begin{figure}[!htb]
\includegraphics[width=0.45\textwidth]{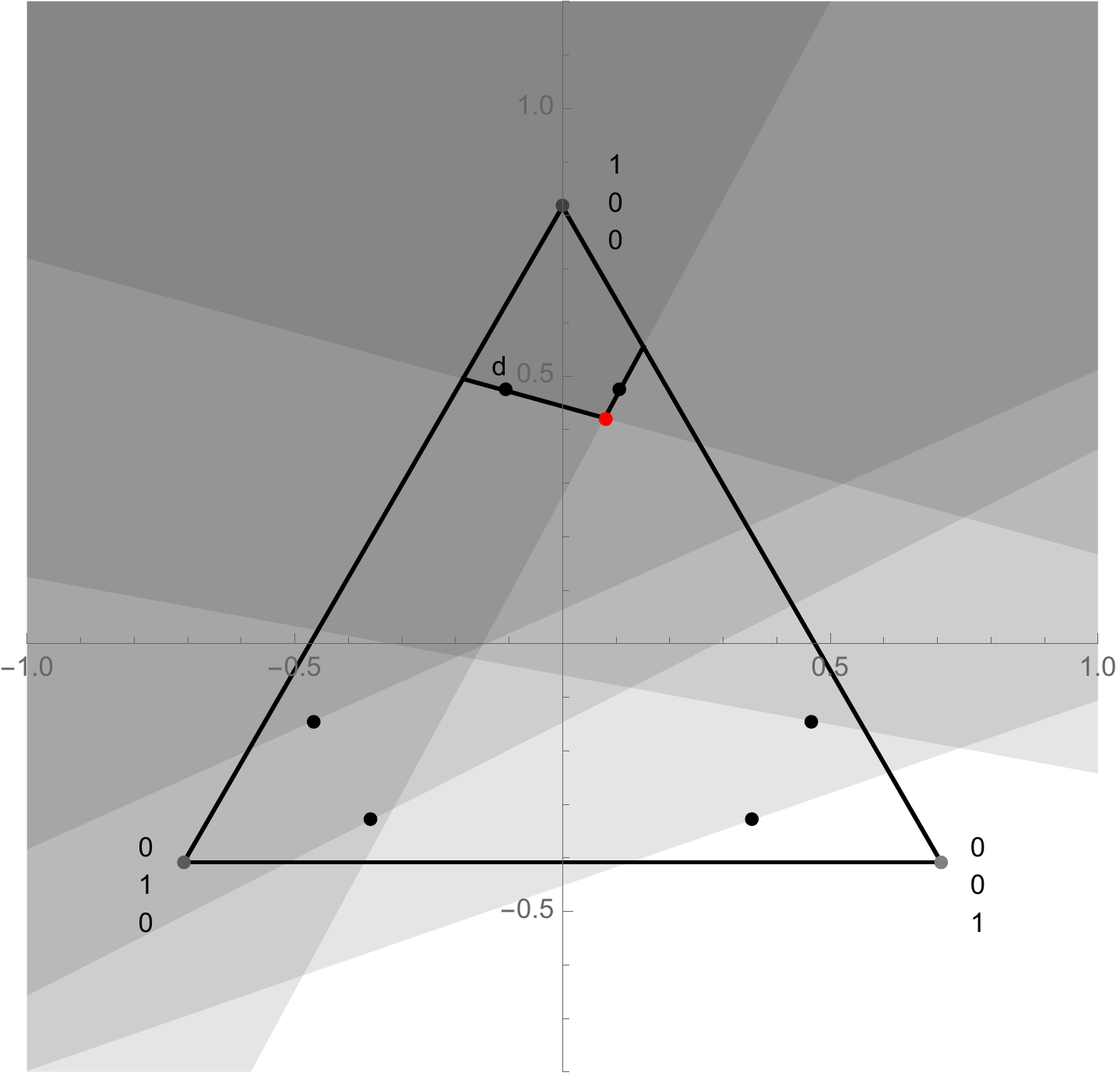}
\includegraphics[width=0.45\textwidth]{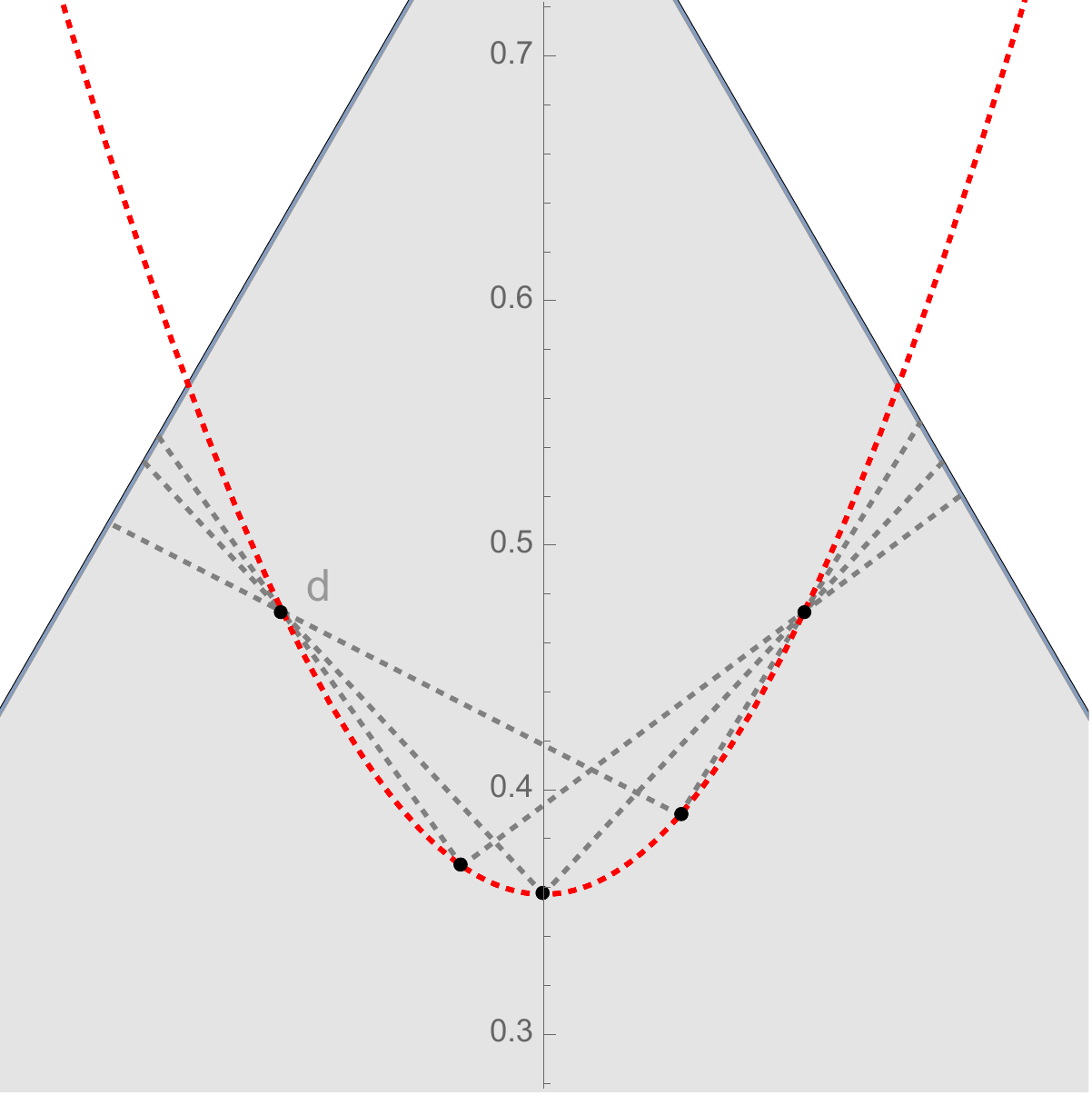}
\centering
\caption{(Colour online). Illustration of how to construct the boundary curves of the set of stabilisable points in the case of $a=0.3$. 
Left: For $\alpha=\begin{pmatrix}0&-0.4&-0.6\end{pmatrix}$ the shaded regions comprise the points where the functionals $\alpha_\pi$ are negative; 
highlighted is the intersection of all the negative regions with the simplex. The points in this region are certainly not stabil{is}able. In particular the intersection point of $\ker(\alpha_{\mathrm{id}})$ and $\ker(\alpha_{\tau_{23}})$ is marked in red.
Right: For three different values of $\alpha$,  parts of $\ker(\alpha_{\mathrm{id}})$ and $\ker(\alpha_{\tau_{23}})$ and their intersections are shown. Taken together, these intersections form the curve given in red, which constitutes a part of the boundary of the set of stabil{is}able points.}
\label{fig:curve-construction}
\end{figure}

Let us now focus on the three-dimensional case.
We will compute a closed curve connecting all these points, which will turn out to be the boundary of the set of stabil{is}able states:
everything (on or) inside the curve will be stabil{is}able and everything outside will be non-stabil{is}able---refer to Fig.~\ref{fig:stabilizable-set} below for two examples.
Let us, e.g., focus on the part of the boundary curve between $d$ and $\tau_{23}\,d$, where $\tau_{23}$ is the transposition acting on the second and third element.
Note that $d$ and $\tau_{23}\, d$ are located in neighbouring Weyl chambers since the elements in $d$ are always increasing or decreasing.
The idea for determining its shape is: for every functional $\alpha$
(in a certain range) 
one can compute a point
$\ker(\alpha_{\mathrm{id}})\cap\ker(\alpha_{\tau_{23}})\cap\Delta^2$ with the property that all points in the simplex ``above'' it cannot be stabil{is}able as shown in Fig.~\ref{fig:curve-construction}.
Moreover, due to Assumption~\textbf A, the curve will always be part of a conic section.
To motivate this approach, note that any point $x$ with $\alpha_\pi(x)<0$ for some $\alpha$ and for all $\pi\in S_n$ in contained in an open neighborhood of non-stabil{is}able points and hence cannot lie on the boundary. Thus we are looking for points lying in the kernel of at least one of the $\alpha_\pi(x)$. Moreover, we really need to find points lying in the intersection of two such kernels, since otherwise a small perturbation applied to $\alpha$ shows that the point has a non-stabil{is}able neighborhood.

Let us now invoke Assumption \textbf{A} so, w.l.o.g., $H_0:=\operatorname{diag}(-1,0,1)\cdot\Delta E$ for some $\Delta E\in\R$, 
and thus $d=(1,a,a^2)/(1+a+a^2)$ with
$a=e^{-\Delta E/T}$.
The generators of our dissipative dynamics \eqref{eq:sigma+} \& \eqref{eq:sigma-} are fully character{is}ed by the (constant) angle $\theta=\arccos(\frac{1}{\sqrt{1+a}})$ in \eqref{eq:thermal_angle}.
With this, the generator of the toy model takes the form (cf.~also \cite{CDC19})
$$
-B=\frac{2}{1+a}
\begin{pmatrix}
-a &     1   & 0 \\
 a & -1-a &1  \\
  0 & a      & -1
\end{pmatrix}\,.
$$

Let us go through the construction of the curve for the special (parabolic) case\footnote{This is the case where the energy gap $|\Delta E|=\ln(4) k_BT$, where we explicitly write the Boltzmann constant $k_B$.
}
where $a=\tfrac14$. It will turn out that the boundary curve between $d$ and $\tau_{23}d$ is fully determined by the family of functionals\footnote{\label{footn:parameter-reduction}%
Since we only care about the component of the functional parallel to the simplex and since the normalisation does not matter, it suffices to consider a one-parameter family of functionals. The exact parametrisation and parameter range are chosen for ease of computation.
}
$\alpha^\lambda$, $\lambda\in[-\frac17,\frac17]$ where
\begin{equation}\label{eq:def_alpha_lambda}
\alpha^\lambda:= -(\tfrac12+\lambda)\begin{pmatrix}0&0&1\end{pmatrix} - (\tfrac12-\lambda)\begin{pmatrix}0&1&0\end{pmatrix},
\end{equation}
so $\alpha^\lambda(x)=\lambda (x_2-x_3)-\frac12(x_2+ x_3)$ for all $x\in\mathbb R^3$.
In order to compute $\ker(\alpha^\lambda_{\mathrm{id}})\cap\ker(\alpha^\lambda_{\tau_{23}})\cap\Delta^2$ we find that $\alpha_{\operatorname{id}}^\lambda=\alpha^\lambda\circ (-B)$
(up to a global factor, which we may omit because we have to normalise later on anyway) equals
\begin{align*}
-(\tfrac12+\lambda)\begin{pmatrix}
0&\frac12&-2
\end{pmatrix}-(\tfrac12-\lambda)\begin{pmatrix}
\frac12&-\frac52&2
\end{pmatrix}
= \begin{pmatrix}
\frac{\lambda}{2}-\frac14&
1-3\lambda &
  4\lambda
\end{pmatrix}.
\end{align*}
Also $\alpha_{\tau_{23}}^\lambda=\alpha_{\operatorname{id}}^{-\lambda}\circ\tau_{23}$ is generated by $
\begin{pmatrix}
-\frac{\lambda}{2}-\frac14&-4\lambda&3\lambda+1
\end{pmatrix}
$.
With this we compute $\ker(\alpha^\lambda_{\mathrm{id}})\cap\ker(\alpha^\lambda_{\tau_{23}})$
to be spanned by ``the'' vector which is orthogonal to the normal vector of both
$\alpha^\lambda_{\operatorname{id}}$ and $\alpha^\lambda_{\tau_{23}}$, that is,
\begin{align*}
\begin{pmatrix}
\frac{\lambda}{2}-\frac14\\1-3\lambda\\4\lambda
\end{pmatrix}\times
\begin{pmatrix}
-\frac{\lambda}{2}-\frac14\\-4\lambda\\
3\lambda+1
\end{pmatrix}=\begin{pmatrix}
1+7\lambda^2\\
-\frac72\lambda^2-\frac34\lambda+\frac14 \\
-\frac72\lambda^2+\frac34\lambda+\frac14
\end{pmatrix}
\end{align*}
Intersecting the line generated by this vector with the standard simplex only introduces a normalising factor since we have: $\ker(\alpha^\lambda_{\mathrm{id}})\cap\ker(\alpha^\lambda_{\tau_{23}})\cap\Delta^2=
\frac16(
4+28\lambda^2,
-14\lambda^2-3\lambda+1 ,
-14\lambda^2+3\lambda+1)^\top
$.
Finally, we reduce the dimensionality of the problem by isometrically embedding\footnote{As usual this is done using the partial isometry
$
\footnotesize P=\begin{pmatrix}
0                  & \frac{-1}{\sqrt2} & \frac{1}{\sqrt2} \\
\sqrt{\frac{2}{3}} & \frac{-1}{\sqrt6} & \frac{-1}{\sqrt6}
\end{pmatrix}
$.
} the simplex $\Delta^2$ in $\R^2$;
this leads to the (parabolic) boundary curve 
$
(\frac{\lambda}{\sqrt2}, \frac{1+14\lambda^2}{\sqrt6})
$
where $\lambda\in[-\frac17,\frac17]$.

If $a\neq\frac14$ we modify the family of functionals $\alpha^\lambda$ introduced previously by multiplying $\lambda$ in~\eqref{eq:def_alpha_lambda} by $\tfrac12\sqrt{1+2a}|(3+2a)(1-4a)|^{-1/2}$;
however, the idea and the calculations are analogous.
In the hyperbolic\footnote{
The unital scenario $a=1$ is a special case because then $d=\frac{\mathbbm e}3$, so the set of stabil{is}able points collapses to $\{\frac{\mathbbm e}3\}$.} case $a>\frac14$
the boundary curve can be parametr{is}ed via
$$
\left(w \frac{-2\lambda}{\lambda^2-1}, u\frac{\lambda^2+1}{\lambda^2-1} + v\right)
$$
where
$$
v-u=\sqrt{\frac23}\frac{1-a}{1+2a}\,,\, u+v=\sqrt{\frac23}\frac{1-a}{1-4a}\,,\,w=\frac{\sqrt2 (1-a)a}{\sqrt{|(1 + 2a)(3 + 2 a)(1-4a)|}}\,.
$$
For the elliptic case $a\in(0,\frac14)$ one finds
$$
\Big(w\frac{2\lambda}{\lambda^2+1}, u\frac{\lambda^2-1}{\lambda^2+1} + v\Big).
$$
This covers the segment of the curve which connects $d$ and $\tau_{23}d$.
For the rest of the boundary curve note that---due to the permutation symmetry---there
are only two different curve segments, cf. Fig.~\ref{fig:stabilizable-set}.
We have just computed one of them. 
The other one is obtained by re-arranging the elements of $d$ in reverse order and repeating the calculation. One obtains the same formulas with $a$ replaced by $a^{-1}$.

\begin{figure}[!ht]
\includegraphics[width=0.495\textwidth]{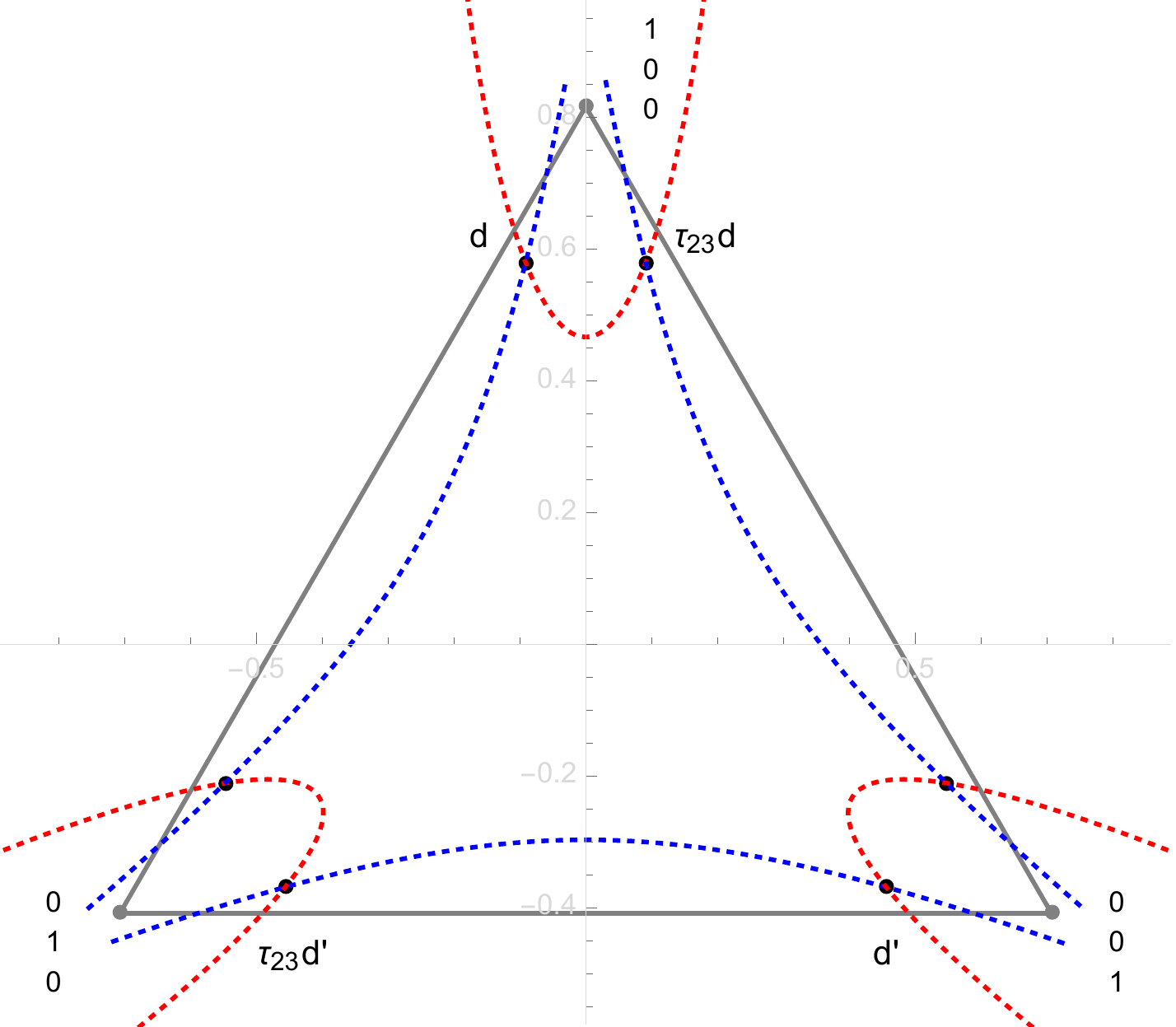}
\includegraphics[width=0.495\textwidth]{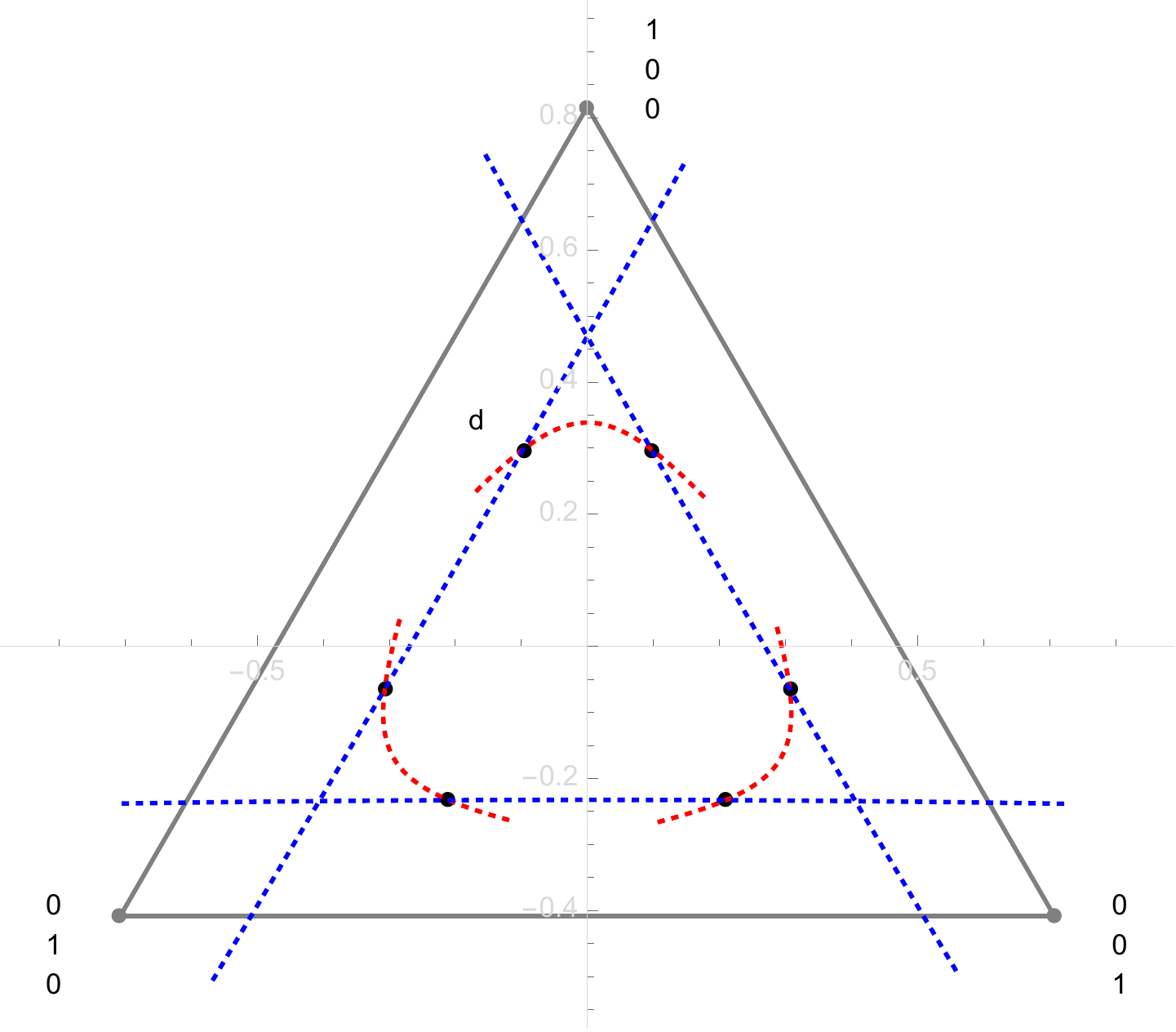}
\centering
\caption{(Colour online).
Left: The set of stabilisable states 
for the equidistant energy case 
with $a=\frac15$. The set is bounded by six conics and contains the permutations of $d$. This is the elliptic case, and part of the ellipse is drawn in red, together with its permuted copies. The blue curve is obtained analogously by taking the hyperbolic case $a=5$ whose fixed point we denote $d'$. 
Right: The same approach gives the boundary of the set of stabilisable states for a random generator $B$ 
numerically.
NB: In general the bounding curves need not be conic sections, and one may obtain a convex shape.}
\label{fig:stabilizable-set}
\end{figure}

We have seen that for each $\alpha$ we obtain an open (convex) region which is certainly not stabil{is}able. Parametr{is}ing $\alpha$ in a circular fashion---i.e.~$\alpha\in S^2\cap\{\mathbbm e\}^\perp$ in accordance with footnote~\ref{footn:parameter-reduction}---shows that this region moves 
continuously around the simplex, and its closure always touches our closed curve 
in such a way that each point outside of the curve is part of this region at some 
point, implying that all these points outside are non-stabil{is}able.

%

It remains to be shown that
every point on the boundary or enclosed within the boundary curve we just computed can in fact be stabil{is}ed.
We will only give a hand-wavy explanation;
again, each $\alpha$ yields a convex region which is not stabil{is}able, and which touches our curve in some point.
Two cases may occur: 
Either one of the halfplanes on which some $\alpha_\pi$ is negative lies outside of the major{is}ation polytope of $d$, in which case no point inside our curve is in this halfplane. 
Otherwise, we are in the case illustrated in Fig.~\ref{fig:curve-construction}.
Here the convex region of non-stabilisable points given by $\alpha$, when intersected with the majorisation polytope of $d$, is a triangle with vertices given by two permutations of $d$ and some point on our curve. Since the curve is always concave, again no point enclosed by the curve lies in the triangle of non-stabilisable points.
The case for arbitrary admissible generators $-B$ in 
the qutrit case is analogous, but the large number of parameters makes the formulas unwieldy. In higher dimensions, the idea of using the functionals $\alpha_\pi$ to determine non-stabilisable points still applies, but it is unclear how to analytically compute the resulting shapes of the stabilisable set.

\paragraph*{Reachable States.} Let us now turn towards the set of reachable states (or, more precisely, its closure)
for some given initial state and any stochastic generator matrix $-B$ with unique fixed point $d\in\Delta^{n-1}$, $d>0$.
Let us use the notation $y\twoheadleftarrow x$ to denote that $y\in\overline{\reach_B(x)}$. 
Then $\twoheadleftarrow$ is a preorder and so it induces an equivalence relation $\sim$ on which it becomes a partial order. 
In other words, $x\sim y$ if and only if $x\in \overline{\reach_B(y)}$ and $y\in\overline{\reach_B(x)}$ meaning there exists an approximately periodic solution through $x$ and $y$.
Note that up to a viability condition, the equivalence classes $[x]$ of this equivalence relation correspond to \emph{control sets} as defined in~\cite[Def.~3.1.2]{CK00}, and the induced partial order corresponds to the \textit{reachability order}~\cite[Def.~3.1.7]{CK00}.

First we observe that the maximally mixed state can always be reached:

\begin{lemma} \label{lemma:e-d-reach}
For all $x\in\Delta^{n-1}$, the vectors $d$ and $\frac1n\mathbbm e$ are in $\overline{\mathfrak{reach}_B(x)}$.
\end{lemma}

\begin{proof}
Since $d$ is the unique fixed point of $e^{-tB}$ for $t>0$, and since it is attractive\footnote{This follows from a basic result on continuous-time Markov chains. Here $-B$ is the transition rate matrix. It is irreducible (in the sense of~\cite[p.~111]{Norris97}) since $d>0$ is the unique fixed point. 
Then~\cite[Thm.~3.6.2]{Norris97} shows that the corresponding Markov chain is ergodic, i.e.~the unique fixed point is attractive.}, $d\twoheadleftarrow x$ for all $x\in\Delta^{n-1}$. Similarly, consider $\hat  B=\tfrac1{n!}\sum_{\pi\in S_n} \pi B\pi^{-1}$. 
Then $\hat B$ is invariant under permutations, which implies that $\hat B\mathbbm e=0$. 
Moreover $\frac1n\mathbbm e$ is the unique fixed point in $\Delta^{n-1}$
since otherwise, by permutation symmetry there would be an open set of fixed points in $\Delta^{n-1}$, and hence $\hat B\equiv0$. This would imply that $B\equiv 0$ as one can check by considering the value of $\hat B$ at the vertices of $\Delta^{n-1}$.
As before, the fixed point $\frac1n\mathbbm e$ of $\hat B$ is attractive.
\end{proof}
\noindent This lemma shows $d\sim\frac1n\mathbbm e$, and that the equivalence class $[d]=[\frac1n\mathbbm e]$ is an invariant control set as defined in~\cite[Def.~3.1.3]{CK00}.

Let us now, again, restrict to the three-dimensional case.
It turns out that this equivalence class is the only one that contains more than a single point:
the idea is that equivalence classes with at least two points lead to (approximately) periodic solutions which must enclose a stabilisable point. Since the set of non-stabilsable points is simply connected, when restricting to a Weyl chamber the periodic solution intersects the set of stabilisable states, which are all equivalent to $\frac13\mathbbm e$.
A proof can be found in App.~\appref{D}.

For any non-stabil{is}able state $x$, it holds that the convex cone generated by $\derv(x)$ is pointed (i.e.~its edge is a point). Hence there are two extremal derivatives at the boundary of the cone, which we will call the left and right extremal derivatives, as seen from $x$. The resulting extremal vector fields are depicted in Fig.~\ref{fig:extremal-vfs-and-reach}. More precisely we have the following result.

\begin{figure}[t]
\includegraphics[width=0.485\textwidth]{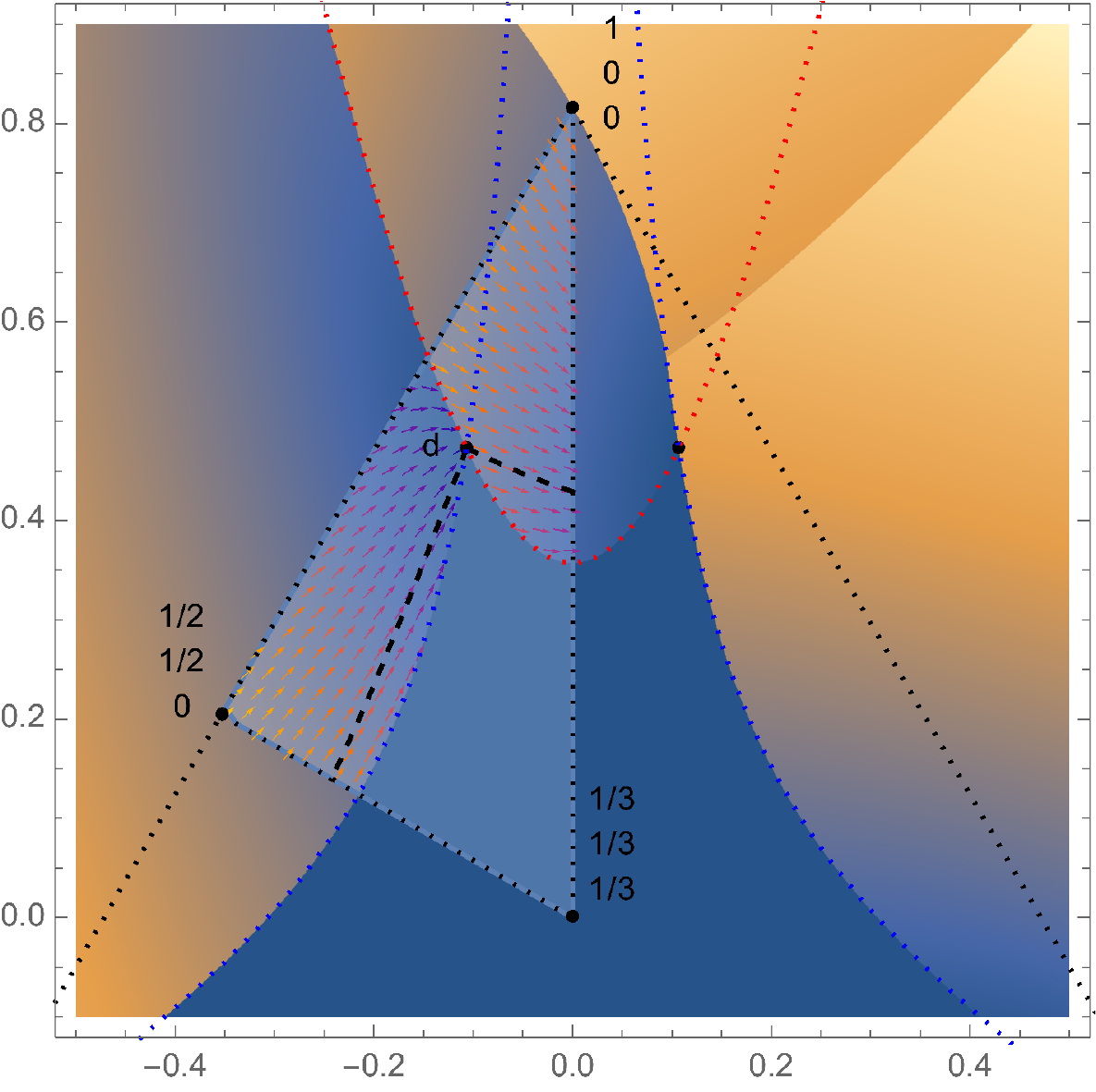}
\mbox{\hspace{.5mm}\raisebox{1.4mm}{\includegraphics[width=0.475\textwidth]{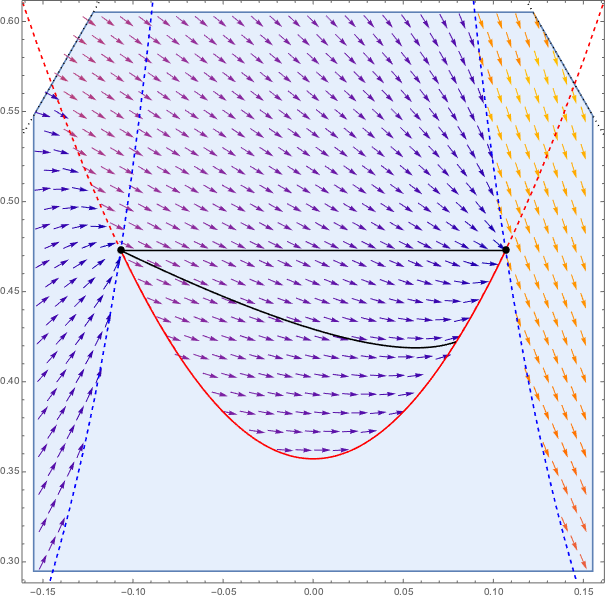}}}
\centering
\caption{(Colour online). 
Left: The left extremal vector field in the case $a=0.3$ depicted in the indicated Weyl chamber. 
The vector field is undefined on the stabil{is}able set, and hence this yields another way to compute the set of stabil{is}able points. 
We plotted again the bounding conics of the stabil{is}able set and the trajectories bounding the reachable set $\overline{\reach_B(d)}$. 
In this case the boundary trajectories are obtained by starting at $d$ and permuting $B$ such that one of the neighbours of $d$ is the unique fixed point. 
The background colors show the norm of the left extremal vector field, and its discontinuities are clearly visible. 
Right: A zoomed-in picture of the ``D''-shaped region considered in the proof of Lem.~\ref{lemma:extremals} again with parameter $a=0.3$.}
\label{fig:extremal-vfs-and-reach}
\end{figure}

\begin{lemma}
On the set $\Delta^{2}\setminus\mathrm{int}(\mathfrak{stab}_B)$ there exist left and right extremal vector fields. The norm of these vector fields might not be continuous, but the direction field is locally Lipschitz continuous, except possibly at $d$ (and its permutations).
\end{lemma}

\begin{proof}
As already mentioned, for any non-stabilisable point $x$, the convex cone generated by $\derv(x)$ is pointed. On the other hand, if for some $x$ the convex cone is the plane, then $x$ is in the interior of the stabilisable set. Hence on the boundary of the stabilisable set, the convex cone is either pointed or a half space. Either way there is a well defined left and right extremal derivative, and so the corresponding vector fields are well-defined.
Locally, for $x\neq \pi d$ the direction field can be seen as a maximum of finitely many smooth functions, and hence it is locally Lipschitz continuous. 
\end{proof}

The discontinuities in the norm are important, as they tell us when the control permutation has to be applied. The shapes of these discontinuities are non-trivial, and we show an example in Fig.~\ref{fig:extremal-vfs-and-reach}.

\begin{figure}[t]
\centering
\includegraphics[width=0.85\textwidth]{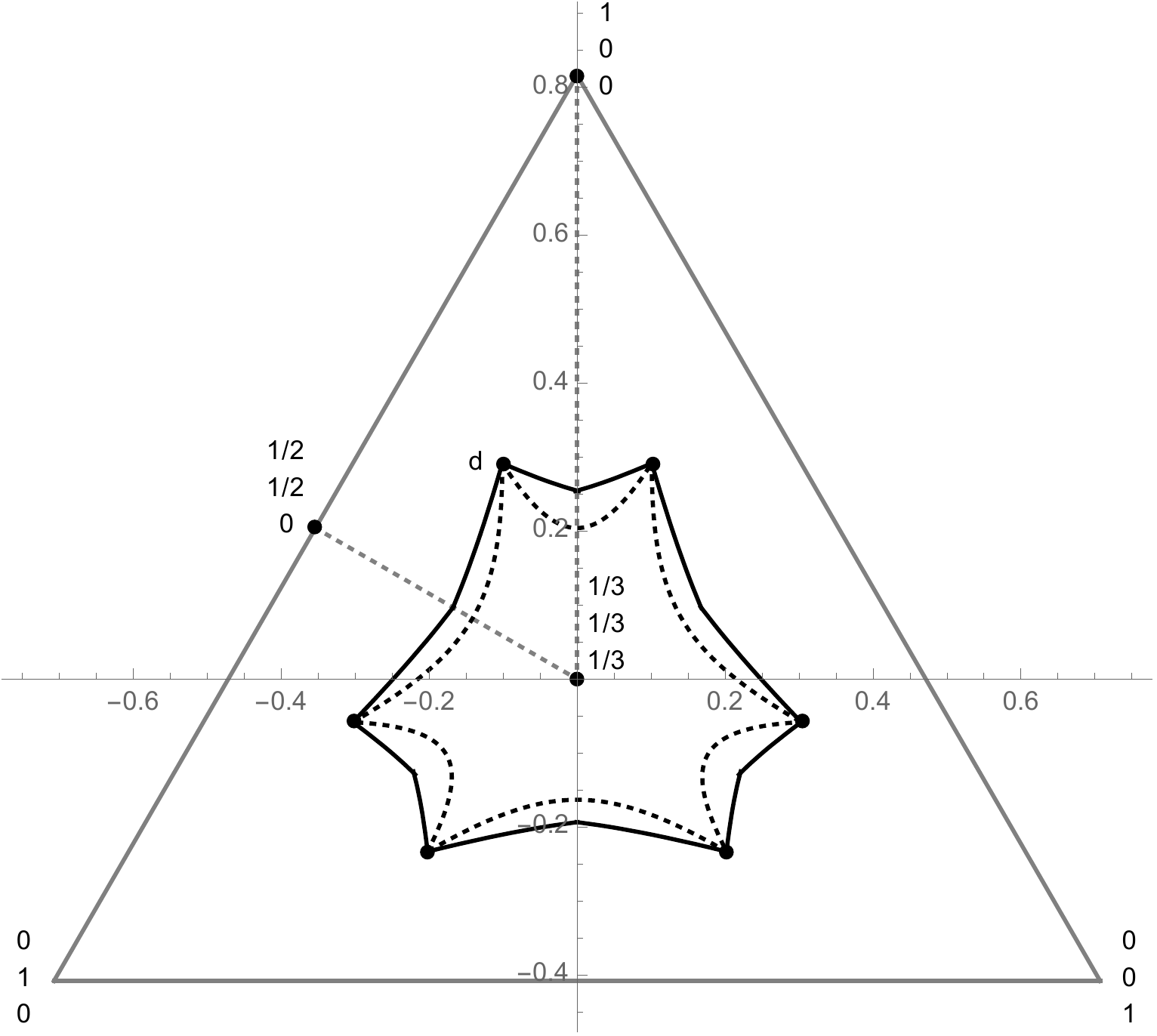}
\caption{The boundary of the reachable set $\overline{\reach_B(d)}$ in the case $a=0.5$ is shown with solid lines, and the boundary of the stabil{is}able set $\stab_B$ is shown using dotted lines (cp.~Fig.~\ref{fig:cdc-mtns}). NB: The curve segments of the boundary of the reachable set are not straight, though the curvatures are hardly visible.}
\label{fig:reach-stab}
\end{figure}

Now the boundary of the reachable set can be computed using solutions following the left and right extremal derivatives.
By the previous lemma these solutions exist and are unique. See again Fig.~\ref{fig:extremal-vfs-and-reach} as well as Fig.~\ref{fig:reach-stab}. Note that the extremal vector fields never vanish where they are defined, and since they are defined on a contractible domain (if restricted to a Weyl chamber) there are no periodic solutions.
This relies on the fact that the state space is two dimensional, see for instance~\cite[Thm.~6.8.2]{Strogatz15}.
Moreover, given a left (right) extremal solution, another solution can only cross it from left to right (right to left);
this can be shown as in the proof of~\cite[Thm.~5.6]{Smirnov02}.

Before we can prove this section's main result 
we need the following topological result about reachable sets:
\begin{lemma} \label{lemma:contract}
Let any $x\in\Delta^2$ be given. Then $\overline{\reach_B(x)}$ is contractible.
\end{lemma}

\begin{proof}
Consider the map $F:\overline{\reach_B(x)}\times [0,1]\to \overline{\reach_B(x)}$ defined by
$$
F(y,t)=\begin{cases}
e^{-\hat B f(t)} y &\text{ if } t<1\\
\tfrac13\mathbbm e &\text{ else, }
\end{cases}
$$
where $\hat B$ is defined as in the proof of Lem.~\ref{lemma:e-d-reach}
and $f:[0,1)\to [0,\infty)$ is any homeomorphism.
It follows from the same lemma that $F$ is continuous, and hence a (strong) deformation retraction.
\end{proof}

\begin{lemma} \label{lemma:extremals}
Invoke Assumption \textbf{A} and further assume that $d\in\Delta^{2}_\downarrow$ where $\Delta^{2}_\downarrow$ denotes the ordered Weyl chamber of the simplex.
The left extremal solution starting from $d$ lies in the complement of the interior of the stabil{is}able region and terminates in the boundary of the Weyl chamber in finite time, without leaving the (classical) major{is}ation polytope of $d$.
The analogous result holds for the right extremal solution.
\end{lemma}

\begin{proof}
From Cor.~\ref{coro:thm_3g} we know that the left extremal solution remains in the majorisation polytope of $d$. 
Let us sketch why this solution cannot enter the set of stabilisable points in $\Delta^{2}_\downarrow$.
Consider the connected region containing $d$ and $\tau_{23}d$ which is bounded by the majorisation polytope and the set of stabilisable points, and is shaped like a ``D'' lying on its belly, so let's call it $D$, see right panel of Fig.~\ref{fig:extremal-vfs-and-reach}.

First note that since there are no fixed points in $D$, every solution reaches the boundary of $D$ in finite time, and by the above it reaches the curved part of the boundary of $D$. Now consider the straight part of the boundary, between $d$ and $\tau_{23}d$. The solutions starting from points close to $\tau_{23}d$ will reach the curved boundary of $D$ on the right side. Hence by continuity all points on the straight part of the boundary have solutions ending up on a connected part of the curved boundary. However, as we have seen before, on the boundary of the set of stabilisable points, the cone generated by the achievable derivatives is a half plane, and hence the left and right extremal vector fields point in opposite directions. 
Therefore in general only one kind of solution can terminate in each point. By symmetry and the above connectedness, all left extremal solutions must terminate on the right side of $D$,
and analogously for the right extremal solution. As noted above, no solution can leave the region delimited by the extremal solutions.
\end{proof}

\begin{thm} \label{thm:reach-d}
The left and right extremal solutions starting at $d$ separate the Weyl chamber into two parts, and the inner part containing $\frac13\mathbbm e$ equals the intersection of $\overline{\reach_B(d)}=[d]$ with $\Delta^{2}_\downarrow$.
\end{thm}

\begin{proof}
Taking the extremal solutions in the ordered Weyl chamber and all of the permuted copies yields a closed curve surrounding and contained in $\overline{\reach_B(d)}$, i.e.~they form the ``outer boundary'' of the reachable set.
By Lem.~\ref{lemma:contract} $\overline{\reach_B(d)}$ is contractible, hence it is equal to the region enclosed by this curve.
\end{proof}


\begin{corollary} \label{coro:stab-reach}
For every $x\in\stab_B$ it holds that $x\sim d$.
\end{corollary}

\begin{proof}
By Lem.~\ref{lemma:e-d-reach} it holds that $d\twoheadleftarrow x$.
Thm.~\ref{thm:reach-d} shows that $\overline{\reach_B(d)}$ is the set enclosed by the left and right extremal solutions  (and their permuted copies) starting at $d$ and ending in the boundary of the Weyl chamber. Moreover Lem.~\ref{lemma:extremals} shows that this set contains $\stab_B$ and hence $x\twoheadleftarrow d$.
\end{proof}


For starting points other than $d$, a similar result holds.

\begin{corollary}
For any point $x$ outside of $[d]$, we can compute the boundary of $\overline{\reach_B(x)}$ in $\Delta^{2}_\downarrow$ by following the left and right extremal solutions until we hit either the boundary of the Weyl chamber or $[d]$. Moreover, the left extremal solution can only terminate in the right boundary of the Weyl chamber or in the left boundary of $[d]$ and vice-versa.
\end{corollary}

\begin{proof}
Since $d>0$ no solution tends to the boundary of the simplex. Hence the left and right extremal solutions must terminate in the boundary of the Weyl chamber or of $[d]$. The fact that the left extremal solution can only terminate in the right boundary of the Weyl chamber or in the left boundary of $[d]$ follows from the fact that integral curves do not intersect and the fact that on the symmetry lines of the simplex, the cone of achievable derivatives opens towards $\frac13\mathbbm e$.
\end{proof}

The more general case\footnote{cp.\ also Rem.~\ref{rem_equidist_necessary}} 
of $B$ {\em not} satisfying Assumption~\textbf{A} can be treated with similar methods. Note, however, that many of our arguments rely on the fact that the state space is two dimensional, and hence it is not clear how to analytically determine stabilisable and reachable sets in higher dimensions.
\section{Conclusions}

We have analysed quantum control systems with thermal resources, where one combines
coherent unitary controls with a switchable coupling to a thermal bath
as additional control.
To this end, we have characterised (the closure of) thermal operations ($\overline{\mathsf{TO}}$), 
enhanced thermal operations ($\mathsf{EnTO}$), 
as well as Gibbs-preserving maps ($\mathsf{Gibbs}$)
within the framework of Lie-semigroup theory. Thereby we could determine (up to Conj.~\ref{conj1}) the structure of the respective
semigroups, their Lie wedges as well as their edges. It is important to note that
the Lie wedges in turn {\em define} 
  the corresponding {\em Markovian} counterparts as the corresponding {\em Lie semigroups}, to wit $\mathsf{MTO}$, $\mathsf{MEnTO}$, $\mathsf{MGibbs}$.
A worked qubit example illustrates how Markovian thermal operations sit inside all thermal operations 
by means of an explicit parametrisation.
In case of single qubits the Markovian thermal operations $\mathsf{MTO}$
even exhaust the entire set of thermal operations $\overline{\mathsf{TO}}$ in the zero-temperature limit.
On a general scale, Thm.~\ref{thm:markov_generator} provides an explicit construction for 
(possibly all) generators of Markovian thermal
operations via temperature-weighted projections out of a total Hamiltonian 
(preserving energy of system and bath) in the Stinespring dilation.

In view of studying reachable sets, the semigroup techniques naturally match with general
concepts of majorisation ($d$-majorisation). For the evolution
of diagonal states under such controlled Markovian dynamics, we have upper bounded 
the reachable sets by inclusions within the standard simplex $\Delta^{n-1}$: they can readily be given in terms 
of the convex hull of extreme points of the $d$-majorisation polytope. 
---
Finally, for the qutrit case, we have explicitly determined and illustrated the geometry of both,
the reachable set and the stabil{is}able set
by techniques of differential inclusion.

%
\section{Outlook}
There are several ways to generalise the toy model with all permutations as controls plus a specific generator 
of a one-parameter semigroup of thermal operations:

First one may define again a reduced control system still on the simplex, but now encapsulating the entire unitary control.
While our results for $T=0$ (Thms.~\ref{thm_1} \& \ref{thm_2}) immediately carry over into this generalisation\footnote{see also \cite[Cor.~5.1.12]{vE_PhD_2020}},
it is not obvious how to adapt our upper bounds (Thm.~\ref{thm:general_dmaj_bound}) and if analytic solutions in the three-dimensional case (Sec.~\ref{sec:Qutrit-Results}) are still obtainable.
Studying the interplay between Markovian operations and general unitary dynamics will be the subject of 
future work~\cite{MEDS23}.

Next one could allow for arbitrary Markovian thermal operations together with all unitary maps---the generated semigroup of which we denote by 
$\mathsf{MTO}_{\mathsf{U}}(H_0,T)$ ($\mathsf{MTO}_{\mathsf{U}}$ in abuse of notation)---and 
ask for best approximations to the corresponding reachable sets given by the semigroup orbit ${\mathsf{MTO}_{\mathsf{U}}(\rho_0)}$,
or analogously ${\mathsf{MEnTO}_{\mathsf{U}}(\rho_0)}$ or ${\mathsf{MGibbs}_{\mathsf{U}}(\rho_0)}$ on a general scale.
For non-zero temperatures, this question boils down to feasible state transfers under $\mathsf{MTO}(H_0,T)$ beyond
the simple bath dynamics of Cor.~\ref{coro_ladder_ops_are_TO}.
Such a setting would generalise 
\cite{LosKor22a} which itself characterised the reachable set $\mathsf{MEnTO}(\rho_0)$\,\footnote{
Note that this would be ${\mathsf{MEnTO}_{\mathsf{U}}(\rho_0)}$ if one allowed for \textit{all} unitary maps instead of just those with $H_0$ as fixed point.
}
for $\rho_0$ quasi-classical and led to a ``Markovian'' generalisation of $d$-majorisation.

Finally lifting considerations to the operator level, the single-qubit observation in this work that
in the zero-temperature limit {\em Markovian} thermal operations converge to general 
thermal operations $\mathsf{MTO}(H_0,T) \to \overline{\mathsf{TO}(H_0,T)}$ begs the question 
what happens in the general case with $\mathsf{MTO}(H_0,T)$ 
as $T \to 0^+$.
---
Beyond Markovianity, in analogy to above take $\overline{\mathsf{TO}}_{\mathsf U}(H_0,T)$ (again $\overline{\mathsf{TO}}_{\mathsf U}$ for short) as the smallest semigroup 
now embracing all unitary operations as well as all thermal operations $\overline{\mathsf{TO}(H_0,T)}$.
%
While the reachable sets $\overline{\mathsf{TO}}_{\mathsf U}(\rho_0)$ are known\footnote{$\overline{\mathsf{TO}}_{\mathsf U}$ acts (approximately) transitively~(i) on the set of all density operators for all $T\in[0,\infty)$ \cite[Prop.~4.12]{vomEnde20Dmaj}
and~(ii) on the set of all states majorised by the initial state if $T=\infty$ \cite[Prop.~5.2.1]{vE_PhD_2020}.}
(also see the first teaser of the introduction), $\overline{\mathsf{TO}}_{\mathsf U}$ is not yet explored on the level of quantum maps either.

\medskip
All these generalisations would help to understand
how Markovianity interrelates with quantum thermodynamics at large.
\newpage

\section*{Appendix \app{A}: A Simple Proof of Proposition~\ref{prop:Lie-wedge}~(iv)}
While $\subseteq$ in Eq.~\eqref{eq:tangent-cone}
is obvious, for $\supseteq$
choose $\gamma \in C^1$ 
as in the r.h.s.~of \eqref{eq:tangent-cone}
with $\dot{\gamma}(0) =: A$. Then given any $t \geq 0$ sufficiently small we compute
\begin{equation*}
n\log \big(\gamma(\tfrac{t}{n})\big)
= n \log \big(\operatorname{id} + \tfrac{t}{n}A + {\scriptstyle \mathcal{O}}(\tfrac{t}{n})\big)
 = n \big(\tfrac{t}{n}A + {\scriptstyle \mathcal{O}}(\tfrac{t}{n})\big) \to tA 
\end{equation*}
as $n\to\infty$.
Be aware that we are allowed to apply the logarithm to $\gamma(\frac{t}{n})$ and, more importantly, $(\gamma(\frac{t}{n}))^n$ because
\begin{equation*}
\begin{split}
  \|\operatorname{id} & -( \gamma(\tfrac{t}n))^n\| = \|\operatorname{id} - (\operatorname{id} + \tfrac{t}n A + {\scriptstyle \mathcal{O}}(\tfrac{t}{n}))^n\|\\
  & = \|\operatorname{id} - \operatorname{id} - n \big(\tfrac{t}n A + {\scriptstyle \mathcal{O}}(\tfrac{t}{n})\big) - { n \choose 2} \big(\tfrac{t}n A + {\scriptstyle \mathcal{O}}(\tfrac{t}{n})\big)^2 - \dots \|\\
  & \leq  \| At + n {\scriptstyle \mathcal{O}}(\tfrac{t}{n})\| \,\cdot\, \| \operatorname{id} + n^{-1}{ n \choose 2} \big(\tfrac{t}n A + {\scriptstyle \mathcal{O}}(\tfrac{t}{n})\big) - \dots\| \leq Ct (\| A\| + \varepsilon)
\end{split}
\end{equation*}
for all $n$ sufficiently large as $t$ was chosen to be suitably small. Here $C$ can be obtained by a brute force estimate of the remaining terms of the binomial formula.
Now because $S$ is a closed semigroup we are able to conclude that
\begin{equation*}
  \lim_{n \to \infty}\gamma\big(\tfrac{t}{n}\big)^n = \lim_{n \to \infty}e^{\log \big((\gamma(\tfrac{t}{n}))^n\big)} = \lim_{n \to \infty}e^{n\log (\gamma(\tfrac{t}{n}))} =  e^{tA}
\end{equation*}
is in $S$. Thus $(e^{tA})_{t\geq 0}\subseteq S$, again due to the semigroup property of $S$.

\section*{Appendix \app{B}: Proof of Theorem~\ref{thm_lie_global}}

  (a): The result can be found in \cite[Thm.~V.1.13]{HHL89} (again under a much more general setting). Here the
  straightforward proof under our assumptions: let $S \subseteq\mathcal B(\mathcal Z)$ be a closed subgroup and $\mathsf{L}(S)$
  its Lie wedge. Then $S_0$ is obviously a closed subsemigroup contained in $S$ and thus $\mathsf{L}(S_0) = \mathsf{L}(S)$. This implies
$
    S_0 = \overline{\langle \exp\big(\mathsf{L}(S)\big) \rangle}_{\rm SG} = \overline{\langle \exp\big(\mathsf{L}(S_0)\big) \rangle}_{\rm SG}\,,
$
  i.e.~$S_0$ is a Lie subsemigroup. Moreover, let $S'$ be any other Lie subsemigroup contained in $S$. Then one has
  $\mathsf{L}(S') \subseteq \mathsf{L}(S)$ and thus
$
    S' = \overline{\langle \exp(\mathsf{L}(S')) \rangle}_{\rm SG} \subseteq \overline{\langle \exp(\mathsf{L}(S)) \rangle}_{\rm SG} =S_0\,.
$
  Hence $S_0$ is the largest Lie subsemigroup of $S$.
   (i): Moreover, let $S'$ be any other Lie subsemigroup contained in $S$. Then one has
  $\mathsf{L}(S') \subseteq \mathsf{L}(S)$ and thus
$
    S' = \overline{\langle \exp(\mathsf{L}(S')) \rangle}_{\rm SG} \subseteq \overline{\langle \exp(\mathsf{L}(S)) \rangle}_{\rm SG} =S_0\,.
$
Hence $S_0$ is the largest Lie subsemigroup of $S$. (ii): For piecewise constant controls, the reachable set of the identity
  obviously coincides with the semigroup generated by $\exp\big(\mathsf{L}(S))$. Moreover, as any locally integrable functions
  can be ($L^1$-norm) approximated on bounded intervals by piecewise constant functions, we conclude that the closure
  of the reachable set $\overline{\mathfrak{reach}({\rm id})}$ equals $\overline{\langle \exp\big(\mathsf{L}(S)\big) \rangle_{\rm SG}}= S_0$.
  
  (b): The case $S \subseteq\GL(\mathcal Z)$ is treated in \cite[Cor.~VI.5.2]{HHL89}; the case $S \subseteq\mathcal B(\mathcal Z)$
  follows readily from the fact that for every Lie subsemigroup the set of its invertible elements (i.e.~$S \cap \GL(\mathcal Z)$)
  is dense in
  $S$.  

\section*{Appendix \app{C}: Proof Idea of Equation \eqref{eq:stinespring_taylor}}

A simple way of verifying \eqref{eq:stinespring_taylor} is to expand the exponentials:
\begin{align*}
\operatorname{tr}_{\mathbb C^m}\big(e^{-itH}((\cdot)\otimes\omega)e^{itH}\big)&=\sum_{j,k=0}^\infty \operatorname{tr}_{\mathbb C^m}\Big(\frac{(-itH)^j}{j!}((\cdot)\otimes\omega)\frac{(itH)^k}{k!}\Big)\\
&=\sum_{j,k=0}^\infty \frac{(-1)^j(it)^{j+k}}{j!k!}   \operatorname{tr}_{\mathbb C^m}\big(H^j((\cdot)\otimes\omega)H^k\big)\\
&= \sum_{\ell=0}^\infty(it)^\ell\sum_{j=0}^\ell \frac{(-1)^{j}}{j!(\ell-j)!}\operatorname{tr}_{\mathbb C^m}\big(H^j((\cdot)\otimes\omega)H^{\ell-j}\big)
\end{align*}
Therefore the first-order term is
$
-it(\operatorname{tr}_{\mathbb C^m}(H((\cdot)\otimes\omega))-
\operatorname{tr}_{\mathbb C^m}((\cdot)\otimes\omega)H))
$
which by \cite[Eqs.~(14) \& (15)]{vE22_Stinespring} equals
$
-it(\operatorname{tr}_{\omega}(H)(\cdot)-(\cdot)\operatorname{tr}_{\omega}(H))=-it\operatorname{ad}_{\operatorname{tr}_{\omega}(H)}
$.
Similarly, the second-order term comes out to be
\begin{equation}\label{eq:davies_generator}
t^2\Big(\operatorname{tr}_{\mathbb C^m}\big(H((\cdot)\otimes\omega)H\big) -\frac12\operatorname{tr}_{\omega}(H)(\cdot)-\frac12(\cdot)\operatorname{tr}_{\omega}(H) \Big)\,.
\end{equation}
Defining $\Phi_H:=\operatorname{tr}_{\mathbb C^m}(H((\cdot)\otimes\omega)H) $---which is completely positive---the second factor from Eq.~\eqref{eq:davies_generator} can be re-written as $\Phi_H-\frac{\Phi_H^*(\mathbbm1)}{2}(\cdot)-(\cdot)\frac{\Phi_H^*(\mathbbm1)}{2}$.
This is known to be the generator of a quantum-dynamical semigroup,
and one recovers Eq.~\eqref{eq:lindblad_V} by choosing the $V_j$ as Kraus 
operators of $\Phi_H$, see \cite[Ch.~9, Thm.~4.2 \& Eq.~(4.16)]{Davies76}.
A straight-forward computation shows that a set of Kraus operators of $\Phi_H$ is given by $(\sqrt{r_k}\operatorname{tr}_{|g_k\rangle\langle g_j|}(H))_{j,k=1}^m$.
Altogether this yields \eqref{eq:stinespring_taylor}.

\section*{Appendix \app{D}: Periodic Solutions in the Qutrit System}

We show that in the qutrit case, periodic solutions enclose a stabilisable point
which---as we will see below---implies that 
$[d]=[\frac13\mathbbm e]$ is the only non-trivial equivalence class.
We work in the setting of Sec.~\ref{sec:Qutrit-Results}~using Assumption \textbf{A}, in particular we think of the control system being given in the form of the differential inclusion~\eqref{eq:diff-incl-toy-model}.

\begin{lemma} \label{lemma:periodic-stab}
Let $x:S^1\to\Delta^2$ be a smooth, periodic, injective solution of the differential inclusion with non-vanishing derivative. Then the region enclosed by $x$ contains a stabilisable point.
\end{lemma}

\begin{proof}
This is a direct generalisation of~\cite[Ch.~5.2, Thm.~1]{Aubin84}
which states that if an upper semicontinuous differential inclusion with non-empty, closed, convex values is defined on a compact convex set and satisfies a viability condition, then it has a stabilisable point.
By the Schoenfliess Theorem, 
see~\cite[Ch.~9, Thm.~6]{Moise77}, the interior region of $x(S^1)$ is homeomorphic to an open disk, and hence by the Riemann mapping theorem, there is even a biholomorphism.
Now note that since $x$ is an injective immersion and $S^1$ is compact, it is an embedding, and hence the image is a smooth curve. 
Thus, by~\cite[Thm.~3.1]{Bell90}, the Riemann mapping extends to a diffeomorphism of the closure of the interior region to the closed disk.
Finally we can pull back the differential inclusion to the disk and apply the aforementioned theorem to find a stabilisable point.
\end{proof}

The idea of the result we want to prove is that if two points are equivalent but distinct, then they must be equivalent to some stabil{is}able point. To prove this in general we need the following approximation result.

\begin{lemma} \label{lemma:surgery}
Let $x\neq y$ and $y\twoheadleftarrow x$, and assume that $y$ is not stabil{is}able. Let a solution $\tilde y$ starting at $y$ be given. Then for every $\varepsilon>0$ small enough we can modify the differential inclusion in the region $B_\varepsilon(y)$ without creating new stabil{is}able points and such that there is a smooth solution $\tilde x$ starting at $x$ and ending at $y$ such that the concatenation of $\tilde x$ and $\tilde y$ is smooth.
\end{lemma} 


\begin{proof}[Sketch of proof]
Using translations and rotations we may assume that $y=(0,0)$ and the cone generated by $\derv(y)$ is contained in the upper halfplane and symmetric about the vertical axis. By continuity and assuming that $\delta$ is small enough, there are inner and outer approximations of this cone in $B_\delta(y)$ which are both pointed. 
We may assume (e.g., by extending $x$ backwards) that $x'(0)$ lies in the inner approximating cone.
We will only modify the differential inclusion within the lower half of this disk. Now assume that for some small enough $0<\varepsilon\ll1$ we have a smooth solution $\tilde x$ starting at $x$ that ends $\varepsilon$-close to $y$. Then by slightly enlarging the outer cone we may assume that $\tilde x$ enters the unit disk within the negative of the outer cone. One can see that it is possible to modify the differential inclusion inside $B_R\setminus B_r$ for some $0<r<R<1$ such that there is a smooth solution entering $B_r$ inside of the inner approximating cone, while making sure that the cone always lies in the upper halfplane, so that no stabilisable points are created. 
\end{proof}

\begin{proposition} \label{prop:equivalence-classes}
If $x\neq y$ and $x\sim y$, then $x\sim\frac13\mathbbm e$. 
\end{proposition}

\begin{proof}
If $x$ or $y$ is stabilisable, then by Cor.~\ref{coro:stab-reach} it is equivalent to $d$ and we are done.
Hence we assume that neither $x$ nor $y$ is stabilisable.
Let $\varepsilon>0$ small enough be given.
Since $x\sim y$, we may apply Lem.~\ref{lemma:surgery} twice to obtain a smooth, periodic solution passing through $x$ and $y$ for a slightly modified differential inclusion, which does not introduce new stabilisable points.
Without loss of generality we may assume that this solution is injective and has non-vanishing derivative.
By Lem.~\ref{lemma:periodic-stab} it encloses a stabilisable point.
However, if we work in a Weyl chamber, the non-stabilisable set is simply-connected, and so the periodic solution intersects the stabilisable region in some point $s$.
Hence there is a point $\varepsilon$-close to $x$ which is reachable from $s$ (and by Cor.~\ref{coro:stab-reach} also from $\frac13\mathbbm e$). Letting $\varepsilon$ go to $0$ this shows that $x\twoheadleftarrow\frac13\mathbbm e$.
\end{proof}

\section*{Appendix \app{E}: Markovian Thermal Single-Qubit Operations with Different Temperatures}
We compute the product of two Markovian thermal operations in the single-qubit case and show that the result is again Markovian and thermal.
Recall that every thermal qubit operation in $\mathsf{MTO}(H_0,T)$ is represented by three parameters\footnote{
Strictly speaking, one has to replace $\varepsilon$ by $\mu\varepsilon$ since the former leads to an ill-defined composition rule if $\mu=0$.
However, this special case can easily be dealt with, so for the sake of simplicity and clarity we will treat $\varepsilon$ as independent.
}:
$\mu,\varepsilon\in\mathbb R,\, c\in\C$.
A matrix representation is given by
\begin{equation}\label{eq:app_E_G}
G(\mu,\varepsilon,c)=\begin{pmatrix}
1-\varepsilon\mu & \mu    & 0 \\
\varepsilon\mu    & 1-\mu & 0 \\
0                   & 0      & c
\end{pmatrix}.
\end{equation}
Considering the product $G(\mu_3,\varepsilon_3,c_3)=G(\mu_1,\varepsilon_1,c_1)\, G(\mu_2,\varepsilon_2,c_2)$ we find that
\begin{align*}
\mu_3                &= \mu_1+\mu_2-\mu_1\mu_2(1+\varepsilon_1)\\
\varepsilon_3\mu_3 &= \varepsilon_1\mu_1 + \varepsilon_2\mu_2 - \mu_1\mu_2(1+\varepsilon_1)\varepsilon_2
\end{align*}
as well as $c_3=c_1c_2$.
Note that this product in general is not commutative.

Actually, in order for~\eqref{eq:app_E_G} to describe a thermal operation recall that $\mu,\varepsilon\in[0,1]$ and $|c|^2\leq(1-\varepsilon\mu)(1-\mu)$, and for the operation to be Markovian the latter condition is replaced by $|c|^2\leq 1-\mu(1+\varepsilon)$.

\begin{lemma}
Let parameters $\mu_1,\varepsilon_1,\mu_2,\varepsilon_2\in[0,1]$, $c_1,c_2\in\mathbb C$ be given and set $(\mu_3,\varepsilon_3,c_3):=(\mu_1,\varepsilon_1,c_1)\cdot(\mu_2,\varepsilon_2,c_2)$.
Then the latter is thermal, and if the initial parameters are Markovian,
then so is $(\mu_3,\varepsilon_3,c_3)$.
\end{lemma}
\begin{proof}
Because the above composition rule on $\mathbb R^2\times\mathbb C$ is defined via \eqref{eq:app_E_G}, i.e.~$
\mu_3=\mu_1 (1 -\mu_2) + (1 - \varepsilon_1 \mu_1)\mu_2
$
and $\varepsilon_3\mu_3=\varepsilon_2\mu_2 (1 -\mu_1)  + \varepsilon_1\mu_1(1 - \varepsilon_2 \mu_2)$,
the assumption $\mu_1,\varepsilon_1,\mu_2,\varepsilon_2\in[0,1]$ directly implies $\mu_3,\varepsilon_3\geq 0$.
Next, using $1-\mu_3=(1-\mu_1)(1-\mu_2)+\mu_1\mu_2\varepsilon_1$ and $1-\varepsilon_3\mu_3=(1-\varepsilon_1\mu_1)(1-\varepsilon_2\mu_2)+\varepsilon_1\mu_1\mu_2$
we find $1-\mu_3\geq 0$, $\mu_3-\mu_3\varepsilon_3\geq 0$, and $|c_3|^2\leq (1-\varepsilon_3\mu_3)(1-\mu_3)$.
Together this shows that $(\mu_3,\varepsilon_3,c_3)$ is thermal.
%
Finally, the statement about Markovianity follows from $1-\mu_3-\varepsilon_3\mu_3=(1-\mu_1-\varepsilon_1\mu_1)(1-\mu_2-\varepsilon_2\mu_2)$.
\end{proof}

Hence the union of $\mathsf{MTO}(H_0,T)$ over all $T$ yields a semigroup which we denote by $\mathsf{MTO}(H_0)$.
%
%
One readily verifies that $\mathsf{MTO}(H_0)$ is weakly exponential, and
it turns out that it is even locally exponential:
to see this we need to find a neighbourhood basis of the identity which is exponential. Since $\mathsf{MTO}(H_0)$ decomposes into a stochastic part and a complex part (as shown above), it suffices to argue for each separately. Indeed the neighbourhood basis
$
\mathcal{U}_k = \{ (\mu,\varepsilon,c)\in\mathsf{MTO}(H_0)  : \mu,\varepsilon\mu,|\ln(c)|<1/k \}
$
is exponential, since the one parameter semigroups in the stochastic part are straight lines and the image of $c$ under $\ln$ forms a halfdisk.
\bibliographystyle{mystyle}
\bibliography{control21vJan20}             

\begin{thebibliography}{10}

\bibitem{Abiuso19}
Abiuso, P. and Giovannetti, V., \emph{Phys. Rev. A} \textbf{99} (2019), 052106.

\bibitem{Alhambra19}
Alhambra, {\'{A}}., Lostaglio, M., and Perry, C., \emph{{Quantum}} \textbf{3}
  (2019), 188.

\bibitem{book_HybridSytems96}
Alur, R., Henzinger, T., and Sontag, E., \emph{{Hybrid Systems III:
  Verification and Control}}, Lecture Notes in Computer Science (LNCS), Vol.
  1066, Springer, New York, 1996.

\bibitem{Ando89}
Ando, T., \emph{Lin. Alg. Appl.} \textbf{118} (1989), 163.

\bibitem{Aubin84}
Aubin, J.-P. and Cellina, A., \emph{{Differential Inclusions: Set-Valued Maps
  and Viability Theory}}, Springer, Berlin Heidelberg, 1984.

\bibitem{Bell90}
Bell, S., \emph{Bull. Amer. Math. Soc. (N.S.)} \textbf{22} (1990), 233.

\bibitem{BSH16}
Bergholm, V., Wilhelm, F., and Schulte-Herbr{\"u}ggen, T., \emph{{Arbitrary
  $n$-Qubit State Transfer Implemented by Coherent Control and Simplest
  Switchable Local Noise}}, 2016, \url{https://arxiv.org/abs/1605.06473}.

\bibitem{Bhatta20}
Bhattacharya, S., Bhattacharya, B., and Majumdar, A., \emph{J. Phys. A: Math.
  Theor.} \textbf{53} (2020), 335301.

\bibitem{QThermo2018}
Binder, F., Correa, L., Gogolin, C., Anders, J., and Adesso, G. (eds.),
  \emph{Thermodynamics in the Quantum Regime: Fundamental Aspects and New
  Directions}, Springer International, Cham, 2018.

\bibitem{Brandao15}
Brand{\~a}o, F., Horodecki, M., Ng, N., Oppenheim, J., and Wehner, S.,
  \emph{Proc. Natl. Acad. Sci. U.S.A.} \textbf{112} (2015), 3275.

\bibitem{Brandao13}
Brand{\~a}o, F., Horodecki, M., Oppenheim, J., Renes, J., and Spekkens, R.,
  \emph{Phys. Rev. Lett.} \textbf{111} (2013), 250404.

\bibitem{Bro72}
Brockett, R., \emph{SIAM J. Control} \textbf{10} (1972), 265.

\bibitem{Buscemi05}
Buscemi, F., Keyl, M., D'Ariano, G., Perinotti, P., and Werner, R., \emph{J.
  Math. Phys.} \textbf{46} (2005), 082109.

\bibitem{Bylicka16}
Bylicka, B., Tukiainen, M., Piilo, J., Chru{\'s}ci{\'n}ski, D., and Maniscalco,
  S., \emph{Sci. Rep.} \textbf{6} (2016), 27989.

\bibitem{Chakraborty22}
Chakraborty, S., Das, A., and Chru{\'s}ci{\'n}ski, A., \emph{Phys. Rev. E}
  \textbf{106} (2022), 064133.

\bibitem{Mart14}
Chen, Y., Neill, C., Roushan, P., Leung, N., Fang, M., Barends, R., Kelly, J.,
  Campbell, B., Chen, Z., Chiaro, B., Dunsworth, A., Jeffrey, E., Megrant, A.,
  Mutus, J., O’Malley, P., Quintana, C., Sank, D., Vainsencher, A., Wenner,
  J., White, T., Geller, M., Cleland, A., and Martinis, J., \emph{Phys. Rev.
  Lett} \textbf{113} (2014), 220502.

\bibitem{ChruPas17}
Chru{\'s}ci{\'n}ski, D. and Pascazio, S., \emph{Open Syst. Inf. Dyn.}
  \textbf{24} (2017), 1740001.

\bibitem{Colla22}
Colla, A. and Breuer, H., \emph{Phys. Rev. A} \textbf{105} (2022), 052216.

\bibitem{CK00}
Colonius, F. and Kliemann, W., \emph{{The Dynamics of Control}},
  Birkh{\"a}user, Boston, 2000.

\bibitem{Cwiklinski15}
\'Cwikli\'nski, P., Studzi\'nski, M., Horodecki, M., and Oppenheim, J.,
  \emph{Phys. Rev. Lett.} \textbf{115} (2015), 210403.

\bibitem{dAlessandroBook2022}
D'Alessandro, D., \emph{{Introduction to Quantum Control and Dynamics}},
  Chapman \& Hall CRC, Boca Raton, 2022, 2nd ed.

\bibitem{Dann22}
Dann, R., Megier, N., and Kosloff, R., \emph{Phys. Rev. Research} \textbf{4}
  (2022), 043075.

\bibitem{Davies76}
Davies, E., \emph{{Quantum Theory of Open Systems}}, Academic Press, London,
  1976.

\bibitem{Oliveira22}
de~Oliveira~Jr., A., Czartowski, J., {\.Z}yczkowski, K., and Korzekwa, K.,
  \emph{Phys. Rev. E} \textbf{106} (2022), 064109.

\bibitem{Ding21}
Ding, Y., Ding, F., and Hu, X., \emph{Phys. Rev. A} \textbf{103} (2021),
  052214.

\bibitem{DiHeGAMM08}
Dirr, G. and Helmke, U., \emph{GAMM-Mitteilungen} \textbf{31} (2008), 59.

\bibitem{DHKS08}
Dirr, G., Helmke, U., Kurniawan, I., and Schulte-Herbr{\"u}ggen, T., \emph{Rep.
  Math. Phys.} \textbf{64} (2009), 93.

\bibitem{CDC19}
Dirr, G., vom Ende, F., and Schulte-Herbr{\"u}ggen, T., \emph{Proc. IEEE Conf.
  Decision Control (IEEE-CDC)} \textbf{58} (2019), 2322.

\bibitem{Elliott09}
Elliott, D., \emph{{Bilinear Control Systems: Matrices in Action}}, Springer,
  London, 2009.

\bibitem{EL77}
Evans, D. and Lewis, J., \emph{{Dilations of Irreversible Evolutions in
  Algebraic Quantum Theory}}, Communications of the Dublin Institute for
  Advanced Studies, Vol.~24, Dublin, 1977.

\bibitem{Faist17}
Faist, P., Oppenheim, J., and Renner, R., \emph{New J. Phys.} \textbf{17}
  (2015), 1.

\bibitem{Roadmap2015}
Glaser, S., Boscain, U., Calarco, T., Koch, C., K{\"o}ckenberger, W., Kosloff,
  R., Kuprov, I., Luy, B., Schirmer, S., Schulte-Herb{\"u}ggen, T., Sugny, D.,
  and Wilhelm, F., \emph{Eur. Phys. J. D} \textbf{69} (2015), 279.

\bibitem{GFKVS78}
Gorini, V., Frigerio, A., Kossakowski, A., Verri, M., and Sudarshan, E.,
  \emph{Rep. Math. Phys.} \textbf{13} (1978), 149.

\bibitem{GK76}
Gorini, V. and Kossakowski, A., \emph{J. Math. Phys.} \textbf{17} (1976), 1298.

\bibitem{GKS76}
Gorini, V., Kossakowski, A., and Sudarshan, E., \emph{J. Math. Phys.}
  \textbf{17} (1976), 821.

\bibitem{Gour19}
Gour, G., \emph{Rev. Mod. Phys.} \textbf{91} (2019), 025001.

\bibitem{Gour22}
Gour, G., \emph{PRX Quantum} \textbf{3} (2022), 040323.

\bibitem{Gour15}
Gour, G., M{\"u}ller, M., Narasimhachar, V., Spekkens, R., and Halpern, N.,
  \emph{Phys. Rep.} \textbf{583} (2015), 1.

\bibitem{Gruenbaum03}
Gr{\"u}nbaum, B., \emph{{Convex Polytopes}}, Graduate Texts in Mathematics,
  Vol.~221, Springer, New York, 2003, 2 ed.

\bibitem{HHL89}
Hilgert, J., Hofmann, K.-H., and Lawson, J., \emph{{Lie Groups, Convex Cones,
  and Semigroups}}, Clarendon Press, Oxford, 1989.

\bibitem{HN12}
Hilgert, J. and Neeb, K., \emph{{Structure and Geometry of Lie Groups}},
  Springer Monographs in Mathematics, Springer, Berlin, 2012.

\bibitem{LNM1552}
Hilgert, J. and Neeb, K.-H., \emph{{Lie Semigroups and Their Applications}},
  Springer, Berlin, 1993.

\bibitem{Mart09}
Hofheinz, M., Wang, H., Ansmann, M., Bialczak, R., Lucero, E., Neeley, M.,
  O'Connell, A., Sank, D., Wenner, J., Martinis, J., and Cleland, A.,
  \emph{Nature} \textbf{459} (2009), 546.

\bibitem{HofRupp97div}
Hofmann, K.-H. and Ruppert, W., \emph{{Lie Groups and Subsemigroups with
  Surjective Exponential Function}}, Memoirs Amer. Math. Soc. 618, American
  Mathematical Society, Providence, 1997.

\bibitem{HJ1}
Horn, R. and Johnson, C., \emph{{Matrix Analysis}}, Cambridge University Press,
  Cambridge, 1987.

\bibitem{Horodecki13}
Horodecki, M. and Oppenheim, J., \emph{Nat. Commun.} \textbf{4} (2013), 2059.

\bibitem{Ding19}
Hu, X. and Ding, F., \emph{Phys. Rev. A} \textbf{99} (2019), 012104.

\bibitem{Janzing2000}
Janzing, D., Wocjan, P., Zeier, R., Geiss, R., and Beth, T., \emph{Int. J.
  Theor. Phys.} \textbf{39} (2000), 2717.

\bibitem{Jurdjevic97}
Jurdjevic, V., \emph{{Geometric Control Theory}}, Cambridge University Press,
  Cambridge, 1997.

\bibitem{JS72}
Jurdjevic, V. and Sussmann, H., \emph{J. Diff. Equat.} \textbf{12} (1972), 313.

\bibitem{Koch22}
Koch, C., Boscain, U., Calarco, T., Dirr, G., Filipp, S., Glaser, S., Kosloff,
  R., Montangero, S., Schulte-Herbrüggen, T., Sugny, D., and Wilhelm, F.,
  \emph{EPJ Quantum Technol.} \textbf{9} (2022), 19.

\bibitem{Korzekwa17}
Korzekwa, K., \emph{Phys. Rev. A} \textbf{95} (2017), 052318.

\bibitem{Kosloff13}
Kosloff, R., \emph{Entropy} \textbf{15} (2013), 2100.

\bibitem{Koss72}
Kossakowski, A., \emph{Bull. Acad. Pol. Sci., Ser. Sci. Math. Astron. Phys.}
  \textbf{20} (1972), 1021.

\bibitem{Koss72b}
Kossakowski, A., \emph{Rep. Math. Phys.} \textbf{3} (1972), 247.

\bibitem{Kuratowski66}
Kuratowski, K., \emph{{Topology I}}, Academic Press, New York, 1966.

\bibitem{book_impulsive89}
Lakshmikantham, V., Bainov, D., and Simeonov, P., \emph{{Theory of Impulsive
  Differential Equations}}, Series in Modern Applied Mathematics, Vol. 6, World
  Scientific, Singapore, 1989.

\bibitem{Lawson99}
Lawson, J., in: Ferreyra, G. (ed.), \emph{{Differential Geometry and Control}},
  207--221, American Mathematical Society, Providence, 1999.

\bibitem{Leela1991}
Leela, S., McRae, F., and Sivasundaram, S., \emph{J. Math. Anal. Appl.}
  \textbf{177} (1993), 24.

\bibitem{Lindblad76}
Lindblad, G., \emph{Commun. Math. Phys.} \textbf{48} (1976), 119.

\bibitem{Lostaglio19r}
Lostaglio, M., \emph{Rep. Prog. Phys.} \textbf{82} (2019), 114001.

\bibitem{Lostaglio18}
Lostaglio, M., Alhambra, {\'A}., and Perry, C., \emph{Quantum} \textbf{2}
  (2018), 52.

\bibitem{Lostaglio15_2}
Lostaglio, M., Jennings, D., and Rudolph, T., \emph{Nat. Commun.} \textbf{6}
  (2015), 6383.

\bibitem{LosKor22a}
Lostaglio, M. and Korzekwa, K., \emph{Phys. Rev. A} \textbf{106} (2022),
  012426.

\bibitem{MEDS23}
Malvetti, E., vom Ende, F., Dirr, G., and Schulte-Herbr{\"u}ggen, T.,
  \emph{{Reachability and Stabilizability for Markovian Quantum Systems with
  Fast Hamiltonian Control}}, in preparation, 2023.

\bibitem{MarshallOlkin}
Marshall, A., Olkin, I., and Arnold, B., \emph{{Inequalities: Theory of
  Majorization and Its Applications}}, Springer, New York, 2011, 2 ed.

\bibitem{Mazurek19}
Mazurek, P., \emph{Phys. Rev. A} \textbf{99} (2019), 042110.

\bibitem{Moise77}
Moise, E., \emph{{Geometric Topology in Dimensions 2 and 3}}, Graduate Texts in
  Mathematics, Vol.~47, Springer, New York, 1977.

\bibitem{Norris97}
Norris, J., \emph{{Markov Chains}}, Cambridge Series in Statistical and
  Probabilistic Mathematics, Cambridge University Press, Cambridge, 1997.

\bibitem{Parker96}
Parker, D. and Ram, P., \emph{{Greed and Majorization}}, Technical Report,
  Department of Computer Science, University of California, Los Angeles, 1996.

\bibitem{PG06}
P{\'e}rez-Garc{\'i}a, D., Wolf, M., Petz, D., and Ruskai, M., \emph{J. Math.
  Phys.} \textbf{47} (2006), 083506.

\bibitem{Ptas22}
Ptaszy\'nski, K., \emph{Phys. Rev. E} \textbf{106} (2022), 014114.

\bibitem{Renes14}
Renes, J., \emph{Eur. Phys. J. Plus} \textbf{129} (2014), 153.

\bibitem{rooney2018}
Rooney, P., Bloch, A., and Rangan, C., \emph{IEEE Trans. Automat. Contr.}
  \textbf{63} (2018), 672.

\bibitem{Rudin91}
Rudin, W., \emph{{Functional Analysis}}, McGraw--Hill, New York, 1991, 2 ed.

\bibitem{Sagawa19}
Sagawa, T., Faist, P., Kato, K., Matsumoto, K., Nagaoka, H., and Brand{\~a}o,
  F., \emph{J. Phys. A} \textbf{54} (2021), 495303.

\bibitem{Schrijver86}
Schrijver, A., \emph{{Theory of Linear and Integer Programming}}, Wiley
  Interscience, New York, 1986.

\bibitem{OSID17}
Schulte-Herbr{\"u}ggen, T., Dirr, G., and Zeier, R., \emph{Open Syst. Inf.
  Dyn.} \textbf{24} (2017), 1740019.

\bibitem{Shiraishi20}
Shiraishi, N., \emph{J. Phys. A} \textbf{53} (2020), 425301.

\bibitem{Smirnov02}
Smirnov, G., \emph{{Introduction to the Theory of Differential Inclusions}},
  Amer. Math. Soc., Providence, Rhode Island, 2002.

\bibitem{Spaventa22}
Spaventa, G., Huelga, S., and Plenio, M., \emph{Phys. Rev. A} \textbf{105}
  (2022), 012420.

\bibitem{Strogatz15}
Strogatz, S., \emph{{Nonlinear Dynamics and Chaos: With Applications to
  Physics, Biology, Chamistry and Engineering}}, CRC Press, Boca Raton, 2015,
  2nd ed.

\bibitem{Szczygielski13}
Szczygielski, K., Gelbwaser-Klimovsky, D., and Alicki, R., \emph{Phys. Rev. E}
  \textbf{87} (2013), 012120.

\bibitem{vE_PhD_2020}
vom Ende, F., \emph{{Reachability in Controlled Markovian Quantum Systems: An
  Operator-Theoretic Approach}}, Ph.D. thesis, TU Munich, 2020,
  \url{https://arxiv.org/pdf/2012.03496.pdf}.

\bibitem{vomEnde22thermal}
vom Ende, F., \emph{J. Math. Phys.} \textbf{63} (2022), 112202.

\bibitem{vE22_Stinespring}
vom Ende, F., \emph{Open Syst. Inf. Dyn.} \textbf{30} (2023), 2350003.

\bibitem{vomEnde20Dmaj}
vom Ende, F., \emph{Lin. Multilin. Alg.} \textbf{70} (2023), 4023.

\bibitem{vE_dirr_semigroups}
vom Ende, F. and Dirr, G., \emph{J. Math. Phys.} \textbf{60} (2019), 122702.

\bibitem{vomEnde22}
vom Ende, F. and Dirr, G., \emph{Lin. Alg. Appl.} \textbf{649} (2022), 152.

\bibitem{OSID19}
vom Ende, F., Dirr, G., Keyl, M., and Schulte-Herbr{\"u}ggen, T., \emph{Open
  Syst. Inf. Dyn.} \textbf{26} (2019), 1950014.

\bibitem{PolytopeDegen22}
vom Ende, F. and Malvetti, E., \emph{{The Thermomajorization Polytope and Its
  Degeneracies}}, 2022, \url{https://arxiv.org/abs/2212.04305}.

\bibitem{Wolf08a}
Wolf, M. and Cirac, J., \emph{Commun. Math. Phys.} \textbf{279} (2008), 147.

\bibitem{McDermott_TunDissip_2019}
Wong, C., Wilen, C., McDermott, R., and Vavilov, M., \emph{Quant. Sci.
  Technol.} \textbf{4} (2019), 025001.

\bibitem{Yamabe1950}
Yamabe, H., \emph{Osaka Math. J.} \textbf{2} (1950), 13.

\bibitem{Mart13}
Yin, Y., Chen, Y., Sank, D., O'Malley, P., White, T., Barends, R., Kelly, J.,
  Lucero, E., Mariantoni, M., Megrant, A., Neill, C., Vainsencher, A., Wenner,
  J., Korotkov, A., Cleland, A., and Martinis, J., \emph{Phys. Rev. Lett.}
  \textbf{110} (2013), 107001.

\bibitem{Yuan10}
Yuan, H., \emph{IEEE Trans. Automat. Contr.} \textbf{55} (2010), 955.

\bibitem{Zare88}
Zare, R., \emph{Angular Momentum}, Wiley Interscience, New York, 1988.

\end{thebibliography}
\end{document}